%% file: finalArxiv.tex
\documentclass[a4paper,UKenglish,cleveref,pdftex,thm-restate,numberwithinsect]{lipics-v2021}
\usepackage{stmaryrd}
\usepackage{bbding}
\newcommand{\cat}[1]{\ensuremath{\mathbf{#1}}}

\newcommand{\pbc}{\mathsf{Pb}}
\newcommand{\pbe}{\mathsf{ePb}}
\usepackage{amsthm,amsmath,amssymb,mathrsfs, dsfont}
\usepackage{hyperref, cleveref} %simboli matematici
\usepackage[all, cmtip]{xy}

\usepackage[utf8]{inputenc} % Input encoding - per caratteri particolari
\usepackage{graphicx} % Per includere immagini esterne
\usepackage[tickmarkheight=.5em,textwidth=\marginparwidth,textsize=small]{todonotes}
\usepackage{mathtools}
\usepackage{csquotes}

\usepackage{tikz}
\usepackage[draft]{tikzit}
\input{hypergraph.tikzdefs}
\input{hypergraph.tikzstyles}

\usetikzlibrary{decorations.markings}

\usepackage[usestackEOL]{stackengine}
\newcommand\functorop[1][l]{\csname#1functor\endcsname}
\newcommand\lfunctorop[3]{%
	\setbox0=\hbox{$#2$}%
	\kern\wd0%
	\ensurestackMath{\Centerstack[c]{#1\\ \mathllap{#2\;\,}\mathclap{\DownArrow}\\#3}}%
}		
\newcommand\rfunctorop[3]{%
	\setbox0=\hbox{$#2$}%
	\ensurestackMath{\Centerstack[c]{#1\\\mathclap{\UpArrow}\mathrlap{\,\;#2}\\#3}}%
	\kern\wd0%
}

\setstackgap{L}{1.3\normalbaselineskip}
\newcommand\UpArrow{\rotatebox[origin=c]{90}{$\longrightarrow$\,}}
\newcommand\DownArrow{\rotatebox[origin=c]{-90}{$\longrightarrow$\,}}
\newcommand\functor[1][l]{\csname#1functor\endcsname}
\newcommand\lfunctor[3]{%
	\setbox0=\hbox{$#2$}%
	\kern\wd0%
	\ensurestackMath{\Centerstack[c]{#1\\ \mathllap{#2\;\,}\mathclap{\DownArrow}\\#3}}%
}
\newcommand\rfunctor[3]{%
	\setbox0=\hbox{$#2$}%
	\ensurestackMath{\Centerstack[c]{#1\\\mathclap{\DownArrow}\mathrlap{\,\;#2}\\#3}}%
	\kern\wd0%
}
\newcommand\functormapsto{\mathrel{\ensurestackMath{\Centerstack[c]{\longmapsto\\ \\\longmapsto}}}}
\setstackgap{L}{1.3\normalbaselineskip}

\newcommand{\lgh}{\mathsf{lg}}
\newcommand{\eq}{\mathsf{eq}}
\newcommand{\quo}{\mathsf{Q}}

\DeclareMathAlphabet{\mymathbb}{U}{BOONDOX-ds}{m}{n}

\newcommand{\Set}{\mathbf{Set}}

\newcommand{\eg}[0]{\mathbf{EGG}}

\makeatletter
\def\@citecolor{blue}%
\def\@urlcolor{blue}%
\def\@linkcolor{blue}%

\def\orcidID#1{\smash{\href{http://orcid.org/#1}{\protect\raisebox{-1.25pt}{\protect\includegraphics{orcid_color.eps}}}}}
\makeatother

\def\B{\textbf {\textup{B}}}
\def\C{\textbf {\textup{C}}}
\def\D{\textbf {\textup{D}}}
\def\X{\textbf {\textup{X}}}
\def\Y{\textbf {\textup{Y}}}
\def\E{\textbf {\textup{E}}}
\def\T{\textbf {\textup{1}}}
\def\A{\textbf {\textup{A}}}

\newcommand{\catname}[1]{\textbf{\textup{#1}}}

\newcommand{\hyp}{\catname{Hyp}}
\newcommand{\hyps}{\catname{Hyp}_{\Sigma}}
\newcommand{\EqHyp}{\catname{EqHyp}} %equivalence hypergraphs

\newcommand{\EqTG}{\catname{EqTG}}
\newcommand{\EqTGs}{\catname{EqTG}_{\Sigma}}

\newcommand{\tg}[0]{\catname{TG}_{\Sigma}}

\newcommand{\pro}{\mathsf{prod}}

\newcommand{\cod}{\mathsf{cod}}

\newcommand{\dom}{\mathsf{dom}}

\newcommand{\egg}{\mathsf{e}\text{-}\catname{EqHyp}}

\newcommand{\ari}{\mathsf{ar}}

\newcommand{\comma}[2]{#1\hspace{1pt} {\downarrow}#2}

\newcommand{\commentato}[1]{ {} }

\usepackage{relsize}
\newcommand{\Sum}{\mathlarger{\sum}}

\newcommand{\mor}{\mathsf{Mor}}
\newcommand{\mon}{\mathsf{Mono}}
\newcommand{\reg}{\mathsf{Reg}}
\newcommand{\mto}{\rightarrowtail}
\newcommand{\eto}{\twoheadrightarrow}
\newcommand{\id}[1]{\mathsf{id}_{#1}}

\tikzset{->-/.style={decoration={
			markings,
			mark=at position #1 with {\arrow{>}}},
			postaction={decorate}
			}}

\title{EGGs are adhesive!}
\titlerunning{EGGs are adhesive!} %TODO optional, please use if title is longer than one line

\author{Roberto Biondo}
{Department of Computer Science, University of Pisa, Pisa, IT}
{r.biondo@studenti.unipi.it}{}{}

\author{Davide Castelnovo}
{Department of Computer Science, University of Pisa, Pisa, IT}
{castelnovod@gmail.com}
{https://orcid.org/0000-0002-5926-5615}{}

\author{Fabio Gadducci}
{Department of Computer Science, University of Pisa, Pisa, IT}
{fabio.gadducci@unipi.it}
{https://orcid.org/0000-0003-0690-3051}{}

\authorrunning{R.~Biondo, D.~Castelnovo, F.~Gadducci}

\Copyright{Roberto Biondo and Davide Castelnovo and Fabio Gadducci} 

 %TODO mandatory; please add comma-separated list of keywords

\funding{The research has been partially supported by the Italian MUR - Call PRIN 2022 - Project 20228KXFN2 ``Spatio-Temporal Enhancement of Neural nets for Deeply
Hierarchical Automatised Logic'' (STENDHAL) 
and by the University of Pisa - Call PRA 2022 - Project 2022\_99 ``Formal Methods for the Healthcare Domain based on Spatial Information'' (FM4HD).}

\nolinenumbers

\begin{document}

\maketitle

\begin{abstract}
The use of rewriting-based visual formalisms is on the rise. 
In the formal methods community, this is due also to the introduction of adhesive
categories, where most properties of classical approaches to graph transformation, 
such as those on parallelism and confluence, can be rephrased and proved in a general and 
uniform way.
E-graphs (EGGs) are a formalism for program optimisation 
via an efficient implementation of equality saturation. 
In short, EGGs can be  defined as (acyclic) term graphs with an additional notion of 
equivalence on nodes that is closed under the operators of the signature.
Instead of replacing the components of a program, the optimisation step 
is performed by adding new components and linking them to 
the existing ones via an equivalence relation, until an optimal program is reached.
This work describes EGGs via adhesive categories. 
Besides the benefits in itself of a formal presentation, which renders the 
properties of the data structure precise, the description of the addition of equivalent 
program components using standard graph transformation tools offers the advantages 
of the adhesive framework in modelling, for example, concurrent updates.
%
%Besides 
%the benefits \emph{per se} of a formal presentation, making precise the properties of 
%the data structure, describing the addition of equivalent program components 
%via the standard tools of graph transformation offers the advantages given by the 
%adhesive framework in modelling e.g. concurrent updates.
\end{abstract}

\section{Introduction}
The introduction of \emph{adhesive categories} marked a watershed moment for the algebraic approaches 
to the rewriting of graph-like structures~\cite{lack2005adhesive,ehrig2006fundamentals}.
Until then, key results of the approaches on e.g. parallelism and confluence had to be proven 
over and over again for each different formalism at hand, %let it be either xxx or xx, 
despite the obvious similarity of the procedure.
Adhesive categories provides such a disparate set of formalisms with a common abstract framework 
where many of these general results can be recast and uniformly proved once and for all.
 
\vspace{.1cm}
\noindent
\begin{minipage}[l]{.78\linewidth}In short, following the double-pushout (DPO) approach
to graph transformation~\cite{CorradiniMREHL97,ehrig2006fundamentals}, 
%~\cite{ehrig1999handbook,ehrig2006fundamentals,corradini1997algebraic}, 
a rule is given by two arrows $l: K \to L$ and $r: K \rightarrow R$
and its application requires a match $m: L \to G$: the rewriting step from $G$
to $H$ is then given by the diagram aside, whose halves are pushouts.
  \end{minipage}%
    \hfill
  \begin{minipage}[r]{.20\linewidth }
    \xymatrix@C=.5cm@R=.5cm{
      L \ar[d]_{m}
      & K \ar[r]^r \ar[l]_{l} \ar[d] & R \ar [d] \\
      G & C \ar[r] \ar[l]                    & H
    }
%    \xymatrix@C=.5cm@R=.5cm{
%      L \ar[d]_{m} \ar@{}[dr]|{(1)}
%      & I \ar[r]^r \ar@{>->}[l]_{l} \ar[d] \ar@{}[dr]|{(2)} & R \ar [d] \\
%      G & C \ar@{->}[r] \ar@{>->}[l]                    & H\\
 %     & J \ar@{->}[ul] \ar@{->}[u]|{k} \ar@{->}[ur]
 %   }
  \end{minipage}
\vspace{.1cm}

\noindent
Thus, $L$ and $R$ are the left- and right-hand side of the rule, respectively, while $K$ witnesses those parts that must 
be present for the rule to be executed, yet that are not affected by the rule itself.
Should a category be adhesive, and the arrows of the rules monomorphisms, the presence of a match ensures that
if the pushout complement $C \to G$ exists then it is unique, hence a rewriting step can be deterministically performed.
The theory of $\mathcal{M}$-adhesivity~\cite{azzi2019essence,heindel2009category} extends the core framework, ensuring that if the arrows 
of the rules are in a suitable class $\mathcal{M}$ of monomorphisms
%and the match is in $\mathcal{N}$, 
then the benefits of adhesivity
%the classical results from the theory of graph transformation 
can be recovered~\cite{ehrig2012,ehrig2014adhesive}. 
If only the left-hand side belongs to $\mathcal{M}$, the theory is still under development, as witnessed e.g.
by~\cite{BaldanC0G24}.
However, despite the elegance and effectiveness,
% and the immediate advantages, 
%to be adhesive, 
proving that a given category satisfies the conditions 
for being $\mathcal{M}$-adhesive can be a daunting task. For this reason, sufficient criteria have been provided for the core 
framework, e.g. that every elementary topos is adhesive \cite{lack2006toposes}, as well as for the extended one of
$\mathcal{M}$-adhesivity~\cite{CastelnovoGM24}.
For some structures such as \emph{hypergraphs with equivalence} in~\cite{concur2006}, the question 
of their $\mathcal{M}$-adhesivity has not yet been settled.

\emph{E-graphs} (shortly, EGGs) are an up-and-coming formalism for program optimisation and synthesis via a compact 
representation and efficient implementation of equality saturation. 
Albeit a classical data structure~\cite{DetlefsNS05}, EGGs received 
new impulse after the seminal~\cite{WillseyNWFTP21} and
developed a thriving community, as witnessed by the official website~\cite{eggs}.
The key idea of rewriting-based program optimisation is to perform the manipulation of a syntactical description 
of a program, replacing some of its components in such a way that the semantics is preserved while 
the computational costs of its actual execution are improved. Instead of directly removing sub-programs, EGGs just add the 
new components and link them to the older ones via the equivalence relation, until an optimal program is 
reached and extracted.

EGGs can be concisely defined as term graphs equipped with a notion of equivalence on nodes
that is closed under the operators of a signature~\cite[Section~4.2]{DetlefsNS05}.
%(i.e., graphs whose components are labelled with the operators of a signature)
In the presentation of term graphs via string diagrams~\cite{CastelnovoGM24}, EGGs are (hyper)trees 
whose edges are labelled by operators and with
the possibility of sharing subtrees, with an additional equivalence relation $\equiv$ on nodes that 
is closed under composition. In plain words and using a toy example:
if $a$ and $b$ are two constants such that $a \equiv b$, then $f(a) \equiv f(b)$ for any unary operator $f$.

%EGGs received a lot of attention in the implementation  side, as it is natural, since they have been introduced precisely as
%a fast way to recover equality saturation. 

Building on the criterion developed in~\cite{CastelnovoGM24}, this work proves that both hypergraphs with equivalence
and EGGs form an 
$\mathcal{M}$-adhesive category for a suitable choice of $\mathcal{M}$.
The advantages from this characterisation are two-fold. On the one side, 
we put the benefits \emph{per se} of a formal presentation, making precise the properties of the data structure. 
On the other side, describing the optimisation steps via the DPO approach
offers the tools for modelling their parallel and concurrent execution
and for proving their confluence and termination.

\emph{Synopsis}
The paper has the following structure. 
In Section~\ref{sec:ade} we briefly recall 
the theory of $\mathcal{M}$-adhesive categories
and of kernel pairs.
In Section~\ref{sec:hyper} we present the graphical structures of our interest, 
 (labelled) hypergraphs and term graphs, and we provide a
functorial characterisation, which allows for proving their adhesivity properties.
This is expanded in Section~\ref{hypereq} for proving the $\mathcal{M}$-adhesivity
of hypergraphs 
and term graphs with equivalence and in Section~\ref{eggs} of
their variants where equivalences are closed with respect to operator application,
thus subsuming EGGs.
In Section~\ref{rewriting} we put the machinery at work, showing how the optimisation steps
can be rephrased as the application of term graph rewriting rules.
Finally, in Section~\ref{conclusioni} we draw our conclusions, hint at future endeavours and offer some 
brief remarks on related works.
For the sake of space, the proofs appear in the appendices. 
%For the sake of completeness, 
%and in order to fix the notation, we prove all the results 
%recalled in the background section, besides those that are original 
%to our work.

\section{Facts about $\mathcal{M}$-adhesive categories and kernel pairs}\label{sec:ade}

%
%\begin{notation}
We open this background section by fixing some notation.
Given a category $\X$ we do not distinguish notationally between $\X$ and its class of objects, so
``$X\in \X$'' means that $X$ is an object of $\X$. We let $\mor(\X)$, $\mon(\X)$ and $\reg(\X)$ denote the class of all arrows, monos and regular monos of $\X$, respectively.  Given an object $X$, we  denote by $?_X$ the unique arrow from an initial object into $X$ and by $!_X$ that  unique arrow from $X$ into a terminal one. We will also use the notation $e\colon X\eto Y$ to denote that an arrow $e\colon X\to Y$ is a regular epi. 
%\end{notation}

\subsection{$\mathcal{M}$-adhesivity}\label{subsec:ade}
The key property of $\mathcal{M}$-adhesive categories is the \emph{Van Kampen condition}~\cite{brown1997van,johnstone2007quasitoposes,lack2005adhesive},
and for defining it we need some notions.
%and to define it we introduce some notions.  
Let  $\X$ be a category. A subclass $\mathcal{A}$ of $\mor(\X)$ is said to be

\vspace{-.25cm}
\parbox{11cm}{\begin{itemize}
	\item		\emph{stable under pushouts (pullbacks)} if for every pushout (pullback) square as the one aside, if $m \in \mathcal{A}$ ($n\in \mathcal{A}$) then $n \in \mathcal{A}$ ($m \in \mathcal{A}$);
		\item \emph{closed under composition} if $h, k\in \mathcal{A}$ implies $h\circ k\in \mathcal{A}$ whenever $h$ and $k$ are composable;
		\item \emph{closed under decomposition}, or \emph{left-cancellable}, if whenever $g$ and $g\circ f$ belong to $\mathcal{A}$, then $f\in \mathcal{A}$. 
\end{itemize}}\hfill
\parbox{1cm}{
\xymatrix{A \ar[r]^{f} \ar[d]_{m}& B \ar[d]^{n} \\ C \ar[r]_{g} & D}}

	\begin{definition}
	Let $\mathcal{A}\subseteq \mor(\X)$ be a class of arrows in a category $\X$ and consider the cube below on the right. 
	
\vspace{-.25cm}
\parbox{9.5cm}{We say that the bottom square is an \emph{$\mathcal{A}$-Van Kampen square} if
	\begin{enumerate}
		\item it is a pushout square;
		\item 	whenever the cube above has pullbacks as back and left faces and the vertical arrows belong to $\mathcal{A}$, then its top face is a pushout 
		if and only if the front and right faces are pullbacks.
	\end{enumerate}} \hfill
	\parbox{3cm}{
	\xymatrix@C=10pt@R=6pt{&A'\ar[dd]|\hole_(.65){a}\ar[rr]^{f'} \ar[dl]_{m'} && B' \ar[dd]^{b} \ar[dl]_{n'} \\ C'  \ar[dd]_{c}\ar[rr]^(.7){g'} & & D' \ar[dd]_(.3){d}\\&A\ar[rr]|\hole^(.65){f} \ar[dl]^{m} && B \ar[dl]^{n} \\C \ar[rr]_{g} & & D }
}

	Pushout squares that enjoy only the ``if'' half of item (2) above are called \emph{$\mathcal{A}$-stable}. A $\mor(\X)$-Van Kampen square is called  \emph{Van
		Kampen} and a $\mor(\X)$-stable square  \emph{stable}.
\end{definition}

We can now define $\mathcal{M}$-adhesive categories.

\begin{definition}[\cite{azzi2019essence,ehrig2012,ehrig2014adhesive,lack2005adhesive,heindel2009category}]
	Let $\X$ be a category and $\mathcal{M}$ a subclass of
	$\mon(\X)$  including  all isos, closed under composition, decomposition,  and stable under pullbacks and pushouts.  The category  $\X$ is said to be \emph{$\mathcal{M}$-adhesive} if
	\begin{enumerate}
		\item it has \emph{$\mathcal{M}$-pullbacks}, i.e.~pullbacks along arrows of $\mathcal{M}$;
		\item it has \emph{$\mathcal{M}$-pushouts}, i.e.~pushouts along arrows of $\mathcal{M}$;
		\item  $\mathcal{M}$-pushouts are $\mathcal{M}$-Van Kampen squares.
	\end{enumerate}
	A category $\X$ is said to be \emph{strictly $\mathcal{M}$-adhesive}
	if $\mathcal{M}$-pushouts are Van Kampen. We write $m\colon X \mto Y$ to denote that an arrow $m\colon X\to Y$ belongs to $\mathcal{M}$.
\end{definition}

\begin{remark}
	\label{rem:salva} Our notion of $\mathcal{M}$-adhesive category corresponds to what in \cite{ehrig2006fundamentals} is called \emph{weak adhesive HLR category}, while a strict $\mathcal{M}$-adhesive categories corresponds to \emph{adhesive HLR} ones. Finally, 	\emph{adhesivity} and \emph{quasiadhesivity} 
	\cite{lack2005adhesive,garner2012axioms} coincide with strict
	$\mon(\X) $-adhesivity and strict $\reg(\X)$-adhesivity,
	respectively. %On the other hand
\end{remark}

$\mathcal{M}$-adhesivity is well-behaved with respect to  the construction of (co-)slice and functor categories \cite{mac2013categories}, as well with respect to subcategories, as shown by the following properties, taken from~\cite[Thm.~4.15]{ehrig2006fundamentals}, \cite[Prop.~3.5]{lack2005adhesive} and  \cite[Thm.~2.12]{CastelnovoGM24}.

\begin{proposition}
	\label{thm:slice-functors}
	Let $\X$ be an (strict) $\mathcal{M}$-adhesive category. Then the following hold
	\begin{enumerate}
		\item if $\Y$ is an (strict) $\mathcal{N}$-adhesive category, $L\colon \Y\to \A$ a functor preserving $\mathcal{N}$-pushouts and $R\colon \X\to \A$ one preserving $\mathcal{M}$-pullbacks, then $\comma{L}{R}$ is (strictly) $\comma{\mathcal{N}}{\mathcal{M}}$-adhesive, where
		\[\comma{\mathcal{N}}{\mathcal{M}}:=\{(h,k) \in \mor(\comma{L}{R}) \mid h\in \mathcal{N}, k\in \mathcal{M}\}\]
		\item for every object $X$
		the categories $\X/X$  and $X/\X$ are, respectively, (strictly) $\mathcal{M}/X$-adhesive and (strictly) $X/\mathcal{M}$-adhesive, where
		\[\mathcal{M}/X:=\{m\in \mor(\X/X) \mid m\in
		\mathcal{M}\} \hspace{25pt} X/\mathcal{M}:=\{m\in \mor(X/\X) \mid m\in \mathcal{M}\}\]
		\item for every small category $\Y$, the category $\X^\Y$ of
		functors $\Y\to \X$ is (strictly) $\mathcal{M}^{\Y}$-adhesive, where
		$\mathcal{M}^{\Y}:=\{\eta \in \mor(\X^\Y) \mid \eta_Y \in
		\mathcal{M} \text{ for every } Y\in \Y\}$;
		\item if $\Y$ is a full subcategory of $\X$ closed under pullbacks and $\mathcal{M}$-pushouts, then $\Y$ is (strictly) $\mathcal{N}$-adhesive for every class of arrows $\mathcal{N}$ of $\Y$ contained in $\mathcal{M}$ that is stable under pullbacks and pushouts, contains all isos, and is closed under composition and decomposition.
	\end{enumerate} 
\end{proposition}

We briefly list some examples of $\mathcal{M}$-adhesive categories.

\begin{example}
	\label{ex:adhesive}
	$\cat{Set}$ is adhesive, and, more generally, every topos is
	adhesive~\cite{lack2006toposes}. By the closure properties above, every presheaf $[\cat{X},\cat{Set}]$ is adhesive, thus the category
	$\cat{Graph} = [ E \rightrightarrows V, \cat{Set}]$ is adhesive
	where $E \rightrightarrows {V}$ is the two objects category with two
	morphisms $s,t \colon{E} \to {V}$. Similarly, various
	categories of hypergraphs can be shown to be adhesive, such as term
	graphs and hierarchical graphs~\cite{CastelnovoGM24}. Note that the category $\cat{sGraphs}$ of simple graphs, 
	i.e.~graphs without parallel edges, is
	$\reg{(\cat{sGraphs})}$-adhesive~\cite{BehrHK23} but not
	quasiadhesive.
\end{example}

\iffalse 
\begin{remark}\label{rem:deco}
We can point out an important property of strict $\mathcal{M}$-adhesive categories with pullbacks.  Consider the solid part of the cube aside, whose bottom case is an $\mathcal{M}$-pushout.
	
	\parbox{9.5cm}{	Given an arrow $d\colon X\to D$, we can present the object $X$ has a pushout: indeed, consider the following cube, in which all the three vertical squares are pullbacks.}
	\parbox{3cm}{
	\xymatrix@C=15pt@R=9pt{&V\ar[dd]|\hole_(.65){a}\ar[rr]^{v} \ar@{>.>}[dl]_{u} && Y \ar[dd]^{b} \ar@{>->}[dl]_{y} \\ Z  \ar[dd]_{c}\ar[rr]^(.7){z} & & X \ar[dd]_(.3){d}\\&A\ar[rr]|\hole^(.65){f} \ar@{>->}[dl]^{m} && B \ar@{>->}[dl]^{n} \\C \ar[rr]_{g} & & D }}

Now, notice that, since the front square is a pullback, then the dotted arrow $u\colon V\to Z$ exists. Moreover, the usual composition and decomposition property of pullbacks \cite{mac2013categories} entails that the left face of the cube so obtained is a pullback too, proving that $u$ is in $\mathcal{M}$ and that the top square is a pushout.

	Clearly if the arrow $p\colon X\to D$ is in $\mathcal{M}$, we can omit the assumptions of strictness and the existence of all pullbacks.
\end{remark}
\fi

We can state some useful properties of $\mathcal{M}$-adhesive category
(see, for instance, \cite[Thm.~4.26]{ehrig2006fundamentals} or \cite[Fact 2.6]{azzi2019essence}). A proof is provided in \Cref{regmono-proof}.
\begin{restatable}{proposition}{regmono}\label{prop:regmono}
	Let $\X$ be an $\mathcal{M}$-adhesive category. Then the following hold
	\begin{enumerate}
		\item every $\mathcal{M}$-pushout square is also a pullback;
		\item every arrow in $\mathcal{M}$ is a regular mono.
	\end{enumerate}
\end{restatable}
%  [Proof in ]

\subsection{Some properties of kernel pairs and regular epimorphisms}

In this section we recall the definition and some properties of \emph{kernel pairs}.
%which will be useful in the next sections. 

\noindent 
\parbox{11cm}{
\begin{definition}
    A \emph{kernel pair} for an arrow $f\colon A \to B$ is an object $K_f$ together with two arrows $\pi^1_f, \pi^2_f\colon K_f \rightrightarrows A$, denoted as $(K_f, \pi^1_f, \pi^2_f)$, such that the square aside is a pullback.
\end{definition}}\hfill 
\parbox{2cm}{\xymatrix@R=15pt{K_f \ar[r]_{\pi^2_f} \ar[d]_{\pi^1_f}& A \ar[d]^f \\ A \ar[r]^{f}& B}}

\smallskip
\begin{remark}\label{prop:pairng_of_kernel_pairs_mono}
If $(K_f, \pi^1_f, \pi^2_f)$ is a a kernel pair for $f\colon X \to Y$ and a product of $X$ with itself exists, then the canonical arrow $\langle \pi^1_f, \pi^2_f\rangle \colon K_f \to X \times X$ is a mono.
\end{remark}

\begin{remark}\label{prop:kermono}
An arrow $m\colon M\to X$ is a mono if and only if it admits $(M, \id{M}, \id{M})$ as a kernel pair.
\end{remark}

Together with \Cref{lem:pb1}, the previous remarks allow us to prove the following result.

\begin{proposition}\label{cor:kermono}
	Let $f\colon X\to Y$ be an arrow and $m\colon Y\to Z$ a mono. If
	$(K_f, \pi_f^1, \pi_f^2)$ is a kernel pair for $f\colon X\to Y$, then it is also a kernel pair for $m\circ f$.
\end{proposition}

Regular epis are particular well-behaved with respect to their kernel pairs.

\begin{restatable}{proposition}{epic}\label{prop:reg_epi_coeq}
    Let $e\colon X \eto Y$ be a regular epi in a category $\X$ with a kernel pair $(K_e, \pi^1_e, \pi^2_e)$. Then, $e$ is the coequalizer of $\pi^1_e$ and $\pi^2_e$.
\end{restatable}
%[Proof in \Cref{epic-proof}]

\commentato{ 
\begin{restatable}{corollary}{natepi}\label{cor:reg_epi_components_reg_epi_nat_trans}
    Let $\X$ be a category with pullbacks and $\phi\colon F \to G$ a natural transformation between functors $F, G: \D \rightrightarrows \X$. If $\phi_d$ is a regular epi for every $d$, then $\phi$ is a regular epi.
\end{restatable}
%[Proof in \Cref{natepi-proof}]

From the previous result we deduce that the class of regular epis is closed under colimits.
% as shown by the next lemma

\begin{restatable}{lemma}{epicol}\label{lemma:nat_trans_reg_epi_canonical_arrow_reg_epi}
    Let $F,G\colon \D\rightrightarrows \X$ be two diagrams, and suppose that $\X$ has all colimits of shape $\D$. Let $(X, \{x_d\}_{d \in \D})$ and $(Y, \{y_d\}_{d\in D})$ be the colimits of $F$ and $G$, respectively.  If $\phi\colon  F \to G$ is a natural transformation whose components are regular epis, then the arrow induced by $\phi$ from $X$ to $Y$ is a regular epi.
\end{restatable}
%[Proof in \Cref{epicol-proof}]
}
We conclude this section exploring some properties of kernel pairs in an $\mathcal{M}$-adhesive category. 
The results below are simple, yet they appear to be original, and we give their proofs in \Cref{kpp-proof}.

\begin{restatable}{lemma}{kpp}\label{lemma:kern_pairs_pres_pullbacks}
	Let $f\colon X \to Y$ and $g\colon Z \to W$ be arrows admitting kernel pairs and suppose that the solid part of the four squares below is given. 
	If the leftmost square is commutative, then there is a unique arrow $k_h\colon K_f \to K_g$ making the other three commutative.
	\[\xymatrix@R=15pt{X \ar[r]^{h}\ar[d]_{f}& Z \ar[d]^{g} & K_f \ar@{.>}[r]^{k_h} \ar[d]_{\pi^1_f}& K_g \ar[d]^{\pi^1_g} & K_f  \ar@{.>}[r]^{k_h} \ar[d]_{\pi^2_f}& K_g \ar[d]^{\pi^2_g}& K_f  \ar@{.>}[r]^{k_h} \ar[d]_{(\pi^1_f,\pi^2_f)} & K_g \ar[d]^{(\pi^1_g, \pi^2_g)}\\ Y \ar[r]_t& W & X \ar[r]_{h} & Z & X\ar[r]_{h} &Z & X\times Z \ar[r]_{h\times h}& X\times Z}\]
	
		Moreover, the following hold
		\begin{enumerate}
			\item if $h$ is a mono then $k_h$ is a mono;
			\item if the leftmost square is a pullback then the central two are pullbacks;
			\item if $h$ is mono and the leftmost square is a pullback then the rightmost is a pullback.
		\end{enumerate} 
\end{restatable}
%[Proof in \Cref{kpp-proof}]

The previous result allows us to deduce the following lemma in an $\mathcal{M}$-adhesive context.

%\todo{A questo punto la proposizione seguente possiamo anche buttarla via se serve recuperare spazio}
\noindent
\parbox{7.5cm}{
\begin{restatable}{proposition}{mpo}\label{lem:mpo}
	Let $\X$ be a strict $\mathcal{M}$-adhesive category with all pullbacks, and suppose that in the cube aside the top face is an $\mathcal{M}$-pushout and all the vertical faces are pullbacks. Then the right square is a pushout.
\end{restatable}}\hfill 
\parbox{6cm}{\xymatrix@C=10pt@R=6pt{&A'\ar[dd]|\hole_(.65){a}\ar[rr]^{f'} \ar@{>->}[dl]_(.55){m'} && B' \ar[dd]^{b} \ar@{>->}[dl]_(.55){n'} & K_a\ar[rr]^{k_{f'}} \ar[dd]_{k_{m'}}&& K_b \ar[dd]^{k_{n'}} \\ C'  \ar[dd]_{c}\ar[rr]^(.7){g'} & & D' \ar[dd]_(.3){d}\\&A\ar[rr]|\hole^(.65){f} \ar[dl]^{m} && B \ar[dl]^{n}  & K_{c} \ar[rr]_{k_{g'}}&& K_d\\C \ar[rr]_{g} & & D }}

%[Proof in \Cref{mpo-proof}]
%
%\input{./fact_sys}

Focusing on $\Set$, we can prove a sharper result (see \Cref{proof:kerset} for the proof).

\noindent
\parbox{7.5cm}{\begin{restatable}{lemma}{mps}\label{prop:kerset}
	Suppose that in $\Set$ the commuting cube in the diagram on the left is given, whose top face is a pushout, the left and bottom faces are pullbacks,  and $n\colon B\mto D$ is an injection. 
	Then the following hold
	\begin{enumerate}
		\item the right face of the cube is a pullback;
		\item the right square, made by the kernel pairs of the vertical arrows, is a pushout.
	\end{enumerate}
\end{restatable}}\hfill 
\parbox{6cm}{\xymatrix@C=10pt@R=5pt{&A'\ar[dd]|\hole_(.65){a}\ar[rr]^{f'} \ar@{>->}[dl]_(.55){m'} && B' \ar[dd]^{b} \ar@{>->}[dl]_(.55){n'} & K_a\ar[rr]^{k_{f'}} \ar[dd]_{k_{m'}}&& K_b \ar[dd]^{k_{n'}} \\ C'  \ar[dd]_{c}\ar[rr]^(.7){g'} & & D' \ar[dd]_(.3){d}\\&A\ar[rr]|\hole^(.65){f} \ar@{>->}[dl]^{m} && B \ar@{>->}[dl]^{n}  & K_{c} \ar[rr]_{k_{g'}}&& K_d\\C \ar@{>->}[rr]_{g} & & D }}

\section{Hypergraphical structures}\label{sec:hyper}

In this section we briefly recall the notion of \emph{hypergraph}. In order to do so, a pivotal role is played by the \emph{Kleene star monad} $(-)^\star\colon \Set\to \Set$, also known as 
\emph{list monad},
sending a set to the free monoid on it \cite{sakarovitch2009elements,Wadler95}.
We recall some of its proprieties.

\begin{proposition}\label{prop:fact}
	Let $X$ be a set and $n\in \mathbb{N}$. Then the following facts hold
	\begin{enumerate}
		\item there are arrows $v_{n}\colon X^n\to X^\star $ such that $(X^\star, \{v_{n}\}_{n\in \mathbb{N}})$ is a coproduct;
		\item for every arrow $f\colon X\to Y$, $f^\star\colon X^\star \to Y^\star$ is the coproduct of the family $\{f^n\}_{n\in \mathbb{N}}$;
		\item $(-)^\star$ preserves all \emph{connected limits} \cite{carboni1995connected}, in particular it preserves pullbacks and equalizers.
	\end{enumerate}
\end{proposition}

\begin{remark}\label{rem:mono}
	Preservation of pullbacks implies that $(-)^\star$ sends monos to monos.
\end{remark}

\begin{remark}\label[remark]{rem:length}Notice that $1^\star$ can be canonically identified with $\mathbb{N}$, thus for every set $X$ the arrow $!_\X\colon X\to 1$ induces a \emph{length function} $\lgh_{X}\colon X^\star \to \mathbb{N}$, which sends a word to its length.
\end{remark}

\subsection{The category of hypergraphs}

We open this section with the definition of hypergraphs and we show how to label them with an algebraic signature.  

\noindent 
\parbox{10.8cm}{\begin{definition}An \emph{hypergraph} is a 4-uple $\mathcal{G}:=(E_\mathcal{G}, V_\mathcal{G}, s_\mathcal{G}, t_\mathcal{G})$ made by two sets $E_\mathcal{G}$ and $V_\mathcal{G}$, called respectively the sets of \emph{hyperedges} and \emph{nodes}, plus a pair of \emph{source} and \emph{target} arrows  $s_\mathcal{G}, t_\mathcal{G}\colon E_\mathcal{G}\rightrightarrows V_\mathcal{G}^\star$. 
		
\hspace{10pt} A \emph{hypergraph morphism} $(E_\mathcal{G}, V_\mathcal{G}, s_\mathcal{G}, t_\mathcal{G})\to (E_\mathcal{H}, V_\mathcal{H}, s_\mathcal{H}, t_\mathcal{H})$ is a pair $(h,k)$ of functions $h\colon E_\mathcal{G}\to E_\mathcal{H}$, $k\colon V_\mathcal{G}\to V_\mathcal{H}$ such that the diagrams on the right are commutative.

\hspace{10pt}	We define $\hyp$ to be the resulting category.
\end{definition}}\hfill\parbox{3cm}{\xymatrix@R=15pt{E_{\mathcal{G}} \ar[d]_{h} \ar[r]_{s_{\mathcal{G}}}& V^\star_{\mathcal{G}}  \ar[d]^{k^\star}\\E_{\mathcal{G}} \ar[r]^{s_{\mathcal{H}}} & V^\star_{\mathcal{H}}} \hspace{1pt}\\ \xymatrix@R=15pt{E_{\mathcal{G}} \ar[d]_{h} \ar[r]_{t_{\mathcal{G}}}& V^\star_{\mathcal{G}}  \ar[d]^{k^\star}\\E_{\mathcal{G}} \ar[r]^{t_{\mathcal{H}}} & V^\star_{\mathcal{H}}} }

Let $\pro^\star$ be the functor $\Set\to \Set$ sending $X$ to the product $X^\star\times X^\star$.  We can use it to present $\hyp$ as a comma category.

\begin{proposition}\label[proposition]{prop:com}
	$\hyp$ is isomorphic to $\comma{\id{\Set}}{\pro^\star}$
\end{proposition}

Note that by hypothesis $(-)^\star$ preserves pullbacks, while $\pro$ is continuous by definition, 
hence by \Cref{prop:com} and \Cref{cor:mono} we can deduce the following result.

\begin{corollary}\label[corollary]{cor:monhyper}
	A morphism $(h,k)$ is a mono in $\hyp$ if and only if both its components are injective functions.
\end{corollary}

Applying \Cref{thm:slice-functors} we also get the next corollary (cfr.~\cite[Fact 4.17]{ehrig2006fundamentals}).

\begin{corollary}\label[corollary]{prop:hypadh}
	 $\hyp$ is an adhesive category.
\end{corollary}

\Cref{prop:left,prop:com} allow us to deduce immediately the following.

\begin{proposition}\label{cor:left}  The forgetful functor $U_{\hyp}$ which sends a hypergraph $\mathcal{G}$ to its set of nodes has a left adjoint $\Delta_{\hyp}$.
\end{proposition}

\begin{example}Since the initial object of $\catname{Set}$ is the empty set,  $\Delta_{\hyp}(X)$ is the hypergraph which has $X$ as set of nodes, $\emptyset$ as set of hyperedges, and $?_X$ as source and target function.
\end{example}

%\todo{FG ne lascerei uno, ripreso magari dalla intro e presentato come string diagram}

\begin{example}\label[example]{exa_2}
We represent hypergraphs visually: dots denote nodes and boxes hyperedges. Should we be interested in their identity, we put a name near the corresponding dot or box. 
Sources and targets are represented by lines between dots and squares: the lines from the sources of a hyperedge comes from the left of the box, while the lines to the targets 
exit from the right of the box.
Let us look at the picture below. It represent a hypergraph $\mathcal{G}$ with sets $V_\mathcal{G} = \{v_1, v_2, v_3, v_4\}$ and $E_\mathcal{G} = \{h_1, h_2, h_3, h_4\}$ of nodes 
and edges, respectively, such that  $h_1, h_2$ have no source and $h_3, h_4$ a pair of nodes
	
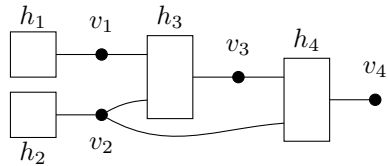
\begin{center}
        \begin{tikzpicture}
                \begin{pgfonlayer}{nodelayer}
                        \node[style=none] (lab1) at (-1.8, 0.8) {$h_1$};
                        \node[style=small box] (h1) at (-1.8, 0.3) {};
                        \node[style=small box] (2) at (-1.8, -0.5) {};
                        \node[style=none] (lab2) at (-1.8, -1.0) {$h_2$};
                        \node[style=none] (h1b) at (-1.8, 0.3) {};
                        \node[style=none] (th2) at (0, 0.3) {};
                        \node[style=none] (bh2) at (0, -0.3) {};
                        \node[style=none] (h2b) at (0, 0) {};
                        \node[style=none] (th3) at (1.8, 0){};
                        \node[style=none] (bh3) at (1.8, -0.6){};
                        \node[style=none] (h3c) at (1.8, -0.3){};
                        \node[style=none] (2c) at (-1.8, -0.5){};
                        \node[style=none] (lab3) at (0, 0.8) {$h_3$};
                        \node[style=medium box] (h2) at (0, 0) {};
                        \node[style=none] (lab4) at (1.8, 0.5) {$h_4$};
                        \node[style=medium box] (h3) at (1.8, -0.3) {};
                        \node[style=none](labv2) at (-0.9, -0.9) {$v_2$};
                        \node[style=node] (v2) at (-0.9, -0.5) {};
                        \node[style=none] (labv1) at (-0.9, 0.7) {$v_1$};
                        \node[style=node] (v1) at (-0.9, 0.3) {};
                        \node[style=none] (labv3) at (0.9, 0.4) {$v_3$};
                        \node[style=node] (v3) at (0.9, 0) {};
                        \node[style=none] (labv4) at (2.7, 0.1) {$v_4$};
                        \node[style=node] (v4) at (2.7, -0.3) {};
                \end{pgfonlayer}
                \begin{pgfonlayer}{edgelayer}
                        \draw (h1b.center) to (th2.center);
                        \draw (2c.center) to (v2);
                        \draw[in=180, out=30] (v2) to (bh2.center);
                        \draw (h2b.center) to (th3.center);
                        \draw (v4) to (h3c.center);
                        %\draw (h3c) to (v4);
                        \draw[in=180, out=-30](v2) to (bh3);
                \end{pgfonlayer}
        \end{tikzpicture}
	\end{center}
\end{example}

\commentato{
\begin{example}\label[example]{exa_2}
	Let $V_\mathcal{G} = \{v_1, v_2, v_3, v_4, v_5\}$, and $E_\mathcal{G} = \{h_1, h_2, h_3, h_4\}$. Then we define:
	\[\begin{matrix}
		s_{\mathcal{G}}(h_1)\colon 0\to V_{\mathcal{G}} & s_{\mathcal{G}}(h_1)=?_{V_\mathcal{G}} &&
		s_{\mathcal{G}}(h_2)\colon 0\to V_{\mathcal{G}} & s_{\mathcal{G}}(h_2)=?_{V_\mathcal{G}}\\
		s_{\mathcal{G}}(h_3)\colon 2\to V_{\mathcal{G}} & \begin{matrix} 
					0 \mapsto v_1\\
					1\mapsto v_2
				\end{matrix}&& s_{\mathcal{G}}(h_4)\colon 2 \to V_\mathcal{G} & \begin{matrix}
					0 \mapsto v_3\\
					1\mapsto v_2
				\end{matrix}\\
		t_{\mathcal{G}}(h_1)\colon 1\to V_{\mathcal{G}} & 0 \mapsto v_1 &&
		t_{\mathcal{G}}(h_2)\colon 1\to V_{\mathcal{G}} & 0\mapsto v_3\\
		t_{\mathcal{G}}(h_3)\colon 1\to V_{\mathcal{G}} & 1\mapsto v_3 &&
		t_\mathcal{G}(h_4)\colon 1 \to V_\mathcal{G} & 0 \mapsto v_4
	\end{matrix}\]

	Now, we can depict $\mathcal{G}$ as

	\begin{center}
		        \begin{tikzpicture}
                \begin{pgfonlayer}{nodelayer}
                        \node[style=none] (lab1) at (-1.8, 0.8) {$h_1$};
                        \node[style=small box] (h1) at (-1.8, 0.3) {};
                        \node[style=small box] (2) at (-1.8, -0.5) {};
                        \node[style=none] (lab2) at (-1.8, -1.0) {$h_2$};
                        \node[style=none] (h1b) at (-1.8, 0.3) {};
                        \node[style=none] (th2) at (-0.3, 0.3) {};
                        \node[style=none] (bh2) at (-0.3, -0.3) {};
                        \node[style=none] (h2b) at (0, 0) {};
                        \node[style=none] (th3) at (1.5, 0){};
                        \node[style=none] (bh3) at (1.5, -0.6){};
                        \node[style=none] (h3c) at (1.8, -0.3){};
                        \node[style=none] (2c) at (-1.8, -0.5){};
                        \node[style=none] (lab3) at (0, 0.8) {$h_3$};
                        \node[style=medium box] (h2) at (0, 0) {};
                        \node[style=none] (lab4) at (1.8, 0.5) {$h_4$};
                        \node[style=medium box] (h3) at (1.8, -0.3) {};
                        \node[style=none](labv2) at (-0.9, -0.9) {$v_2$};
                        \node[style=node] (v2) at (-0.9, -0.5) {};
                        \node[style=none] (labv1) at (-0.9, 0.7) {$v_1$};
                        \node[style=node] (v1) at (-0.9, 0.3) {};
                        \node[style=none] (labv3) at (0.9, 0.4) {$v_3$};
                        \node[style=node] (v3) at (0.9, 0) {};
                        \node[style=none] (labv4) at (2.7, 0.1) {$v_4$};
                        \node[style=node] (v4) at (2.7, -0.3) {};
                \end{pgfonlayer}
                \begin{pgfonlayer}{edgelayer}
                        \draw[style=diredge] (h1b.center) to (v1);
                        \draw[style=diredge] (v1) to (th2.center);
                        \draw[style=diredge] (2c.center) to (v2);
                        \draw[style=diredge, in=180, out=30] (v2) to (bh2.center);
                        \draw[style=diredge] (v3) to (th3.center);
                        \draw[style=diredge] (h3c.center) to (v4);
                        \draw[style=diredge] (h2b.center) to (v3);
                        \draw[style=diredge,in=180, out=-30](v2) to (bh3);
                \end{pgfonlayer}
        \end{tikzpicture}
\end{center}

\end{example}
}

\begin{remark}\label{rem:functor}
	It is worth to point out, as first noted in \cite{bonchi2022string}, that $\hyp$ is equivalent to a category of presheaves. 
	Indeed, consider the category $\catname{H}$ in which the set of objects is given by $ (\mathbb{N}\times \mathbb{N}) \cup \{\bullet\}$ and arrows are given by the identities $\id{k,l}$, $\id{\bullet}$ and exactly $k+l$ arrows $f_i\colon (k,l)\rightarrow \bullet$, where $i$ ranges from $0$ to $k+l-1$. 
	The functors $\catname{H}\to \Set$ corresponds exactly to hypergraphs: nodes correspond to the image of $\bullet$ while the set of hyperedges with source of length $k$ and target of length $l$ corresponds to the image of $(k,l)$ (see \cite{CastelnovoGM24}).
\end{remark}

In particular, \Cref{rem:functor} entails the following result. 

\begin{proposition}\label{prop:cocomp}
	$\hyp$ has all limits and colimits.
\end{proposition}

\subsubsection{Labelling hypergraph with an algebraic signature}\label{sssect:hyp_alg_sign}

Our interest for hypergraphs stems from their use as a graphical representation of algebraic terms. We thus need a way to label hyperedges with symbols taken from a signature.

\noindent
\begin{minipage}[l]{.73\linewidth}
\begin{definition}
An \emph{algebraic signature} $\Sigma$ is a pair $(O_\Sigma, \ari_\Sigma)$ given by a \emph{set of operations} $O_\Sigma$ and an \emph{arity function} $\ari_\Sigma\colon O_\Sigma \to \mathbb{N}$. 
We define the \emph{hypergraph $\mathcal{G}_\Sigma$ associated with $\Sigma$} taking $O_\Sigma$ as set of hyperedges, $1$ as set of nodes, so that $1^\star$ is $\mathbb{N}$, $\ari_\Sigma$ as the source function and $\gamma_1$, which always picks the element $1$, as target function. The category $\hyps$ of \emph{algebraically labelled hypergraphs} is the slice category $\hyp/\mathcal{G}^\Sigma$.
\end{definition}

\iffalse
\todo{Decidere se tenere questo esempio}

\begin{example}\label[example]{exa_3} Let $\Sigma=(O_\Sigma, \ari_\Sigma)$ be an algebraic signature in $\Set$. This simply amount to a set of \emph{operations} with an associated natural number, called \emph{arity}. 	For instance let $\Sigma_G$ be the signature of groups, then $\mathcal{G}^{\Sigma_G}$ can be depicted as the picture aside.
	\begin{center}
		\begin{tikzpicture}
			\node[circle,fill=black,inner sep=0pt,minimum size=6pt,label=above:{$v$}] (V) at (0,0) {};
			\node(E)at(-2, 0.4){$e$};
			\node(M)at(0, 2.15){$\cdot$};
			\node(I)at(2, 0.5){$(-)^{-1}$};
			\draw[->-=.5](-1.75,0)--(V)node[pos=0.5, above,font=\fontsize{7}{0}\selectfont]{$1$};
			\draw[->-=.5](V)..controls(-0.5,0.5)and(-0.8,1)..(-0.25,1.6)node[pos=0.5, right,font=\fontsize{7}{0}\selectfont]{$2$};
			\draw[->-=.5](V)..controls(-1,0.6)and(-1,1.1)..(-0.25,1.9)node[pos=0.5, left,font=\fontsize{7}{0}\selectfont]{$1$};
			\draw[->-=.5](0.25,1.75)..controls(0.8,0.8)and(0.5,0.5)..(V)node[pos=0.5, right,font=\fontsize{7}{0}\selectfont]{$1$};
			\draw[->-=.5](V)--(1.75,0)node[pos=0.5, above,font=\fontsize{7}{0}\selectfont]{$1$};
			\draw[->-=.5](2.25,0)..controls(3.5,0)and(2.5,-2)..(V)node[pos=0.5, below,font=\fontsize{7}{0}\selectfont]{$2$};
			\draw[rounded corners] (-2.25, -0.25) rectangle (-1.75, 0.25) {};
			\draw[rounded corners] (-0.25, 1.5) rectangle (0.25, 2) {};
			\draw[rounded corners] (2.25, -0.25) rectangle (1.75, 0.25) {};
		\end{tikzpicture}
	\end{center}
\end{example}\fi

\begin{example}\label[example]{exa_3}
	Let $\Sigma = (O_\Sigma, \ari_\Sigma)$ be the signature with $O_\Sigma =  \mathbb{N} \uplus A \uplus \{*, /\}$, 
	where $n$ stands for any natural number and $a$ for any element in $A$, both sets of constants, and
	%so that $\ari_\Sigma(n)=\ari_\Sigma(a) = 0$, 
	$\ari_\Sigma(*) = \ari_\Sigma(/) = 2$. Then the hypergraph $\mathcal{G}^\Sigma$ is depicted as the picture aside.
\end{example}
\end{minipage}
%\begin{center}
\hfill
\begin{minipage}[r]{.22\linewidth}
\begin{tikzpicture}
        \begin{pgfonlayer}{nodelayer}
                \node[style=node] (v) at (0, 0){};
                \node[style=none] (vlab) at (0, -0.4){$v$};
                
                \node[style=small box] (n) at (-1.8, 0.4) {};
                \node[style=none] (nlab) at (-1.8, 0.9) {$n$};
                \node[style=none] (nattach) at (-1.6, 0.4) {};

                \node[style=small box] (const) at (-1.8, -0.4) {};
                \node[style=none] (constlab) at (-1.8, -0.9){$a$};
                \node[style=none] (constattach) at (-1.6, -0.4){};

                \node[style=medium box] (times) at (0, 1.8) {};
                \node[style=none] (timeslab) at (0, 2.6) {$*$};
                \node[style=none] (timesin1) at (-0.3, 2.1) {};
                \node[style=none] (timesin2) at (-0.3, 1.5) {};
                \node[style=none] (timesout) at (0.2, 1.8){};

                \node[style=medium box] (div) at (0, -1.8) {};
                \node[style=none] (divlab) at (0, -2.6) {$/$};
                \node[style=none] (divin1) at (-0.3, -2.1) {};
                \node[style=none] (divin2) at (-0.3, -1.5) {};
                \node[style=none] (divout) at (0.2, -1.8){};
        \end{pgfonlayer}

        \begin{pgfonlayer}{edgelayer}
                \draw[style=diredge, in=180, out= 112.5] (v) to (timesin2.center); 
                \draw[style=diredge,out=135, in= 180] (v) to (timesin1.center);
                \draw[style=diredge,in=165, out=0] (nattach.center) to (v);
                
                \draw[style=diredge,out=0, in=-165] (constattach.center) to (v);
                \draw[style=diredge,in=180, out= -112.5] (v) to (divin2.center); 
                \draw[style=diredge,out=-135, in= 180] (v) to (divin1.center);

                \draw[style=diredge, in=30, out=0] (timesout.center) to (v);
                \draw[style=diredge, in=-30, out=0] (divout.center) to (v);
        \end{pgfonlayer}
\end{tikzpicture}
\end{minipage}

\Cref{cor:mono} and \Cref{thm:slice-functors} give us immediately an adhesivity result for $\hyp_{\Sigma}$ and a characterisation of monos in it.

\begin{proposition}\label[proposition]{prop:mono} Let $\Sigma$ be an algebraic signature. Then the following hold
	\begin{enumerate}
		\item an arrow $(h,k)$ in $\hyp_{\Sigma}$ is a mono if and only if $h$ and $k$ are injective;
		\item $\hyp_{\Sigma}$ is an adhesive category. 
	\end{enumerate}
\end{proposition}

\noindent 
\parbox{10.8cm}{\begin{remark}\label[remark]{rem:label}	
Let $\mathcal{H}=(E, V, s, t)$ be a hypergraph, by definition we know that $U_{\hyp}(\mathcal{G}^{\Sigma})$ is the terminal object $1$, so an arrow $\mathcal{H}\rightarrow \mathcal{G}^{\Sigma}$, is determined by a function $h\colon E_\mathcal{H}\to O_\Sigma$  making the two squares on the right commutative (cfr.~\Cref{rem:length}).

\hspace{10pt}Now, consider a coprojection $v_n\colon V^n_\mathcal{H}\to  V^\star_{\mathcal{H}}$.  By the second diagram above entails that $t_{\mathcal{H}}$ factors via the inclusion $v_1\colon V_{\mathcal{H}}\to V_\mathcal{H}^{\star}$ of words of length $1$, i.e.~all hyperedges must have a single target vertex, that is $t_{\mathcal{H}}=v_1\circ \tau_{\mathcal{H}}$ for some $\tau_{\mathcal{H}}\colon E_{\mathcal{H}}\to V_{\mathcal{H}}$.
\end{remark}} \hfill \parbox{3cm}{\xymatrix@R=15pt{E_{\mathcal{H}} \ar[d]_{h}\ar[r]_{s_{\mathcal{H}}} & V^\star_{\mathcal{H}} \ar[d]^{\lgh_{V_{\mathcal{H}}}} \\  O_\Sigma \ar[r]^{\ari_\Sigma}& \mathbb{N}} \hspace{1pt}\\ \xymatrix@R=15pt{E_{\mathcal{H}} \ar[d]_{h}\ar[r]_{t_{\mathcal{H}}} & V^\star_{\mathcal{H}} \ar[d]^{\lgh_{V_{\mathcal{H}}}} \\  O_\Sigma \ar[r]^{\gamma_1}& \mathbb{N}}}

$\hyp_{\Sigma}$ has a forgetful functor $U_{\Sigma}\colon \hyp_{\Sigma}\to \X$ which sends $(h,k)\colon \mathcal{H}\to \mathcal{G}^{\Sigma}$ to $U_{\hyp}(\mathcal{H}$). Now, since $U_{\hyp}(\mathcal{G}^{\Sigma})=1$ 
for every set $X$ we can define $\Delta_{\Sigma}(X)\colon \Delta_{\hyp}(X)\to \mathcal{G}^{\Sigma}$  as $(?_{O_\Sigma}, !_{X})$. It is straightforward to see that in this way we get a left adjoint to $U_\Sigma$.

\begin{proposition} $U_\Sigma$
	has a left adjoint $\Delta_\Sigma$.
\end{proposition}

We extend our graphical notation of hypergraphs to labeled ones putting the label of an hyperedge $h$ inside its corresponding square.

\begin{example}\label[example]{lab_2}
	Consider again $\Sigma$ the signature of \Cref{exa_3}, then the hypergraph $\mathcal{G}$ of \Cref{exa_2} can be labeled by a morphism
	$(l, !_{V_\mathcal{G}}): \mathcal{G} \to \mathcal{G}^\Sigma$ that is characterised by the image of the edges. 
	If $l(h_1) = a$, $l(h_2) = 2$, $l(h_3) = *$, and $l(h_4) = /$, we represent it visually 
	by putting the labels of the edges in  $\mathcal{G}^\Sigma$ inside the boxes representing the edges of $\mathcal{G}$.
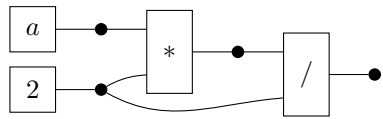
		\begin{center}
		\begin{tikzpicture}
			\begin{pgfonlayer}{nodelayer}
				\node[style=none] (lab1) at (-1.8, 0.8) {};
				\node[style=small box] (h1) at (-1.8, 0.3) {$a$};
				\node[style=small box] (2) at (-1.8, -0.5) {$2$};
				\node[style=none] (lab2) at (-1.8, -1.0) {};
				\node[style=none] (h1b) at (-1.8, 0.3) {};
				\node[style=none] (th2) at (0, 0.3) {};
				\node[style=none] (bh2) at (0, -0.3) {};
				\node[style=none] (h2b) at (0, 0) {};
				\node[style=none] (th3) at (1.8, 0){};
				\node[style=none] (bh3) at (1.8, -0.6){};
				\node[style=none] (h3c) at (1.8, -0.3){};
				\node[style=none] (2c) at (-1.8, -0.5){};
				\node[style=none] (lab3) at (0, 0.8) {};
				\node[style=medium box] (h2) at (0, 0) {$*$};
				\node[style=none] (lab4) at (1.8, 0.5) {};
				\node[style=medium box] (h3) at (1.8, -0.3) {$/$};
				\node[style=none](labv2) at (-0.9, -0.9) {};
				\node[style=node] (v2) at (-0.9, -0.5) {};
				\node[style=none] (labv1) at (-0.9, 0.7) {};
				\node[style=node] (v1) at (-0.9, 0.3) {};
				\node[style=none] (labv3) at (0.9, 0.4) {};
				\node[style=node] (v3) at (0.9, 0) {};
				\node[style=none] (labv4) at (2.7, 0.1) {};
				\node[style=node] (v4) at (2.7, -0.3) {};
			\end{pgfonlayer}
			\begin{pgfonlayer}{edgelayer}
				\draw (h1b.center) to (th2.center);
				\draw (2c.center) to (v2);
				\draw[in=180, out=30] (v2) to (bh2.center);
				\draw (h2b.center) to (th3.center);
				\draw (v4) to (h3c.center);
				%\draw (h3c) to (v4);
				\draw[in=180, out=-30](v2) to (bh3);
			\end{pgfonlayer}
        	\end{tikzpicture}
	\end{center}
\end{example}

\subsection{Term Graphs}
Term graphs have been adopted as a convenient way to represent terms over a signature with an explicit sharing of sub-terms,
and as such have been a convenient tool for an efficient implementation of term rewriting~\cite{Plu:TGR-ENTCS}.
We elaborate here on the presentation given in~\cite{CastelnovoGM24}.

\begin{definition}\label[definition]{def:tg}\index{term graph}
	Given an algebraic signature $\Sigma$, we say that a labelled hypergraph $(l, !_{V_\mathcal{G}})\colon \mathcal{G}\to \mathcal{G}^{\Sigma}$ is a \emph{term graph} if $t_\mathcal{G}$ is a mono. We define $\tg$ to be the full subcategory of $\hyp_{\Sigma}$ given by term graphs and denote by $I_\Sigma$ the inclusion. Restricting $U_\Sigma\colon \hyp_{\Sigma}\to \catname{Set}$ we get a forgetful functor $U_{\tg}\colon \tg\to \catname{Set}$.
\end{definition}

\begin{remark}\label[remark]{rem:mono2}By \Cref{rem:label}, we know that if $\mathcal{G}$ is a term graph then $t_{\mathcal{G}}=v_1\circ \tau_{\mathcal{G}}$, where $v_1$ is the coprojection of $V_{\mathcal{G}}$ into $V^\star_{\mathcal{G}}$.  Notice that since $t_{\mathcal{G}}$ is a mono then $\tau_{\mathcal{G}}$ is a mono.
\end{remark}

\begin{example}
The labelled hypergraph of \Cref{lab_2} is a term graph.
\end{example}

We now examine some properties of $\tg$, in order to study its adhesivity properties. We begin noticing that, for every set $X$,  $\Delta_\Sigma(X)$ has no hyperedges and so is a term graph. this yields at once the following.

\begin{proposition}\label[proposition]{term:left}The forgetful functor $U_{\tg}$ has a left adjoint $\Delta_{\tg}$.
\end{proposition}
\commentato{
\begin{proof}
	This follows noticing that $\Delta_{\Sigma}(X)$ is a term graph for every object $X$.
\end{proof}}

We can list some other categorical properties of $\tg$ (see \cite[Sec.~5]{CastelnovoGM24}).

\begin{proposition}\label{prop:tlim}
Let $\Sigma$ be an algebraic signature. Then the following hold
\begin{enumerate}
	\item if  $(i,j)\colon \mathcal{H}\to \mathcal{G}$ is a mono from $(l, !_{V_\mathcal{G}})\colon \mathcal{G}\to \mathcal{G}^{\Sigma}$ to $(l', !_{V_\mathcal{H}})\colon \mathcal{H}\to \mathcal{G}^{\Sigma}$ in $\hyp_\Sigma$ and the latter is in $\tg$, then also the former is in $\tg$
	\item $\tg$ has equalizers, binary products and pullbacks and they are created by $I_\Sigma$.
\end{enumerate}
\end{proposition}

\begin{remark}
	$\tg$ in general does not have terminal objects. 
	%Consider an algebraic signature in $\Set$. 
	Since $U_{\tg}$ preserves limits, if a terminal object exists it must have the singleton as set of nodes, therefore the set of hyperedges must be empty or a singleton. 
	Hence, for a counterexample, it suffices to take a signature with two operations $a$ and $b$, both of arity $0$.
	$\tg$ is not an adhesive category, either. 
	In particular, as noted in e.g.~\cite{CastelnovoGM24}, 
	 it does not have pushouts along all monos. 
\end{remark}
	
\commentato{
\begin{remark}
	$\tg$ in general does not have terminal objects. Consider an algebraic signature in $\Set$. Since $U_{\tg}$ preserves limits, if a terminal object exists it must have the singleton as set of nodes, therefore the set of hyperedges must be empty or a singleton $\{h\}$. Hence, for a counterexaqmple it suffices to take as signature the one given by two operations $a$ and $b$, both of arity $0$; we have three term graphs with only one node $v$: $\Delta_{\tg{\Sigma}}(\{v\})$, $(l_a, !_{V_{\mathcal{G}}})\colon \mathcal{G}_a\to \mathcal{G}^{\Sigma}$ and $(l_b, !_{V_{\mathcal{G}}})\colon \mathcal{G}_b\to \mathcal{G}^{\Sigma}$.
	\begin{center}\begin{tikzpicture}
			
			\node[circle,fill=black,inner sep=0pt,minimum size=6pt,label=left:{$v$}] (V) at (5,0) {};
			\node[circle,fill=black,inner sep=0pt,minimum size=6pt,label=left:{$v$}] (U) at (3,0) {};
			\node at(5,1.25){$a$};	
			\node at(5,1.7){$h$};	
			
			\draw[->-=.5](5,1)--(V)node[pos=0.5, right,font=\fontsize{7}{0}\selectfont]{$1$};
			
			\draw[rounded corners] (4.75, 1) rectangle (5.25, 1.5) {};

			\node[circle,fill=black,inner sep=0pt,minimum size=6pt,label=left:{$v$}] (V) at (7,0) {};
			
			\node at(7,1.25){$b$};	
			\node at(7,1.7){$h$};	
			
			\draw[->-=.5](7,1)--(V)node[pos=0.5, right,font=\fontsize{7}{0}\selectfont]{$1$};
			
			\draw[rounded corners] (6.75, 1) rectangle (7.25, 1.5) {};
			
		\end{tikzpicture}
	\end{center}
	There are no morphisms in $\tg$ between the last two and from the last two to the first one, therefore none of them can be terminal.
\end{remark}

\begin{remark}
	$\tg$ is not an adhesive category. In particular it does not have pushouts along all monos. For instance, if we take the three term graphs of the previous remark, then have two arrows
	$(?_{\{h\}}, \id{\{v\}})\colon \Delta_{\tg}(\{v\})\to (l_a, !_{V_{\mathcal{G}_a}})$ and $(?_{\{h\}}, \id{\{v\}})\colon \Delta_{\tg}(\{v\})\to (l_b, !_{V_{\mathcal{G}_a}})$ which cannot be completed to a square. Indeed if $(q, !_{V_\mathcal{H}})\colon \mathcal{H}\to \mathcal{G}^\Sigma$ is another term graph with $(g_E, g_V)\colon (l_a, !_{V_{\mathcal{G}}})\to (q, !_{V_\mathcal{H}})$ and $(k_E, k_V)\colon (l_a, !_{V_{\mathcal{G}}})\to (q, !_{V_\mathcal{H}})$  such that 
	\[(g_E, g_V)\circ (?_{\{h\}}, \id{\{v\}}) = (k_E, k_V)\circ (?_{\{h\}}, \id{\{v\}})\]
	then $g_V=k_V$ and
	\[t_{\mathcal{H}}(g_E(h))=g^\star_V(t_{\mathcal{G}}(h))=g_V^\star(\delta_v)=k^\star_V(\delta_V)=k^\star_V(t_{\mathcal{G}}(h))=t_{\mathcal{H}}(k_E(h))\]
	so that we also have $g_E=k_E$, but then
	\[
	a=l_a(h)=q(g_E(h))=q(k_E(h))=l_b(h)=b\]
\end{remark}
}

\begin{definition}
	Let $(l, !_{V_{\mathcal{G}}})\colon \mathcal{G}\to \mathcal{G}^{\Sigma}$  be a term graph. A \emph{input node} is an element of $V_{\mathcal{G}}$ not in the image of $\tau_{\mathcal{G}}$.  A morphism $(f,g)$ between
	 $(l, !_{V_{\mathcal{G}}})\colon \mathcal{G}\to \mathcal{G}^{\Sigma}$ and $(l, !_{V_{\mathcal{H}}})\colon \mathcal{H}\to \mathcal{G}^{\Sigma}$ in $\tg$, is said to \emph{preserve input nodes} if $g$ sends input nodes to input nodes.
\end{definition}

\commentato{
\todo{Se questo remark sotto non serve nel pezzo sulle equivalenze possiamo toglierlo}
\begin{remark}\label{prop:image}
	Suppose that $(f,g)\colon (l, !_{V_{\mathcal{G}}})\to (l', !_{V_{\mathcal{H}}})$ preserves input nodes. Then  if $\tau_{\mathcal{H}}(h)=g(v)$ for some $v\in V_{\mathcal{G}}$ then $h$ belongs to the image of $f$. Indeed, by hypothesis $v$ must be in the image of $\tau_{\mathcal{G}}$ and so there exists $k$ such that $\tau_{\mathcal{G}}(k)=v$. But then $\tau_{\mathcal{H}}(f(k))=g(v)$ and we can conclude that $f(k)=h$.
\end{remark}
}

Preservation of input nodes characterizes regular monos in $\tg$.

\begin{proposition}\label[proposition]{lem:reg} Let $(i,j)$ be a mono between two term graphs  $(l, !_{V_{\mathcal{G}}})\colon \mathcal{G}\to \mathcal{G}^{\Sigma}$ and  $(l', !_{V_{\mathcal{H}}})\colon \mathcal{H}\to \mathcal{G}^{\Sigma}$. Then it is a regular mono if and only if it preserves the input nodes.
\end{proposition}

This characterization, in turn, provides us with the following result \cite{CastelnovoGM24,castelnovo2023thesis}. 

\begin{lemma}\label[lemma]{prop:push} Consider three term graphs $(l_0, !_{V_\mathcal{G}})\colon \mathcal{G}\to \mathcal{G}^{\Sigma}$, $(l_1, !_{V_\mathcal{H}})\colon \mathcal{H}\to \mathcal{G}^{\Sigma}$ and $(l_2, !_{V_\mathcal{K}})\colon \mathcal{K}\to \mathcal{G}^{\Sigma}$. Given $(f_1, g_1)\colon (l_0, !_{V_\mathcal{G}})\to (l_1, !_{V_\mathcal{H}})$, $(f_2, g_2)\colon (l_0, !_{V_\mathcal{G}})\to (l_2, !_{V_\mathcal{K}})$, if $(f_1, g_1)$ is a regular mono, then its pushout  $(p, !_{V_{\mathcal{P}}})\colon \mathcal{P}\to \mathcal{G}^{\Sigma}$ in $\hyp_{\Sigma}$ along $(f_2, g_2)$ is a term graph.
\end{lemma}

\Cref{thm:slice-functors,prop:mono}, \Cref{lem:reg} and \Cref{prop:push} allow us to recover the following result, previously proved by direct computation in \cite[Thm.~4.2]{CorradiniG05} (see also \cite[Cor.~5.15]{CastelnovoGM24} for the details of the argument presented here).
\begin{corollary}\label{cor:term}
	The category $\tg$ is quasiadhesive.
\end{corollary}

\section{Adding equivalences to hypergraphical structures}
\label{hypereq}
The use of hypergraphs equipped with an equivalence relation over their nodes has been argued as a convenient way to express concurrency in the DPO approach to rewriting~\cite{concur2006}.
This section presents the framework by means of adhesive categories, including also its version for term graphs, as a stepping stone towards an analogous characterisation for e-graphs.

\subsection{Hypergraphs with equivalence}

Let us start with the case of general hypergraphs. These were introduced in~\cite{concur2006}, even if no general consideration about their structure as a category was proved, and adhesivity, 
which is the main focus here, was yet to be presented to the world.

\begin{definition}
	A \emph{hypergraph with equivalence} $\mathcal{G} = (E_\mathcal{G}, V_{\mathcal{G}}, Q_\mathcal{G}, s_\mathcal{G}, t_\mathcal{G}, q_\mathcal{G})$ is a 6-tuple such that $\mathcal{G} = (E_\mathcal{G}, V_{\mathcal{G}}, s_\mathcal{G}, t_\mathcal{G})$ is a hypergraph, $Q_\mathcal{G}$ is a set and $q_{\mathcal{G}}: V_{\mathcal{G}}\eto Q_{\mathcal{G}}$ is a surjection called \emph{quotient map}. 
	A morphism $h\colon \mathcal{G\to H}$ is a triple $(h_E, h_V, h_Q)$ such that the following diagrams commute
	\[\xymatrix@R=15pt{
		{E_\mathcal{G}}\ar[r]^{s_\mathcal{G}}\ar[d]_{h_E} & {V_{\mathcal{G}}^\star}\ar[d]^{h_V^\star} & {E_\mathcal{G}}\ar[r]^{t_\mathcal{G}}\ar[d]_{h_E} & {V_{\mathcal{G}}^\star}\ar[d]^{h_V^\star} & {V_\mathcal{G}}\ar@{>>}[r]^{q_\mathcal{G}}\ar[d]_{h_V} & {Q_{\mathcal{G}}} \ar[d]^{h_Q} \\
		{E_\mathcal{H}}\ar[r]_{s_\mathcal{H}} & {V_{\mathcal{H}}^\star}	& {E_\mathcal{H}}\ar[r]_{t_\mathcal{H}} & {V_{\mathcal{H}}^\star}& {V_\mathcal{H}}\ar@{>>}[r]_{q_\mathcal{H}} & {Q_\mathcal{H}}
	}\]
	The category of hypergraphs with equivalence and their morphisms is denoted $\EqHyp$.
	
\end{definition}

\begin{remark}
	Notice that in $\Set$ the classes of surjections, epis and regular epis coincide.
\end{remark}

\begin{remark}\label{rem:eqhyp_morphs}
	Morphisms of hypergraphs with equivalences are uniquely determined by the first two components. That is, if $h_1 = (h_E, h_V, f)$ and $h_2 = (h_E, h_V, g)$ are two morphisms $\mathcal{G} \rightrightarrows \mathcal{H}$, then we have
	$
	f \circ q_\mathcal{G} = q_\mathcal{H}\circ h_V =g\circ q_\mathcal{G}.
	$
	Since $q_\mathcal{G}$ is epi, we obtain $f = g$.
\end{remark}

Forgetting the quotient part yields a functor $T\colon \EqHyp \to \hyp$ sending a hypergraph with equivalence $(E_\mathcal{G}, V_{\mathcal{G}}, C_\mathcal{G}, s_\mathcal{G}, t_\mathcal{G}, q_\mathcal{G})$ to $(E_{\mathcal{G}}, V_{\mathcal{G}}, s_\mathcal{G}, t_{\mathcal{G}})$.   We now explore some of its properties to deduce information on the structure of $\EqHyp$.  
A proof is in \Cref{proof:forghyp}.

\begin{restatable}{proposition}{fhyp}\label{prop:forghyp}  Consider the forgetful functor $T\colon \EqHyp \to \hyp$. Then the following hold
	\begin{enumerate}
		\item$T$ is faithful;
		\item $T$ has a left adjoint;
		\item $T$ has a right adjoint.
	\end{enumerate}
\end{restatable}

\begin{corollary}\label{cor:limcolim}
	The functor $T$ preserves limits and colimits.
\end{corollary}

From the previous results we get the following characterization of monos in $\EqHyp$.

\begin{corollary}\label{cor:mono1}
	An arrow $(h_E, h_V, h_Q): \mathcal{G \to H}$ in $\EqHyp$ is a mono if and only if $(h_E, h_V)$ is a mono in $\hyp$.
\end{corollary}

Now, we can consider the forgetful functor $U_{\eq}\colon \EqHyp\to \Set$ obtained by composing $T$ and $U_{\hyp}$.  By \Cref{cor:left} and the second point of \Cref{prop:forghyp} we get the following.

\begin{corollary}\label{cor:ladj}
	$U_{\eq}$ has a left adjoint $\Delta_{\eq}\colon \Set \to \EqHyp$.
\end{corollary}

Notice that there is another functor $K: \EqHyp \to \Set$ sending $(E, V, Q, s, t, q)$ to $Q$, and a morphism $(h_E, h_V, h_Q)$ to $h_Q$. We  exploit it to compute limits and colimits in $\EqHyp$. A full proof of the following lemma can be found in \Cref{proof:comp}.

\noindent
\parbox{11.4cm}{
\begin{restatable}{lemma}{comp}\label{prop:eqhyp_complete}
Consider a diagram $F\colon \D \to \EqHyp$ and let $(E_D, V_D, Q_D, s_D, t_D, q_D)$ be the image of an object $D$. Then the following hold
\begin{enumerate}
		\item $F$ has a colimit, which is preserved by $K$;
	\item consider a cone $(L, \{l_D\}_{D\in \D})$ limiting  for $K \circ F$ and let $((E, V), \{(\pi^D_E, \pi^D_V)\}_{D\in \D})$ be one for $T\circ F$, then $F$ has a limit $(E, V, Q, s, t, q)$ and there is a mono $m\colon Q\mto L$ such that the diagram on the right commutes for every $D\in \D$.
\end{enumerate}
\end{restatable}}\hfill 
\parbox{3cm}{\xymatrix{V \ar@{>>}[d]_{q} \ar[r]^{\pi^D_V}& V_D \ar@{>>}[dd]^{q_D}\\Q \ar@{>.>}[d]_{m}&\\ L \ar[r]_{l_d} & Q_D}}   

\begin{restatable}{corollary}{mn}\label{cor:mono2}
	An arrow $(h_E, h_V, h_Q): \mathcal{G\to H}$ in $\EqHyp$ is a regular mono if and only if all its components are injective functions.
\end{restatable}

\commentato{
\todo{Ripensandoci il funtore quoziente non serve: l'unico punto in cui si usa nella tesi è che il quoziente di un colimite è il colimite dei quozienti, ma questo segue dal punto $1$ di \Cref{prop:eqhyp_complete}. Inoltre la dimostrazione della creazione è sbagliata (e probabilmente quella è una cosa falsa), io cancellerei}

We can now build another functor $\EqHyp \to \hyp$: the idea is that from an object of $\EqHyp$ one can build a hypergraph taking as set of nodes directly the codomain of the quotient map.  

\begin{definition}
	We define the \emph{quotient functor} $\quo :\EqHyp\to \hyp $ as the one sending $(E, V, Q, s, t, q)$ to $(E, Q, q^{\star}\circ s, q^{\star}\circ t)$ and an arrow $(h_E, h_V, h_Q)$ to $(h_E, h_Q)$.
\end{definition}

\commentato{
\begin{remark}
	The action of the functor on a morphism of hypergraphs with equivalences gives a morphism of hypergraphs,
	in fact $q^{\star}_\mathcal{H} \circ s_\mathcal{H} \circ h_E = q^{\star}_\mathcal{H} \circ h_V^\star \circ s_\mathcal{G} = h_Q^\star \circ q^{\star}_\mathcal{G} \circ s_\mathcal{G}$.
	The same is valid for $t_\mathcal{H}$ and $t_\mathcal{G}$. 
\end{remark}
}

 A proof of the following result can be found in \Cref{proof:quot}.

\begin{restatable}{lemma}{quot}\label{lemma:quot} The functor $\quo$ is a left adjoint.
\end{restatable}

\begin{example} Notice that, due to the way in which limits in $\EqHyp$ are computed (see \Cref{prop:eqhyp_complete}), $\quo$ does not preserve them.  For instance,  let $A$ and $B$ be two non-empty sets and define
	\begin{gather*}
	\mathcal{G}_1:=(\emptyset, A, A, ?_{A^\star}, ?_{A^\star}, \id{A}) \qquad \mathcal{G}_2 := (\emptyset, B, B, ?_{B^\star}, ?_{B^\star}, \id{B}) \\ \mathcal{G}_3 := (\emptyset, A + B, 1, ?_{(A+B)^\star}, ?_{(A+B)^\star}, !_{A + B})	\end{gather*}
		
Now, we have two morphisms $(\id{\emptyset}, \iota_A, !_A)\colon \mathcal{G}_1 \to \mathcal{G}_3$,  $(\id{\emptyset}, \iota_B, !_B)\colon  \mathcal{G}_2 \to \mathcal{G}_3$, where $\iota_A\colon A\mto A+B$  and $\iota_B\colon B \mto A+B$  are the coprojections.

On the one hand, if we apply $\quo$ and compute the pullback of the resulting morphisms, we obtain a hypergraph with no hyperdeges and $A\times B$ as set of nodes. On the other hand, \Cref{prop:eqhyp_complete} we know that taking the pullback of the original two morphisms gives us the hypergraph with equivalence $(\emptyset, \emptyset, Q, ?_{\emptyset^\star}, ?_{\emptyset^\star}, q)$ where $Q$ is the image of $?_{A\times B}$, so that $Q$ is also empty.  
\end{example}
}

A proof is in \Cref{proof:regmono}.
We have now all the ingredients to study the adhesivity properties of $\EqHyp$.  As a first step we need to introduce a class of monos.

\noindent
\parbox{11cm}{
\begin{definition}
	We define $\pbc$ as the class of  regular monos $(h_E, h_V, h_Q)\colon \mathcal{G}\to \mathcal{H}$ in $\EqHyp$ such that the square on the right is a pullback
\end{definition}
} \hfill
	\parbox{2cm}{
	     \xymatrix@R=18pt{
	        V_\mathcal{G} \ar@{>>}[r]_{q_\mathcal{G}} \ar@{>->}[d]_{h_V}& Q_\mathcal{G} \ar@{>->}[d]^{h_Q}\\
	        V_\mathcal{H} \ar@{>>}[r]^{q_\mathcal{H}}  & Q_\mathcal{{H}}
              }
         }   

\vspace{.1cm}
Now, we show that $\pbc$ is a suitable class for adhesivity. A proof is in \Cref{proof:pbmono}.

\begin{restatable}{lemma}{pbm}\label{lem:pbmono}
	The class $\pbc$ contains all isomorphisms, it is closed under composition, decomposition and it is stable under pullbacks and pushouts.
\end{restatable}

\vspace{.1cm}
Finally, we show the key lemma for $\pbc$-adhesivity of $\EqHyp$. A proof is in \Cref{proof:pbstable}.

\begin{restatable}{lemma}{pbs}\label{lemma:stab}
	In $\EqHyp$, $\pbc$-pushouts are stable.
\end{restatable}

\begin{example}
It is noteworthy to show that pushouts along regular monos are not stable. 
\begin{minipage}[l]{.35\linewidth}Consider e.g. the cube aside: the vertices are just graphs without edges, and in the graphs themselves the equivalence classes 
are denoted by encircling nodes with a dotted line. All the arrows are regular monos. It is immediate to see that the bottom face is a pushout and the side faces are pullback. Unfortunately, 
the top face fails to be a pushout
  \end{minipage}\hfill
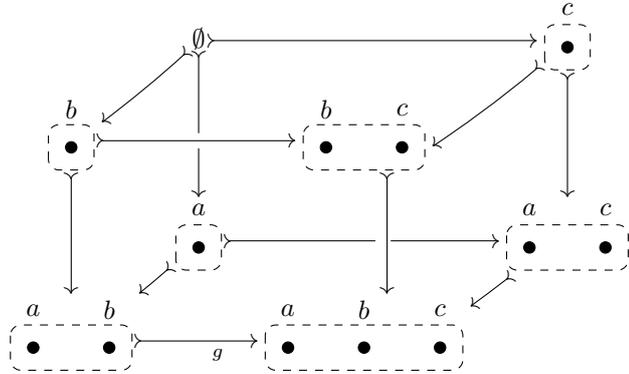
\begin{minipage}[r]{.60\linewidth}
                \xymatrix@C=10pt@R=6pt{
                &\emptyset \ar@{>->}[dd]|\hole \ar@{>->}[rr] \ar@{>->}@(dl,)[dl] &&
                \begin{tikzpicture}[baseline=(v.base)]\begin{pgfonlayer}{nodelayer}
                                \node[style=node](v) at (0, 0){};
                                \node[style=none] at(0, 0.5){$c$};
		\end{pgfonlayer}
                \begin{pgfonlayer}{eqlayer}\draw[dashed, rounded corners] (-0.3, -0.3) rectangle (0.3, 0.3);\end{pgfonlayer}
                \end{tikzpicture} \ar@{>->}[dd] \ar@{>->}@(dl,)[dl]!<6ex,0ex> \\
                \begin{tikzpicture}[baseline=(v.base)]\begin{pgfonlayer}{nodelayer}
                        \node[style=node](v)at(0, 0){};
                        \node[style=none]at(0, 0.5){$b$};
                \end{pgfonlayer}
                \begin{pgfonlayer}{eqlayer}\draw[dashed, rounded corners] (-0.3, -0.3) rectangle (0.3, 0.3);\end{pgfonlayer}
        \end{tikzpicture}
                \ar@{>->}[dd]\ar@{>->}[rr] & &
                \begin{tikzpicture}[baseline=(v.base)]\begin{pgfonlayer}{nodelayer}
                        \node[style=node](v)at(-0.5, 0){};
                        \node[style=none]at(-0.5, 0.5){$b$};
                        \node[style=node]at(0.5, 0){};
                        \node[style=none]at(0.5, 0.5){$c$};
                \end{pgfonlayer}\begin{pgfonlayer}{eqlayer}
                        \draw[dashed, rounded corners](-0.8, 0.3) rectangle (0.8, -0.3);
                \end{pgfonlayer}\end{tikzpicture}
                \ar@<2ex>@{>->}[dd]\\&
                \begin{tikzpicture}[baseline=(v.base)]\begin{pgfonlayer}{nodelayer}
                        \node[style=node](v)at(0, 0){};
                        \node[style=none]at(0, 0.5){$a$};
                \end{pgfonlayer}
                \begin{pgfonlayer}{eqlayer}\draw[dashed, rounded corners] (-0.3, -0.3) rectangle (0.3, 0.3);\end{pgfonlayer}
        \end{tikzpicture}
                \ar@{>->}[rr]|\hole \ar@{>->}@(dl,)[dl] && 
                \begin{tikzpicture}[baseline=(v.base)]\begin{pgfonlayer}{nodelayer}
                        \node[style=node](v)at(-0.5, 0){};
                        \node[style=none]at(-0.5, 0.5){$a$};
                        \node[style=node]at(0.5, 0){};
                        \node[style=none]at(0.5, 0.5){$c$};
                \end{pgfonlayer}\begin{pgfonlayer}{eqlayer}
                        \draw[dashed, rounded corners](-0.8, 0.3) rectangle (0.8, -0.3);
                \end{pgfonlayer}\end{tikzpicture}
                \ar@{>->}@(,r)[dl] \\
                \begin{tikzpicture}[baseline=(v.base)]\begin{pgfonlayer}{nodelayer}
                        \node[style=node](v)at(-0.5, 0){};
                        \node[style=none]at(-0.5, 0.5){$a$};
                        \node[style=node]at(0.5, 0){};
                        \node[style=none]at(0.5, 0.5){$b$};
                \end{pgfonlayer}\begin{pgfonlayer}{eqlayer}
                        \draw[dashed, rounded corners](-0.8, 0.3) rectangle (0.8, -0.3);
                \end{pgfonlayer}\end{tikzpicture}
                \ar@{>->}[rr]_{g} & & 
                \begin{tikzpicture}[baseline=(v.base)]\begin{pgfonlayer}{nodelayer}
                        \node[style=node](v)at(0, 0){};
                        \node[style=none]at(0, 0.5){$a$};
                        \node[style=node]at(1, 0){};
                        \node[style=none]at(1, 0.5){$b$};
                        \node[style=node]at(2, 0){};
                        \node[style=none]at(2, 0.5){$c$};
                \end{pgfonlayer}\begin{pgfonlayer}{eqlayer}
                        \draw[dashed, rounded corners](-.3, 0.3) rectangle (2.3, -0.3);
                \end{pgfonlayer}\end{tikzpicture} }
  \end{minipage}
\end{example}

Having proved stability of $\pbc$-pushouts, we now turn to prove that they are Van Kampen with respect to regular monos.
A proof is found in \Cref{proof:pbvk}.

\begin{restatable}{lemma}{pbvk}\label{lemma:van_kampen}
	In $\EqHyp$, pushouts along arrows in $\pbc$ are $\reg(\EqHyp)$-Van Kampen.
\end{restatable}

\begin{corollary}\label{cor:eqade}
	$\EqHyp$ is $\pbc$-adhesive.
\end{corollary}

\subsection{Term graphs with equivalence}

We are now going to generalize the work done in the previous section equipping term graphs with equivalence relation. First of all we need a notion of labelling for $\EqHyp$. 

\begin{definition}
	Let $\Sigma$ be an algebraic signature and $\mathcal{G}^{\Sigma}$ the hypergraph associated to it. A \emph{labelled hypergraph with equivalence} is a pair $(\mathcal{H}, l)$ where $\mathcal{H}$ is an object of $\EqHyp$ and $l$ a morphism $T(\mathcal{H})\to \mathcal{G}^\Sigma$ of $\hyp$. A \emph{morphism of labelled hypergraphs with equivalence} between $(\mathcal{H}, l)$ and $(\mathcal{H}', l')$ is an arrow $h\colon \mathcal{H}\to \mathcal{H}'$ such that $l= l'\circ T(h)$.
	
	We denote the resulting category by $\EqHyp_\Sigma$.
\end{definition}

Let $(\mathcal{H}, l)$ be an object of $\EqHyp_\Sigma$: since $T$ has a right adjoint $R$ by \Cref{prop:forghyp}, $l\colon T(\mathcal{H})\to \mathcal{G}^\Sigma$ corresponds to the arrow $(l, !_{Q_{\mathcal{H}}})\colon \mathcal{H}\to R(\mathcal{G}^\Sigma)$. It is immediate to see that this correspondence extends to an equivalence with the slice over $R(\mathcal{G}^\Sigma)$.

\begin{proposition}\label{prop:slice}
$\EqHyp_\Sigma$ is equivalent to $\EqHyp/R(\mathcal{G}^\Sigma)$.
\end{proposition}

Let $V_\Sigma\colon \EqHyp_\Sigma\to \EqHyp$ be the functor forgetting the labelling and $\pbc_\Sigma$ the class of morphisms in $\EqHyp_\Sigma$ whose image in $\EqHyp$ is in $\pbc$. From \Cref{prop:eqhyp_complete,cor:pbapp}, the second point of \Cref{thm:slice-functors,cor:eqade}, we can deduce the following.

\begin{proposition}\label{prop:lim}Consider the forgetful functor $V_{\Sigma}\colon \EqHyp_\Sigma\to \EqHyp$. Then the following hold
	\begin{enumerate}
		\item $\EqHyp_\Sigma$ has all colimits and all connected limits, which are created by $V_{\Sigma}$;
		\item $\EqHyp_\Sigma$ is $\pbc_\Sigma$-adhesive.
	\end{enumerate}
\end{proposition}

%In particular, $\EqHyp_\Sigma$ has all equalizers and pullbacks.
Using \Cref{cor:mono1,cor:mono2} we immediately get the following result. 

\begin{corollary}\label{prop:monos_in_eqhyps} Let $h=(h_{E}, h_V, h_Q)$ be an arrow in $\EqHyp_\Sigma$. Then the following hold
	\begin{enumerate}
		\item $h$ is mono if and only if $h_E$ and $h_V$ are injective;
		\item $h$ is a regular mono if and only if $h_E$, $h_V$ and $h_Q$ are injective.
	\end{enumerate}
\end{corollary}

We can now easily define term graphs with equivalence.

\begin{definition}Let $\Sigma$ be an algebraic signature.
	An object $(\mathcal{H}, l)$ of $\EqHyp_\Sigma$ is a \emph{term graph with equivalence} if $l\colon T(\mathcal{H})\to \mathcal{G}^\Sigma$ is a term graph. We denote by $\EqTG_{\Sigma}$ the full subcategory of $\EqHyp_\Sigma$ so defined and by $J_{\Sigma}$ the corresponding inclusion functor.
 \end{definition}

\noindent 
\parbox{9.5cm}{\begin{remark}\label{rem:obv}
Let $T_\Sigma:  \EqHyp_{\Sigma} \to \hyp_{\Sigma}$ the forgetful functor lifting
$T:  \EqHyp \to \hyp$.
%
%$T_\Sigma = T\circ V_\Sigma$.   
Notice that, by definition, $(\mathcal{H}, l)$ is in $\EqTG_{\Sigma}$ amounts to say that $(T(\mathcal{H}), l)$ is in $\tg$. 
Thus there exists a functor $S_\Sigma$ as in the diagram on the right.
%Let $T_\Sigma$ be the functor sending $(\mathcal{H}, l)$ to $(T(\mathcal{H}), T(l))$.   Notice that, by definition, $(\mathcal{H}, l)$ is in $\EqTG_{\Sigma}$ if and only if $T(\mathcal{H})$ is in $\tg$. Thus there exists a functor $S_\Sigma$ as in the diagram on the right.
\end{remark}}
\hfill 
\parbox{4cm}{\xymatrix@R=15pt{\EqTG_{\Sigma} \ar[r]_{J_\Sigma}  \ar[d]_{S_\Sigma}& \EqHyp_\Sigma \ar[d]^{T_\Sigma} \\ \tg \ar[r]^{I_\Sigma}& \hyp_{\Sigma}}}

\begin{remark}\label{rem:t}
	Notice that, by \Cref{cor:limcolim} and \Cref{prop:lim}, the functor $T_\Sigma$ preserves all connected limits and all colimits.
\end{remark}

 The previous remark allows us to prove an analog of \Cref{prop:tlim}. We refer the reader to \Cref{proof:term} for details.
 
 \begin{restatable}{proposition}{trm}\label{prop:term}
$\EqTG_{\Sigma}$ has equalizers, binary products and pullbacks and they are created by $J_\Sigma$.
 \end{restatable}
 
\noindent
\parbox{11cm}{\hspace{15pt}Let now $\mathcal{T}$ be the class of morphism in $\EqTG_{\Sigma}$ whose image through $J_\Sigma$ is in $\pbc_\Sigma$ and whose image through $S_\Sigma$ is a regular mono in $\tg$.  By \Cref{lem:reg,prop:monos_in_eqhyps}, we have that a morphism $(h_E, h_V, h_Q)\colon (\mathcal{G},l)\to (\mathcal{H}, l')$ is in $\mathcal{T}$ if and only if all of its components are injections and the square on the right is a pullback.} \hfill 
\parbox{2cm}{\xymatrix{V_{\mathcal{G}}\ar@{>->}[r]^{h_V} \ar@{>>}[d]_{q_{\mathcal{H}}}& V_{\mathcal{H}} \ar@{>>}[d]^{q_{\mathcal{H}}}\\ Q_{\mathcal{G}} \ar@{>->}[r]_{h_Q}& Q_{\mathcal{H}}} }

\smallskip 
In particular, by \Cref{cor:term} and \Cref{lemma:stab}, $\mathcal{T}$ contains all isomorphisms, is closed under composition and decomposition and stable under pushout and pullbacks.

The following proposition is now an easy corollary of \Cref{prop:push}, \Cref{prop:lim}, and \Cref{rem:obv}. We provide the details in \Cref{proof:tade}.

\begin{restatable}{proposition}{po}\label{prop:po}
$\EqTG_{\Sigma}$ has all $\mathcal{T}$-pushouts, which are created by $J_{\Sigma}$.
\end{restatable}

\begin{corollary}\label{cor:tade}
	$\EqTG_{\Sigma}$ is $\mathcal{T}$-adhesive.
\end{corollary}

\section{EGGs}
\label{eggs}
The previous section proved some adhesivity property for the categories $\EqHyp$ and $\EqTG_{\Sigma}$. We extend these results to
encompass equivalence classes which are closed under the target arrow, i.e.~under operator composition for term graphs, thus precisely capturing EGGs.
%
%\todo{introduction}

\subsection{E-hypergraphs}

We start introducing the notion of \emph{e-hypergraphs}, hypergraphs equipped with an equivalence relation that is closed under the target arrow:
in other words, whenever the relation identifies the source of two hyperedges, it identifies their targets too.

\begin{definition}
	Let $\mathcal{G} = (E, V, Q, s, t, q)$ be a hypergraph with equivalence and $(S, \pi_1, \pi_2)$ a kernel pair for $q^\star \circ s$.
	We will say that $\mathcal{G}$ is an \emph{e-hypergraph} if $q^\star \circ t \circ \pi_1 = q^\star \circ t \circ \pi_2$.
	
	We will denote by $\egg$ the full subcategory of $\EqHyp$ whose objects are e-hypergraphs, and by $I\colon \egg \to \EqHyp$ the associated inclusion functor.
\end{definition}

\begin{example}
	Consider the hypergraph $\mathcal{G}$ of \Cref{exa_2} and consider as quotient the identity $\id{V_\mathcal{G}}\colon V_\mathcal{G}\colon \to\mathcal{G}$. Then the kernel pair of $\id{V^\star_\mathcal{G}} \circ s$ concide with the kernel pair of $s$, which is empty. Thus  $\mathcal{G}$ is, trivially, an e-hypergraph
\end{example}
A first result that we need is that limits in $\egg$ are computed as in $\EqHyp$. Full proofs are in \Cref{proof:elim}.

\begin{restatable}{lemma}{elim}\label{lem:elim}
	$\egg$ has all limits and $I$ creates them.
\end{restatable}

\begin{corollary}\label{cor:ereg}
	If an arrow $h: \mathcal{G \to H}$ in $\egg$ is a regular mono in $\egg$ then $I(h)$ is a regular mono in $\hyp$. 
\end{corollary}

We can now turn to pushouts. We refer again the reader to \Cref{proof:epo} for the details.

\noindent
\parbox{11.5cm}{\begin{restatable}{lemma}{epo}\label{lem:epo}
	Suppose that the square on the right is a pushout in $\EqHyp$ and that $m$ is a mono in $\pbc$. If $\mathcal{G}_1$, $\mathcal{G}_2$ and $\mathcal{G}_3$ are e-hypergraphs, then $\mathcal{P}$ is 
	is an e-hypergraph.
	%so too.
	%an object of $\egg$ too.
\end{restatable}}\hfill 
\parbox{4cm}{\xymatrix@R=15pt{\mathcal{G}_1 \ar[r]^{h}\ar@{>->}[d]_{m}&\mathcal{G}_2\ar[d]^{n}\\\mathcal{G}_3\ar[r]_{z}&\mathcal{P}}}

\commentato{ 
\begin{restatable}{corollary}{rege}\label{cor:rege}
		An arrow $h\colon \mathcal{G \to H}$  is a regular mono in $\egg$ if and only if $I(h)$ is a regular mono in $\hyp$.  In particular a morphism of $\egg$ is a regular mono if and only if all its components are injections.
\end{restatable}
\begin{proof}
	contenuto...
\end{proof}
}

Let now $\pbe$ be the class of arrows in $\egg$ that are sent to $\pbc$ by $I$. By definition and \Cref{lem:pbmono} we have that such class is closed under composition and decomposition, it contains all isomorphisms and it is stable under pullbacks. Using \Cref{lem:pbmono,lem:epo} we can further deduce that it is stable under pushouts. Then \Cref{thm:slice-functors,cor:eqade} yield the following.

\begin{corollary}\label{cor:pbe}
	$\egg$ is $\pbe$-adhesive.
\end{corollary}

\subsection{E-term graphs, a.k.a. EGGs}

We now turn our attention to the labelled context.

\begin{definition} We say that an object $(\mathcal{H}, l)$ of $\EqHyp_\Sigma$ is a \emph{labelled e-hypergraph} if $\mathcal{H}$ is an e-hypergraph. We define the category $\egg_\Sigma$ as the full subcategory of $\EqHyp_\Sigma$ given by labelled e-hypergraphs. We denote by $Z_\Sigma$ the corresponding inclusion functor .
\end{definition}

To prove some adhesivity property of $\egg_\Sigma$, we begin with the following elementary, yet useful observation.

\begin{remark}\label{rem:algegg}
	Given a signature $\Sigma = (O_\Sigma, \ari_\Sigma)$,
	the hypergraph 
	$R(\mathcal{G}^\Sigma) = (O_\Sigma, 1, 1, \ari_\Sigma, \gamma_1, \id{1})$
	is an object of $\egg$.
	Indeed, under the identification of $1^\star$ with $\mathbb{N}$, the kernel $(S, \pi_1, \pi_2)$  of $\id{\mathbb{N}} \circ \ari_\Sigma$, is given by
	$S:=\{(o_1, o_2)\in O_\Sigma \times O_\Sigma \mid \ari_\Sigma(o_1)=\ari_\Sigma(o_2)\}$
	equipped with the two projections.  On the other hand, both $\id{\mathbb{N}} \circ \gamma_1 \circ \pi_1$ and $\id{\mathbb{N}}\circ \gamma_1 \circ \pi_2$ are the function $O_\Sigma \to \mathbb{N}$ constant in $1$.
\end{remark}

\Cref{rem:algegg} now implies at once the following result.

\samepage{\begin{proposition}\label{prop:varie}
	Let $\Sigma$ be an algebraic signature. Then the following hold
	\begin{enumerate}
		\item $\egg_\Sigma$ is equivalent to $\egg/R(\mathcal{G}^\Sigma)$;
		\item there exists a functor $W_\Sigma \colon \egg_\Sigma\to \egg$ forgetting the labeling which creates all colimits, pullbacks and equalizers. 
	\end{enumerate}
\end{proposition}
}

Let $\pbe_\Sigma$ be the class of morphisms in $\egg_\Sigma$ whose image in $\egg$ lies in $\pbe$. Notice that $\pbe_\Sigma$ is also the class of arrows whose image through $Z_\Sigma$ is in $\pbc_\Sigma$.
By \Cref{thm:slice-functors} we get the following result.

\begin{corollary}\label{cor:eggade1}
	$\egg_\Sigma$ is $\pbe_\Sigma$-adhesive.
\end{corollary}

We turn now to term graphs with equivalence.

\begin{definition}\label[definintion]{def:GEqTGs}
	Given a signature $\Sigma$,  we say that an object $(\mathcal{H}, l)$ of $\EqTG_{\Sigma}$ is an \emph{e-term graph} if $\mathcal{H}$ is an e-hypergraph. We define the category $\eg$ as the full subcategory of $\EqTG_{\Sigma}$ given by e-term graphs and denote by $K_\Sigma$ the corresponding inclusion.
\end{definition}

\begin{remark}
	By definition we also have an inclusion functor $Y_\Sigma\colon \eg\to \egg_\Sigma$.
\end{remark}

\noindent 
\parbox{7.5cm}{\begin{remark}
	Now that we have put all the structures of this work in place, it is worthwhile to give a visual map of all the categories that we have discussed/introduced and of the relationships between them. We will use the curved arrows to denote full and faithful inclusions.
\end{remark}}\hfill \parbox{4cm}{\xymatrix@R=15pt@C=15pt{\eg \ar@{^{(}->}[r]^-{Y_\Sigma} \ar@{^{(}->}[d]_{K_\Sigma} & \egg_\Sigma \ar@{^{(}->}[d]^{Z_\Sigma} \ar[r]^{W_\Sigma}& \egg \ar@{^{(}->}[d]^{I}\\ \EqTG_{\Sigma} \ar@{^{(}->}[r]_{J_\Sigma}\ar[d]_{S_\Sigma}& \EqHyp_\Sigma \ar[r]_{V_\Sigma} \ar[d]_{T_\Sigma}& \EqHyp \ar[d]_{T}\\ \tg \ar@{^{(}->}[r]_{I_\Sigma}& \hyp_{\Sigma} \ar[r]_{U_\Sigma}  &\hyp}}

We can now prove our last three results. The reader can find the proof in \Cref{proof:tlim}.

\begin{restatable}{proposition}{limt}\label{prop:limt}
	$\eg$ has equalizers, binary products and pullbacks and they are created by $K_\Sigma$.
\end{restatable}

Let now $\mathcal{T}_\Sigma$ be the class of morphisms of $\eg$ which are sent by $K_\Sigma$ to the class $\mathcal{T}$. Then \Cref{colim,lem:epo,prop:po} allow us to deduce the following result.

\begin{restatable}{proposition}{pot}\label{prop:pot}
	$\eg$ has $\mathcal{T}_\Sigma$-pushouts, which are created by $K_\Sigma$.
\end{restatable}

Reasoning as in the previous sections, \Cref{prop:limt,prop:pot} now give us the following. 

\begin{corollary}
	$\eg$ is $\mathcal{T}_\Sigma$-adhesive.
\end{corollary}

\section{Pros and cons of adhesive rewriting}
\label{rewriting}
The previous sections have shown how hypergraphs and term graphs with equivalence
can be described as suitable $\mathcal{M}$-adhesive categories. 
The same fact holds for their sub-categories where the equivalence is
 closed with respect to operator composition, and this allows to
 model EGGs, as originally presented in~\cite{WillseyNWFTP21}. 

\emph{Sharing.} 
Terms are trees, thus different term graphs may represent the same term, up-to the sharing of sub-terms.
Consider e.g. a constant $a$ and a binary operator $/$: the term $a / a$ admits a few different 
representations as a term graph with equivalence, as shown below

%Being a purely graphical representation, EGGs allows for redundancy and possible ambiguity,
%so that e.g. a term such as $a / a$ admits a few different representations as a term graph, as shown below
%
%\[a \div a \mbox{ con 2 a sharate, con 2 a nella stessa classe di equivalenze, e con 2 a separate}\]

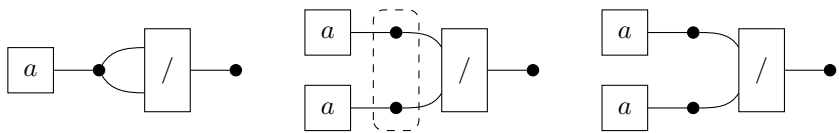
\begin{center}\begin{tikzpicture}[baseline=(w)]
        \begin{pgfonlayer}{nodelayer}
                \node[style=small box](a) at (-1.8, 0){$a$};
                \node[style=none] (aatt) at (-1.5, 0){};
                \node[style=node](v) at (-0.9, 0){};
                \node[style=medium box](div) at (0, 0){$/$};
                \node[style=none](divfst) at (-0.3, -0.3){};
                \node[style=none] (divsnd) at (-0.3, 0.3){};
                \node[style=none] (divout) at (0.3, 0){};
                \node[style=node] (w) at (0.9, 0){};
        \end{pgfonlayer}
        \begin{pgfonlayer}{edgelayer}
                \draw(aatt.center) to (v);
                \draw[out=-60, in=-180] (v) to (divfst.center);
                \draw[out=60, in=-180](v) to (divsnd.center);
                \draw(divout.center) to (w);
        \end{pgfonlayer}
\end{tikzpicture}
\qquad
\begin{tikzpicture}[baseline=(w)]
        \begin{pgfonlayer}{nodelayer}
                \node[style=small box](a1) at (-1.8, -0.5){$a$};
                \node[style=none] (a1att) at (-1.5, -0.5){};
                \node[style=node](v1) at (-0.9, -0.5){};
                \node[style=small box](a2) at (-1.8, 0.5){$a$};
                \node[style=none] (a2att) at (-1.5, 0.5){};
                \node[style=node](v2) at (-0.9, 0.5){};
                \node[style=medium box](div) at (0, 0){$/$};
                \node[style=none](divfst) at (-0.3, -0.3){};
                \node[style=none] (divsnd) at (-0.3, 0.3){};
                \node[style=none] (divout) at (0.3, 0){};
                \node[style=node] (w) at (0.9, 0){};
        \end{pgfonlayer}
        \begin{pgfonlayer}{edgelayer}
                \draw(a1att.center) to (v1);
                \draw(a2att.center) to (v2);
                \draw[out=0, in=-120] (v1) to (divfst.center);
                \draw[out=0, in=120](v2) to (divsnd.center);
                \draw(divout.center) to (w);
        \end{pgfonlayer}
        \begin{pgfonlayer}{eqlayer}
                \draw[dashed, rounded corners] (-1.2, -0.8) rectangle (-0.6, 0.8);
        \end{pgfonlayer}
\end{tikzpicture}
\qquad
\begin{tikzpicture}[baseline=(w)]
        \begin{pgfonlayer}{nodelayer}
                \node[style=small box](a1) at (-1.8, -0.5){$a$};
                \node[style=none] (a1att) at (-1.5, -0.5){};
                \node[style=node](v1) at (-0.9, -0.5){};
                \node[style=small box](a2) at (-1.8, 0.5){$a$};
                \node[style=none] (a2att) at (-1.5, 0.5){};
                \node[style=node](v2) at (-0.9, 0.5){};
                \node[style=medium box](div) at (0, 0){$/$};
                \node[style=none](divfst) at (-0.3, -0.3){};
                \node[style=none] (divsnd) at (-0.3, 0.3){};
                \node[style=none] (divout) at (0.3, 0){};
                \node[style=node] (w) at (0.9, 0){};
        \end{pgfonlayer}
        \begin{pgfonlayer}{edgelayer}
                \draw(a1att.center) to (v1);
                \draw(a2att.center) to (v2);
                \draw[out=0, in=-120] (v1) to (divfst.center);
                \draw[out=0, in=120](v2) to (divsnd.center);
                \draw(divout.center) to (w);
        \end{pgfonlayer}
\end{tikzpicture}
\end{center}
The left-most and the right-most images represent just ordinary term graphs, the middle one is a term graph with equivalence.
As such, they are all objects of $\EqTGs$.
However, note that the right-most image is not an object of $\eg$:
%$\GEqTGs$: 
constants have no input, 
and nodes that are targets of edges with the same label and whose source is the empty word must be equivalent
and possibly coincide,
as in the middle and in the left-most image, respectively.

In general, for (acyclic) term graphs, a maximally and a minimally shared representation
exist, and for our toy example are the left-most and the right-most diagram above, respectively:
it is a standard result for term graph rewriting, see e.g.~\cite{AriolaKP00}. 
These two representations can be interpreted as the final and the initial object of a suitable comma category.
%$\mathcal{G}  \downarrow \tg$, 
respectively.
The same properties carry on 
%when restricting to the full subcategory of acyclic term graphs.
%
for term graphs with equivalence and for EGGs. However, while in the former case the
two representations are the same of those for ordinary term graphs, in the latter case
the minimally shared representation is now the middle diagram.
%
%All the previous considerations can be lifted to the full sub-categories whose objects
%have underlying hypergraphs that are acyclic.
%%an EGG exist, and for our 
%%toy example are the left-most and the middle one above, respectively. In general terms,
%%for an EGG $\mathcal{G}$, such canonical forms are the final and the initial objects of the comma category 
%%$\mathcal{G}  \downarrow \eg$.
%%%\GEqTGs$. 
%%The former part is a standard result in term graph rewriting, see e.g.~\cite{xxx}, while the latter is specific 
%of the setting with equivalences.
%%
%The same construction carries on when restricting to the full subcategory of $\eg$
%%$\GEqTGs$ 
%of those objects whose 
%underlying term graph is \emph{acyclic}.
%%and in this case if $\mathcal{G}$ has no empty nodes, then the unique arrow 
%%from $\mathcal{G}$ to the terminal object is easily shown to be injective on equivalence classes.

\emph{Rewriting.} The theory of $\mathcal{M}$-adhesivity ensures that if the rules are spans of 
arrows in $\mathcal{M}$, then we can lift the standard properties
of the DPO approach that hold in the category of graphs.
% to the category at hand. 
However, as we recalled in the introduction, instead of removing sub-terms, the EGG approach 
chooses to
just add new terms and link them to the older ones via the equivalence relation~\cite{DetlefsNS05}.
% until an optimal program is 
%reached and extracted.
So, the corresponding DPO rules 
are spans $L \leftarrow L \rightarrow R$, where the first component is the identity, thus in $\pbc_\Sigma$,
while the second component may not belong to $\pbc_\Sigma$. 

Consider e.g. an EGG rule such $x / x \to 1$, from the introductory example in~\cite{WillseyNWFTP21}.
Variables in a term graph are represented as nodes that do not occur among the targets of an edge, 
thus the rule can be modelled as the DPO rule below, concisely given by the arrow $L \rightarrow R$

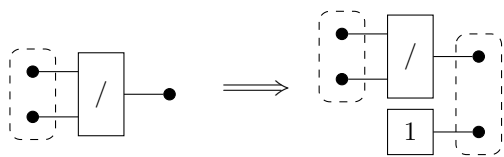
\begin{center}
\hspace{.25cm}
\xymatrix{        
 \begin{tikzpicture}[baseline=(w.base)]\begin{pgfonlayer}{nodelayer}
                \node[style=node](v1) at (-0.9, 0.3){};
                \node[style=node](v2) at (-0.9, -0.3){};
                \node[style=medium box] at (0, 0){$/$};
                \node[style=none](divfst) at (-0.3, 0.3){};
                \node[style=none](divsnd) at (-0.3, -0.3){};
                \node[style=none](divout) at (0.3, 0){};
                \node[style=node] (w) at (0.9, 0){};
        \end{pgfonlayer}        
        \begin{pgfonlayer}{edgelayer}
                \draw(v1) to (divfst.center);
                \draw(v2) to (divsnd.center);
                \draw(divout.center) to (w);
        \end{pgfonlayer}
        \begin{pgfonlayer}{eqlayer}
                \draw[dashed, rounded corners](-1.2, 0.6) rectangle (-0.6, -0.6);
                \end{pgfonlayer}\begin{pgfonlayer}{background}
                        \draw[color=white] (-1.5, 0.2) rectangle (1.5, -0.2);
                \end{pgfonlayer}
        \end{tikzpicture}
	\ar@{=>}[r] &
        \begin{tikzpicture}[baseline=(nil.center)]\begin{pgfonlayer}{nodelayer}
                \node[style=node](v1) at (-0.9, 0.3){};
                \node[style=node](v2) at (-0.9, -0.3){};
                \node[style=medium box] at (0, 0){$/$};
                \node[style=none](divfst) at (-0.3, 0.3){};
                \node[style=none](divsnd) at (-0.3, -0.3){};
                \node[style=none](divout) at (0.3, 0){};
                \node[style=node] (w) at (0.9, 0){};
                \node[style=small box] at (0, -1){$1$};
                \node[style=none](one) at (0.3, -1){};
                \node[style=node](z) at (0.9, -1){};
                \node[style=none](nil) at (0, -0.5){};
        \end{pgfonlayer}        
        \begin{pgfonlayer}{edgelayer}
                \draw(v1) to (divfst.center);
                \draw(v2) to (divsnd.center);
                \draw(divout.center) to (w);
                \draw(one.center) to (z);
        \end{pgfonlayer}
        \begin{pgfonlayer}{eqlayer}
                \draw[dashed, rounded corners](-1.2, 0.6) rectangle (-0.6, -0.6);
                \draw[dashed, rounded corners](0.6, 0.3) rectangle (1.2, -1.3);
        \end{pgfonlayer}\begin{pgfonlayer}{background}
                        \draw[color=white] (-1.5, 0.2) rectangle (1.5, -0.2);
        \end{pgfonlayer}\end{tikzpicture}
}
\end{center}

\noindent
In this case $L$ is the minimally shared EGG corresponding to the term $x / x$,
thus the rule could be applied to the two EGGs
depicted above with a match that is injective on equivalence classes.
In general, if the term graph underlying the EGG to whom the DPO rule is applied is acyclic,
the same property holds for the EGG obtained as the result of the rule application.
The right-hand side is a regular mono, though, hence it is not an arrow in $\pbc_\Sigma$.
The asymmetry between left-hand and right-hand sides falls in the current research about left-linear rules for 
adhesive categories, as pursued e.g. in~\cite{BaldanC0G24}.

\emph{Application conditions.} 
As argued above, the DPO rules that are suitable for the EGG approach have identities as left-hand side,
thus they belong to any choice of $\mathcal{M}$.
However, this would allow for the repeated application of the same rule, hence the 
possibility to keep on performing the same rewriting step.
This is forbidden by using rules \emph{with negative application conditions}, given as the usual span 
$L \leftarrow K \rightarrow R$ plus an additional arrow $n: L\rightarrow N$, such that a match $m: L \to G$ is admissible if it cannot
be factorised through $n$~\cite{HabelHT96}. For EGGs and rules $L \leftarrow L \rightarrow R$, it suffices to choose $n$ as the right-hand side itself.
The theory of $\mathcal{M}$-adhesivity carries on for rules with negative application conditions, see~\cite{ehrig2012,ehrig2014adhesive}.

\section{Conclusions and further and related works}
\label{conclusioni}
The aim of our paper was to extend the theory of $\mathcal{M}$-adhesive categories in order to include EGGs, 
a formalism for program optimisation. 
%via a compact representation and 
%efficient implementation of equality saturation.
%
To do so, we revisited the notions of hyper-graphs and term graphs with equivalence, 
proving that they are $\mathcal{M}$-adhesive categories, and we extended these results in order to
prove the same property for EGGs as term graphs with equivalence satisfying a suitable closure constraint.
Summing up, we proved that EGGs are objects of an $\mathcal{M}$-adhesive category $\eg$ 
and that optimisation steps are obtained via DPO rules (possibly with negative application conditions) whose 
left-hand side is in $\mathcal{M}$,
%, and in fact an identity. 
%and via a match that is in $\mathcal{N}$. 
and that allows for exploiting the properties 
of the $\mathcal{M}$-adhesive framework.
% when discussing about the optimisation process.

\emph{Future works.}
Our result on $\eg$ opens a few threads of research. The first is to check how the $\mathcal{M}$-adhesivity 
of EGGs can be pushed to model their rewriting via the double-pushout (DPO) approach. We have seen that 
the rules adopted in the literature of EGGs appears to be spans whose left-hand side is an identity and 
right-hand side is a regular mono (possibly with negative application 
conditions), and as such they fit the mould of rewriting on left-linear rewriting in $\mathcal{M}$-adhesive categories. 
%The consequences for usability on the restriction of the matches to arrows that are injective on equivalence classes
%though need to be further investigated.
%
However, it still needs to be shown how parallelism and causality, the key features for DPO rewriting
on $\mathcal{M}$-adhesive categories, can be exploited in the context of implementing the 
EGGs updates. 
Moreover, extensions of the EGGs formalism could be suggested by the adhesive 
machinery we developed.
In fact, most of the results
presented here for hypergraphs can be generalised to hierarchical hypergraphs, that is, 
hypergraphs with a hierarchy (a partial order) among 
edges~\cite{ghicaZan,CastelnovoGM24}.
The additional structure is useful for modelling properties such as encapsulation and sandboxing,
and it seems worthwhile to check the expressiveness and applications of hierarchical EGGS.

\emph{Related works.}
%\todo{espandere un poco}
Despite the interest they have ben raising as an efficient data structure, we are not aware of any attempt to provide an 
algebraic characterisation of EGGs, the only exception being~\cite{ghica}. We leave for future work an in-depth comparison 
among the two proposals, which appear to be related yet quite different.
In fact, for both proposals the key intuition is  to use string diagrams to represents term graphs. We adopted a
more down-to-earth approach by equipping
concretely the nodes of a term graph with an equivalence relation, thus directly corresponding to the original presentation~\cite{DetlefsNS05},
at the same time extending and generalising it~\cite{concur2006} in the contemporary jargon of adhesive categories.
The route chosen in~\cite{ghica} is to consider term graphs as arrows of a symmetric monoidal category,
and equipping them with an enrichment over semi-lattices, the resulting formalism being reminiscent of hierarchical hypergraphs.
Thus, the solution in~\cite{ghica}  is  more general than ours since it can be lifted to term graphs defined over categories other than
$\Set$.
At the same time, we both consider the DPO approach for rewriting, even if only our solution guarantees that the resulting category
is actually $\mathcal{M}$-adhesive, thus allowing to exploit all the features of the framework
as they hold for DPO graph transformation.

\newpage

\interlinepenalty 1000
\bibliography{biblio}

\newpage
\appendix

\section{Omitted proofs}

This section contains the proofs which are omitted from the main body of the paper. 
We begin recalling  a well-known fact about composition and decomposition of pullbacks \cite[Lem.~1.1]{lack2005adhesive}.

\noindent
\parbox{10cm}{
\begin{lemma}\label{lem:pb1}
	Let $\X$ be a category, and consider the diagram aside, in which the right square is a pullback. Then the whole rectangle is a pullback if and only if the left square is one.
\end{lemma}} 
	\parbox{4cm}{
	\xymatrix@R=15pt{X \ar[d]_{a} \ar[r]_{f}& \ar[r]_{g} Y \ar[d]^{b}& Z \ar[d]^{c}\\ A \ar[r]^{h}& B \ar[r]^{k}& C}}
\commentato{
\begin{proof}
	$(\Rightarrow)$ Let $q_1\colon Q\to Y$ and $q_2\colon Q\to A$ be two arrows such that $b\circ q_1=h\circ q_2$, if we compute we get
	\[c\circ g\circ q_1  =k\circ b\circ q_1=k\circ h\circ q_2\]
	\parbox{9cm}{Thus by the pullback property of the whole rectangle we get the dotted $l$ in the diagram on the side. 
	All we have to prove is that $f\circ l=q_1$. By construction we know that $g\circ f\circ l= g\circ q_1$, while we also have}
	\hfill \parbox{4cm}{\xymatrix@R=18pt{Q \ar@/_.3cm/[dr]_{q_2}\ar@/^.4cm/[rrr]^{g\circ q_1} \ar@{.>}[r]_{l} & X \ar[d]_{a} \ar[r]_{f}& \ar[r]_{g} Y \ar[d]^{b}& Z \ar[d]^{c}\\ & A \ar[r]_{h}& B \ar[r]_{k}& C}}
	\[	b\circ f\circ l = h\circ a\circ l= h\circ q_2=b\circ q_1\]
	and we can conclude since the right square in the original diagram is a pullback.
	
	For uniqueness: if $l'\colon Q\to X$ is such that $f\circ l'= q_1$ and $a\circ l'=q_2$
	then  $g\circ f \circ l' = g\circ q_1$ and we can conclude applying the pullback property of the outer rectangle.
	
	\smallskip \noindent 
	\parbox{9cm}{
	$(\Leftarrow)$ Take two arrows $q_1\colon Q\to Z$ and $q_2\colon Q\to A$ such that $c\circ q_1=k\circ h \circ q_2$
	We can apply the pullback property of the right square to get the dotted $q\colon Q\to Y$ in the following 

	Now, by construction we have $b\circ q =h\circ q_2$
	and thus, since the left square is a pullback, we get also a unique $l\colon Q\to X$ such that  $f\circ l = q$ and $a \circ l = q_2$
	but then we clearly have} \hfill 
	\parbox{6cm}{\vspace{7pt}
		\xymatrix@R=18pt{Q \ar@/_.3cm/[ddr]_{q_2}\ar@/^.5cm/[drrr]^{ q_1} \ar@{.>}@/^.3cm/[drr]^(.6){ q} \ar@{.>}[dr]^{l}  \\& X \ar[d]_{a} \ar[r]^{f}& \ar[r]^{g} Y \ar[d]_{b}& Z \ar[d]^{c}\\ & A \ar[r]_{h}& B \ar[r]_{k}& C}}
	\[	g\circ f \circ l = g\circ q = q_1	\] 
	
	We are left with uniqueness. Let $l'\colon Q\to X$ be another arrow such that $q_1 = g\circ f \circ l'$ and $q_2=a\circ l'$,  then we must also have
	\[
	b\circ f \circ l'= h\circ a \circ l' = h\circ q_2= b\circ q\]
	which implies $f\circ l' = q$, from which $l=l'$ follows.
\end{proof}
}

\subsection{Proofs for \Cref{sec:ade}}\label{app:primo}

\regmono*
\begin{proof}\label{regmono-proof}
	\begin{enumerate}
		\item  Consider the following cube in which the bottom face is an $\mathcal{M}$-pushout
		
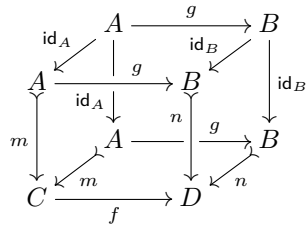
\[\xymatrix@C=15pt@R=9pt{&A\ar[dd]|\hole_(.65){\id{A}}\ar[rr]^{g} \ar[dl]_{\id{A}} && B \ar[dd]^{\id{B}} \ar[dl]_(.6){\id{B}} \\ A  \ar@{>->}[dd]_{m}\ar[rr]^(.65){g} & & B \ar@{>->}[dd]_(.3){n}\\&A\ar[rr]|\hole^(.65){g} \ar@{>->}[dl]^{m} && B \ar@{>->}[dl]^{n} \\C \ar[rr]_{f} & & D}\]
		By construction the top face of the cube is a pushout and the back one a pullback. The left face is a pullback because $m$ is mono, thus the Van Kampen property yields that the front and the right faces are pullbacks too and the claim follows.
		\item Let $m\colon X\mto Y$ be an arrow in $\mathcal{M}$, we can then take its pushout along itself, which, by the previous point, is also a pullback
		
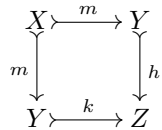
\[\xymatrix{X \ar@{>->}[r]^{m} \ar@{>->}[d]_{m}& Y \ar@{>->}[d]^{h}\\ Y \ar@{>->}[r]^{k} & Z}\]
		It is now immediate to see that $m$ is the equalizer of $h$ and $k$.   \qedhere
	\end{enumerate}
\end{proof}

\kpp*
\begin{proof}\label{kpp-proof}
	We begin by computing 
	\[g \circ h \circ \pi_f^1 =  t \circ f \circ \pi_f^1     =  t \circ f \circ \pi_f^2     =  g \circ h \circ \pi_f^2\]
	so that  existence  and uniqueness of the wanted $k_h$ follow at once from the universal property of $K_g$ as the pullback of $g$ along itself.

	\begin{enumerate}
		\item Let $a,b\colon T\rightrightarrows K_f$ be two arrows such that $k_h\circ a=k_h\circ b$, then we have
\begin{align*}h\circ \pi^1_f\circ a&=\pi^1_g\circ k_h\circ a=\pi^1_g\circ k_h\circ b=h\circ \pi^1_f\circ b\\
	h\circ \pi^2_f\circ a&=\pi^2_g\circ k_h\circ a=\pi^2_g\circ k_h\circ b=h\circ \pi^2_f\circ b
\end{align*}
Since $h$ is mono this entails that
\[	\pi^1_f\circ a=\pi^1_f\circ b \qquad \pi^2_f\circ a=\pi^2_f\circ b\]
and thus $a$ must coincide with $b$.
 
 		\item  
		To prove the second half of the claim, we can notice that, by \Cref{lem:pb1},  two rectangles below are pullbacks
		\[\xymatrix@R=15pt{K_f \ar[r]^{\pi^2_f} \ar[d]_{\pi^1_f} & X \ar[r]^{h} \ar[d]_{f} & Z \ar[d]^{g} &  K_f \ar[r]^{\pi^1_f} \ar[d]_{\pi^2_f} & X \ar[r]^{h} \ar[d]_{f} & Z \ar[d]^{g}\\
			X \ar[r]_{f}& Y \ar[r]_{t} & W & X \ar[r]_{f} & Y  \ar[r]_{t}& W}\]
		
		But then the outer rectangle in the  following diagrams are pullbacks too
		\[\xymatrix@R=15pt{	K_f  \ar@/^.4cm/[rr]^{h\circ\pi^2_f}\ar[r]_{k_h} \ar[d]_{\pi^1_f}& K_g \ar[d]_{\pi^1_g}  \ar[r]_{\pi^2_g}& Z \ar[d]_{g} & K_f  \ar@/^.4cm/[rr]^{h\circ\pi^2_f}\ar[r]_{k_h} \ar[d]_{\pi^1_f}& K_g \ar[d]_{\pi^1_g}  \ar[r]_{\pi^2_g}& Z \ar[d]_{g} \\
			X \ar[r]^{h} \ar@/_.4cm/[rr]_{t\circ f}& Y \ar[r]^{g} & W & X \ar[r]^{h} \ar@/_.4cm/[rr]_{t\circ f}& Y \ar[r]^{g} & W}\]
		
		Thus the left halves of the rectangle above are pullbacks again by \Cref{lem:pb1}.

		\item 		For the last square, let $(t_1, t_2)\colon T\to X\times X$ and $s\colon T\to K_g$ be two arrows such that
		\[(\pi^1_g, \pi^2_g)\circ s=(h\times h)\circ (t_1, t_2)\]
		
		Thus, by point $2$, there exist two arrows $x, y\colon T\rightrightarrows K_f$ fitting in the diagrams below
		\[\xymatrix@R=15pt{T\ar@{.>}[r]_-{x} \ar@/^.3cm/[rr]^{s} \ar@/_.3cm/[dr]_{t_1}& K_f \ar@{->}[r]_{k_h} \ar[d]_{\pi^1_f} & K_g \ar[d]^{\pi^1_g} & T \ar@{.>}[r]_-{y} \ar@/^.3cm/[rr]^{s} \ar@/_.3cm/[dr]_{t_2} & K_f  \ar@{->}[r]_{k_h} \ar[d]_{\pi^2_f}& K_g \ar[d]^{\pi^2_g}\\ & X \ar@{>->}[r]_{h} & Y &&X \ar@{>->}[r]_{h}& Y} \]
		By the first point $k_h$ is mono, thus we can deduce that $x=y$. By construction we have
		\[(\pi^1_f, \pi^2_f)\circ x=(t_1, t_2) \qquad k_h\circ x=s\] 
				
		Uniqueness now follows once again from the fact that $k_h$ is mono.
		\qedhere

	\end{enumerate}

\end{proof}

\noindent
\parbox{7cm}{\mpo*}\hfill \parbox{6cm}{\xymatrix@C=10pt@R=6pt{&A'\ar[dd]|\hole_(.65){a}\ar[rr]^{f'} \ar@{>->}[dl]_(.55){m'} && B' \ar[dd]^{b} \ar@{>->}[dl]_(.55){n'} & K_a\ar[rr]^{k_{f'}} \ar[dd]_{k_{m'}}&& K_b \ar[dd]^{k_{n'}} \\ C'  \ar[dd]_{c}\ar[rr]^(.7){g'} & & D' \ar[dd]_(.3){d}\\&A\ar[rr]|\hole^(.65){f} \ar[dl]^{m} && B \ar[dl]^{n}  & K_{c} \ar[rr]_{k_{g'}}&& K_d\\C \ar[rr]_{g} & & D }}

\noindent 
\parbox{4cm}{\xymatrix@C=10pt@R=6pt{&K_a\ar[dd]|\hole_(.65){\pi^1_a}\ar[rr]^{k_{f'}} \ar@{>->}[dl]_(.55){k_{m'}} && K_b\ar[dd]^{\pi^1_b} \ar@{>->}[dl]_(.55){k_{n'}} \\ K_c  \ar[dd]_{\pi^1_c}\ar[rr]^(.7){k_{g'}} & & K_d \ar[dd]_(.3){\pi^1_d}\\&A'\ar[rr]|\hole^(.65){f'} \ar@{>->}[dl]^(.4){m'} && B' \ar@{>->}[dl]^{n'}  \\C' \ar[rr]_{g'} & & D' }} \hfill \parbox{9cm}{\begin{proof}\label{mpo-proof}
	By  \Cref{prop:regmono}  we know that the top face of the original cube is a pullback. Thus \Cref{lemma:kern_pairs_pres_pullbacks} entails that in the following cube the vertical faces are pullbacks.
	The claim now follows from strict $\mathcal{M}$-adhesivity.
	\end{proof}}

\commentato{
\epic*
\noindent 
\parbox{10cm}{
\begin{proof}\label{epic-proof}
	By hypothesis, there exists a pair $f, g\colon Z \rightrightarrows X$ of which $e$ is the coequalizer.
	 Since $e \circ f = e \circ g$ we get the dotted arrow fitting in the digram aside.

	\parbox{\textwidth}{Let now $h: Z \to V$ be an arrow such that $h \circ \pi_1 = h \circ \pi_2$, then
	\[	h \circ f = h \circ \pi_1 \circ k 
	= h \circ \pi_2 \circ k 
	= h \circ g\]
	and thus there exists a unique $l: Y \to V$ such that $l \circ e = h$. \qedhere}
\end{proof}}\hfill
\parbox{9cm}{\vspace{-2cm}\xymatrix{Z \ar@{.>}[r]_k \ar@/_.3cm/[dr]_{g} \ar@/^.4cm/[rr]^{f} & K_e \ar[r]_{\pi^{1}_e}  \ar[d]_{\pi^2_e}& X \ar[d]^{e}\\ & X \ar[r]_{e} & Y}} 

\natepi*
\begin{proof}\label{natepi-proof}
	Let $(K_i, \pi_d^1, \pi_d^2)$ be the kernel pair of $\phi_d$ for each object $d$ in $\D$. Given an arrow $\alpha: d \to d'$ of $\D$, we have
	\[\phi_{d'} \circ F(\alpha) \circ \pi_d^1 = G(\alpha) \circ \phi_d \circ \pi_d^1 
	= G(\alpha) \circ \phi_d\circ \pi_d^2 
	= \phi_{d'} \circ F(\alpha) \circ \pi_d^2
	\]
	
	Thus, the solid part of the diagram below commutes, yielding the dotted arrow $K(\alpha)$.
	
	\noindent 
\parbox{5cm}{\xymatrix{K_d \ar[r]^{\pi^1_d} \ar[d]_{\pi^2_d} \ar@{.>}[dr]^{K(\alpha)}& F(d) \ar[dr]^{F(\alpha)} \\ F(d) \ar[dr]_{F(\alpha)} & K_{d'} \ar[r]^{\pi^1_{d'}} \ar[d]_{\pi^2_{d'}}& F(d') \ar[d]^{\phi_{d'}}\\ & F(d') \ar[r]_{\phi_{d'}} & G(d')}} \hfill \parbox{9cm}{In this way, we get a functor $K\colon \D\to \X$ mapping $d$ to $K_d$ and each arrow $\alpha$ onto $K(\alpha)$. Indeed, notice that $K(\id{d}) \colon  K_d \to K_d$ is the arrow such that
\begin{align*}
\pi_d^1 \circ K(\id{d})&=F(\id{d})\circ \pi^1_d=\id{F(d)}\circ \pi^{1}_d=\pi^1_d\\
\pi_d^2 \circ K(\id{d})&=F(\id{d})\circ \pi^2_d=\id{F(d)}\circ \pi^{2}_d=\pi^2_d
\end{align*}
Thus, by the universal property of pullbacks, $K(\id{d}) = \id{K_d}$.
}

	Let now $\alpha \colon a \to b$ and $\beta\colon b \to c$ be two arrows in $\D$,  computing, we have
	\begin{align*}
		\pi_c^1 \circ K(\beta \circ \alpha) &
		= F(\beta) \circ F(\alpha) \circ \pi_a^1 
		= F(\beta) \circ \pi_b^1 \circ K(\alpha) 
		= \pi_c^1 \circ K(\beta) \circ K(\alpha)\\
			\pi_c^2 \circ K(\beta \circ \alpha) &
		= F(\beta) \circ F(\alpha) \circ \pi_a^2 
		= F(\beta) \circ \pi_b^2 \circ K(\alpha) 
		= \pi_c^2 \circ K(\beta) \circ K(\alpha)
	\end{align*}
	Allowing us to conclude that $K(\beta \circ \alpha) = K(\beta) \circ K(\alpha)$, proving the functoriality of $K$.

	Hence,  by construction we have two natural transformations $\pi^1, \pi^2\colon  K \rightrightarrows F$. By \Cref{prop:reg_epi_coeq}, every component $\phi_d$ is the coequalizer of $\pi_d^1, \pi_d^2\colon K(d) \rightrightarrows F$, and so $\phi$ is the coequalizer of $\pi^1$ and $\pi^2$.   
\end{proof}

\epicol*

\begin{proof}\label{epicol-proof}
	By \Cref{cor:reg_epi_components_reg_epi_nat_trans}, we know that $\phi\colon F \to G$ is a regular epi, so that there is a functor $E\colon \D\to \X$ and $\eta, \theta\colon E \rightrightarrows F$ such that $\phi$ is the coequalizer of $\eta$ and $\theta$. Let now $(P, \{p_d\}_{d \in \D})$ be the colimit of $E$ and consider the unique arrows $a, b\colon P \rightrightarrows X$ fitting in the squares below
\medskip

\noindent 
	\parbox{4cm}{\xymatrix{E(d) \ar[r]^{p_d} \ar[d]_{\eta_d}& P \ar@{.>}[d]^{a} &E(d) \ar[d]_{\theta_d}
		 \ar[r]^{p_d}& P \ar@{.>}[d]^{b}\\ F(d) \ar[r]_{x_d} & X & F(d) \ar[r]_{x_d} & X}} \hfill	\parbox{8cm}{
	We want to show that $\phi$ coequalizes $a$ and $n$. Let thus $h\colon X \to Z$ be an arrow such that $h \circ a = h \circ b$. Then, for every $d$, 
	we have
	\[h \circ x_d \circ \eta_d = h \circ a \circ p_d 
	= h \circ b \circ p_d 
	= h \circ x_d \circ \theta_d\]}
	
	\noindent
	Thus, there is $h_d: G(d) \to Z$ such that $h\circ x_d = h_d \circ \phi_d$. It is now easy to see that $(Z, \{h_d\}_{d \in \D})$ is a cocone on $G$. Suppose $\alpha: d \to d'$ is an arrow of $\D$, then
	\[h_d\circ \phi_d=h\circ y_d=h\circ y_{d'}\circ F(\alpha)=h_{d'}\circ \phi_{d'}\circ F(\alpha)=h_{d'}\circ G(\alpha)\circ \phi_d\]
		By the hypothesis $\phi_d$ is regular epi for each $d$ and so we can conclude that $h_d = h_{d'} \circ G(\alpha)$.
	
	Therefore, we have an arrow $k\colon Y \to Z$ such that $k \circ y_d = h_d$. But then
	\[
	k \circ \phi \circ x_d = k \circ y_d \circ \phi_d = h_d \circ \phi_d = h\circ x_d\]
	Showing that $k \circ \phi = h$.
	
	For the uniqueness, let $k': Y \to Z$ be another arrow such that $k' \circ \phi = h$. Then we have
	\[k'\circ y_d\circ \phi_d=k'\circ \phi \circ x_d=h\circ x_d=h_d\circ \phi_d\]
	
	Since $\phi_d$ is a regular epi, we have $k' \circ y_d = h_d$ allowing us to conclude.
\end{proof}
}

To prove \Cref{prop:kerset} we need some results about pushouts and coproducts in $\Set$.

\noindent 
\parbox{11.7cm}{
\begin{lemma}\label{lem:po_set}
Suppose that the square aside is a pushout, with $m\colon A\mto B$ an injection. Let $\iota_{m}\colon E\mto B$ be the inclusion of the complement of the image of $m$. Then $(D, \{n, h\circ \iota_m \})$ is a coproduct. In particular, $h\circ \iota_m$ is mono.
\end{lemma}}\hfill \parbox{2cm}{\vspace{0cm}\xymatrix@R=15pt{A \ar@{>->}[r]^{m}  \ar[d]_{g}& B \ar[d]^{h}\\ C \ar@{>->}[r]_{n} & D}}
\begin{proof}
		 Let $f\colon E\to Z$ and $k\colon C\to Z$ be two arrows with the same codomain. By definition of complement and since $m$ is mono we know that $(B, \{m, \iota_m\})$ is a coproduct. Define $\phi\colon B\to Z$ as the unique arrow such that \[\phi\circ m=k\circ g \qquad \phi \circ \iota_m=f\] 
		By the universal property of pushouts there is a unique arrow $\psi \colon D\to Z$ such that $\psi \circ n=k$ and $\psi \circ h=\phi$, thus 
		\[
		\psi \circ h\circ \iota_m =\phi \circ \iota_m=f\]
		
To conclude we have to show that $\psi$ is the unique arrow $D\to Z$ such that $\psi \circ n =k$ and $\psi \circ h\circ \iota_m=f$. Let $\psi'\colon D\to Z$ be another arrow such that $\psi' \circ n=k$ and $\psi' \circ h\circ \iota_m=f$. Then we have
		\[\psi'\circ h\circ m=\psi'\circ n\circ g=k\circ g=\phi \circ m =\psi \circ h\circ m \]
	Since $m$ is mono then $\psi'\circ h=\psi \circ h$,  but $(D, \{n, h\})$ is a pushout and so $\psi'=\psi$. 
\end{proof}

The category of $\Set$ enjoys two other remarkable properties; it is \emph{distributive} and \emph{extensive} \cite{carboni1993introduction}. Distributivity amount to the following property: for every family $\{X_i\}_{i}$ and an object $Y$, the unique morphisms $\phi$ and $\psi $ fitting in the diagrams below, where $j_i$, $k_i$ and $h_i$ are coprojections, are isomorphisms
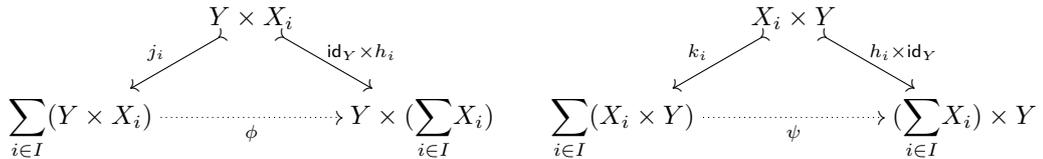
\[\xymatrix@C=15pt@R=15pt{&Y\times X_i \ar@{>->}[dl]_{j_i} \ar@{>->}[dr]^{\id{Y}\times h_i} &&& X_i \times Y \ar@{>->}[dl]_{k_i} \ar@{>->}[dr]^{ h_i \times \id{Y}} \\ \Sum_{i\in I} (Y\times X_i) \ar@{.>}[rr]_{\phi}&& Y\times (\Sum_{i\in I}X_i) & \Sum_{i\in I}  (X_i \times Y )\ar@{.>}[rr]_{\psi}&&  (\Sum_{i\in I}X_i) \times Y }\]

\noindent 
\parbox{11.5cm}{\hspace{15pt} Extensivity means that given a family of objects $\{X_{i}\}_{i\in I}$ and a family of commuting squares as the one on the right, where $j_i$ is a  coprojection, then all the squares are pullbacks if and only if $(Z, \{k_i\}_{i\in I})$ is a coproduct.}\hfill 
\parbox{2cm}{\xymatrix@R=15pt{Z_i \ar[r]_{k_i} \ar[d]_{f_i} & Z\ar[d]^{f} \\ X_i \ar@{>->}[r]^{j_i} & X}}

\noindent
\parbox{11.5cm}{
\begin{remark}\label{rem:disj}
	Notice that extensivity entails that coproducts are \emph{disjoint}, i.e.~the pullback between two coprojections is given by the initial object (with initial maps as coprojections). To see this, let $\{X, \{j_i\}_{i\in I}\}$ be the coproduct of the family $\{X_i\}_{i\in I}$ and $k$ an element of $I$. Then  $(X_k, \{t_i\}_{i\in I})$ 
such that $t_i = \id{X_k}$ if $i=k$, and $?_{X_k}$ otherwise,
	%\end{cases}\]	
	%with
	%\[t_i=\begin{cases}
	%	\id{X_k} & i=k\\?_{X_k} & i\neq k
	%\end{cases}\]
	is a coproduct and so the squares on the right are pullbacks.
\end{remark}}\hfill\parbox{2cm}
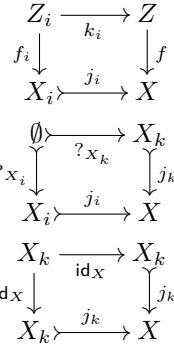
{\xymatrix@R=15pt{\emptyset \ar@{>->}[d]_{?_{X_i}} \ar@{>->}[r]_{?_{X_k}}& X_k \ar@{>->}[d]^{j_k} \\ X_i \ar@{>->}[r]^{j_i} & X} \xymatrix@R=15pt{X_k \ar[r]_{\id{X}} \ar[d]_{\id{X}}& X_k \ar@{>->}[d]^{j_k}\\X_k \ar@{>->}[r]^{j_k} & X}}

In particular, the previous remark yields at once the following fact.

\noindent 
\parbox{11.5cm}{\begin{proposition}\label{prop:isempty}
	Suppose that the square on the right is a pullback. Let $\iota \colon E\mto Y$ be the inclusion of the complement of the image of $p_1$. If there exist two  arrows $t_1\colon T\to E$ and $t_2\colon T\to X$ such that $g\circ \iota\circ t_1= f\circ t_2$ then $T$ is empty.
\end{proposition}}\hfill \parbox{2cm}{\vspace{-.5cm}\xymatrix@R=15pt{P \ar[r]_{p_1} \ar[d]_{p_2}& Y \ar[d]^{g}\\X \ar[r]^{f}& Z}}
\begin{proof}
	By the universal property of pullbacks we get an arrow $t\colon T\to P$ such that 
	\[p_1\circ t=\iota \circ t_1 \qquad p_2\circ t=t_2\]
	\parbox{2cm}
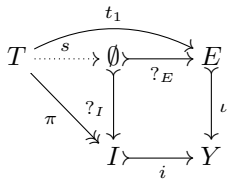
{\xymatrix{T \ar@{.>}[r]^{s}\ar@{>>}[dr]_{\pi} \ar@/^.4cm/[rr]^{t_1} & \emptyset \ar@{>->}[r]_{?_{E}} \ar@{>->}[d]_{?_I}& E \ar@{>->}[d]^{\iota}\\ & I \ar@{>->}[r]_{i} & Y}}\hfill \parbox{10.5cm}{Let $i\colon I\mto Y$ be the inclusion of the image of $p_1$, so that we have also an epi $q\colon P\eto I$. By definition of complement $(Y, \{i, \iota\})$ is a coproduct, thus by  \Cref{rem:disj} the inner square in the diagram on the left is a pullback and we get a dotted arrow $s\colon T\to \emptyset$ filling the diagram. But a set with an arrow towards the empty set must be empty and we conclude. \qedhere }
\end{proof}

Extensivity and distributivity, moreover, yield the following fact.

\begin{lemma}\label{lem:nodim}
	Let $f,g\colon A\rightrightarrows B+C$ be two arrows with a coproduct as codomain. Then $(A, \{j_i\}_{i=0}^
	3)$ is a coproduct, where $j_i\colon A_i\mto A $ is defined by the following pullback squares
	\[\xymatrix@C=30pt@R=15pt{A_0 \ar@{>->}[r]^{j_0}  \ar[d]_{p_0}& A   \ar[d]^{(f,g)}& A_1 \ar@{>->}[r]^{j_1}  \ar[d]_{p_1} & A \ar[d]^{(f,g)} \\
		B\times B \ar@{>->}[r]_-{i_B\times i_B} & (B+C)\times (B+C) & B\times C \ar@{>->}[r]_-{i_B\times i_C} & (B+C)\times (B+C)}\]
	\[ \xymatrix@C=30pt@R=15pt{A_2 \ar@{>->}[r]^{j_2}  \ar[d]_{p_2}& A   \ar[d]^{(f,g)}& A_3 \ar@{>->}[r]^{j_3}  \ar[d]_{p_3} & A \ar[d]^{(f,g)} \\
		C\times B \ar@{>->}[r]_-{i_C\times i_B} & (B+C)\times (B+C) & C\times C \ar@{>->}[r]_-{i_C\times i_C} & (B+C)\times (B+C)}\]
\end{lemma}

We are now ready to exploit the previous properties to prove \Cref{prop:kerset}. % as promised.

\noindent
\parbox{7.5cm}{\mps*}\hfill
\parbox{6cm}{\xymatrix@C=10pt@R=6pt{&A'\ar[dd]|\hole_(.65){a}\ar[rr]^{f'} \ar@{>->}[dl]_(.55){m'} && B' \ar[dd]^{b} \ar@{>->}[dl]_(.55){n'} & K_a\ar[rr]^{k_{f'}} \ar[dd]_{k_{m'}}&& K_b \ar[dd]^{k_{n'}} \\ C'  \ar[dd]_{c}\ar[rr]^(.7){g'} & & D' \ar[dd]_(.3){d}\\&A\ar[rr]|\hole^(.65){f} \ar@{>->}[dl]^{m} && B \ar@{>->}[dl]^{n}  & K_{c} \ar[rr]_{k_{g'}}&& K_d\\C \ar@{->}[rr]_{g} & & D }}

\begin{proof}\label{proof:kerset}
	We can start noticing that $m$, being the pullback of $n$ is mono too. This, in turn, entails that $m'$, being the pullback of $m$, is mono, and so $n'$ is injective too because in $\Set$, as in any adhesive category, monomorphisms are stable under pushouts.  
\vspace{.1cm}
	
	\noindent 
	\parbox{11.5cm}{\hspace{15pt}
	Let $\iota' \colon E'\mto C'$ and $\iota\colon E\to C$ be the inclusions of the complement of the images of of $m'$ and $m$ respectively . By \Cref{lem:po_set}, $(D', \{n', g'\circ \iota\})$ is a coproduct. }\hfill \parbox{4cm}{\vspace{-5pt}\xymatrix@R=15pt{E' \ar@{.>}[d]_{w} \ar@{>->}[r]_{\iota'} & C' \ar[d]^{c}\\ E \ar@{>->}[r]^{\iota} & C}}
	
\vspace{.1cm}

	\noindent
	\parbox{10.5cm}{\hspace{15pt}We can also notice that $c\circ \iota'$ factors through $E$, as shown by the square on the right.To see this, let $e$ be in $E'$ and suppose that  $c(\iota'(e))=m(x)$ for some $x\in A$. Then we can apply the  universal property of pullbacks to build the dotted arrow $v$ in the diagram aside, where $v_0$ and $v_1$ are the arrows picking $x$ and $\iota'(e)$, respectively. But then $e$ must belong to the image of $m'$, which is a contradiction.}\hfill  \parbox{4cm}{\xymatrix@R=15pt{1 \ar[dr]_{v_0}  \ar@/^.3cm/[rr]^{v_1} \ar@{.>}[r]_{v}& A' \ar[d]_{a} \ar@{>->}[r]_{m'} & C' \ar[d]^{c}\\& A  \ar@{>->}[r]^{m}& C}}

	\begin{enumerate}
		\item Consider two arrows $t_1\colon T\to D' $ and $t_2\colon T\to B$ such that $d\circ t_1=n\circ t_2$. By extensivity
		
		\noindent\parbox{9.5cm}{ of $\Set$, we already know that $(T, \{l_i\}_{i=0}^1)$ is a coproduct, where $l_i\colon T_i\mto T$ are defined by the diagram aside, whose two halves are pullbacks. }\hfill\parbox{4cm}{\xymatrix@R=15pt{T_0 \ar@{>->}[r]_{l_0} \ar[d]_{h_0} & T \ar[d]^{t_1} & T_1 \ar@{>->}[l]^{l_1} \ar[d]^{h_1}\\ E' \ar@{>->}[r]^{g'\circ \iota'} & D' & B' \ar@{>->}[l]_{n'}}}
		
		If we compute we get
		\[g\circ c \circ \iota' \circ h_0=d\circ g'\circ \iota' \circ h_0=d\circ t_1\circ l_0=n\circ t_2\circ l_0\]
		
		By hypothesis the bottom face of the given cube is a pullback, therefore there exists an arrow $u\colon T_0\to A$ such that 
		\[f\circ u=n\circ t_2\circ l_0 \qquad m\circ u=c\circ \iota'\circ h_0\]
		
	By \Cref{prop:isempty} we conclude that $T_0$ is empty. Therefore $l_1$ is an isomorphism and we have $n'\circ h_1\circ l^{-1}_1=t_1$. On the other hand
\[n\circ b\circ h_1\circ l^{-1}_1=d\circ n'\circ h_1\circ l^{-1}_1=d\circ t_1=n\circ t_2\]
And we can conclude that $b\circ h_1\circ l^{-1}_1=t_2$ because $n$ is a mono. The claim now follows from the fact that  also $n'$ is mono.
		
		\item 	 By \Cref{lem:po_set,lem:nodim}, the previous point and the third point of \Cref{lemma:kern_pairs_pres_pullbacks}, we can decompose $K_d$ as the coproduct of $K_b$ $K_1$, $K_2$ and $K_3$ with coprojections given by, respectively, $k_{n'}$ and $j_1$, $j_2$ and $j_3$, where these objects and arrows fit in the following four pullbacks
		\[\xymatrix@C=45pt@R=15pt{K_b \ar@{>->}[r]^{k_{n'}}  \ar[d]_{(\pi^1_b, \pi^2_b)}& K_d   \ar[d]^{(\pi^1_d, \pi^2_d)}& K_1 \ar@{>->}[r]^{j_1}  \ar[d]_{(p_1, q_1)} & K_d \ar[d]^{(\pi^1_d, \pi^2_d)} \\
			B'\times B' \ar@{>->}[r]_-{n'\times n'} & D'\times D' & B'\times E' \ar@{>->}[r]_-{n'\times (g'\circ \iota')} & D'\times D'}\]
		\[ \xymatrix@C=45pt@R=15pt{K_2 \ar@{>->}[r]^{j_2}  \ar[d]_{(p_2, q_2)}& K_d   \ar[d]^{(\pi^1_d, \pi^2_d)}& K_3 \ar@{>->}[r]^{j_3}  \ar[d]_{(p_3, q_3)} & K_d \ar[d]^{(\pi^1_d, \pi^2_d)} \\
			E'\times B' \ar@{>->}[r]_-{(g'\circ \iota') \times n'} & D'\times D' & E'\times E' \ar@{>->}[r]_-{(g'\circ \iota')\times (g'\circ \iota')} & D'\times D'}\]

Let us now examine $K_1$ and $K_2$. We have
\begin{align*}
	n\circ b\circ p_1&=d\circ n'\circ p_1=d\circ \pi^1_d\circ j_1=d\circ \pi^2_d\circ j_1=d\circ g'\circ \iota'\circ q_1=g\circ c \circ \iota'\circ q_1\\
	n\circ b\circ q_2 &=d\circ n'\circ p_2=d\circ\pi^2_D\circ j_2
	=d\circ \pi^1_d\circ j_2=d\circ g'\circ \iota'\circ p_2 = g\circ c \circ \iota'\circ p_2
\end{align*}

Since the bottom face of the original cube is a pullback (by hypothesis), we conclude that there exist arrows $k_1\colon  K_1\to A$ and $k_2\colon K_2\to A$ such that
\[m\circ k_1= c \iota'\circ q_1 \quad f\circ k_1=b\circ p_1 \quad m\circ k_2=c \iota'\circ q_2 \quad f\circ k_2=b\circ q_2\]

By \Cref{prop:isempty} we conclude that $K_1$ and $K_2$
are both empty. In particular, this implies that $(K_d, \{k_{n'}, j_3\})$ is a coproduct.

We focus now on $K_3$. By computing we get
	\begin{align*}
	&g\circ \iota \circ w\circ p_3=g\circ c\circ \iota'\circ p_3=d\circ g'\circ \iota'\circ p_3=d\circ \pi^1_d\circ j_3\\=&d\circ \pi^2_d\circ j_3=d\circ g'\circ \iota' \circ q_3=g\circ c\circ \iota'\circ q_3=g\circ \iota \circ w\circ q_3
	\end{align*}

\noindent 
\parbox{3cm}{\xymatrix@R=15pt{ K_3 \ar[r]^{p_3} \ar@{.>}[dr]^{\phi_3} \ar[d]_{q_3}& E' \ar[dr]^{\iota'} \\E' \ar[dr]_{\iota'}&K_c \ar[r]^{\pi^1_c} \ar[d]_{\pi^2_c} & C'  \ar[d]^{c}\\ & C' \ar[r]_{c} & C }}\hfill\parbox{9.5cm}{By \Cref{lem:po_set} $g\circ \iota$ is a monomorphism, thus $w\circ p_3=w\circ q_3$ and therefore
\[c\circ \iota'\circ p_3=\iota\circ w\circ p_3=\iota \circ w\circ q_3=c\circ \iota'\circ q_3\]
Thus we get the dotted $\phi_3\colon K_3\to K_c$ in the diagram aside.}

Moreover, $k_{g'}\circ \phi_3=j_3$ as shown by the following computation.
\begin{align*}
&(\pi^1_d, \pi^2_d)\circ k_{g'}\circ \phi_3=(g'\times g')\circ (\pi^1_c, \pi^2_c)\circ \phi_3=(g'\times g')\circ (\iota'\circ p_3, \iota'\circ q_3)\\=&((g'\circ \iota') \times (g'\circ \iota'))\circ (p_3, q_3)=(\pi^1_d, \pi^2_d)\circ j_3
\end{align*}
 
 \commentato{ 
 
 	\smallskip
 \parbox{7.5cm}{We can go further. The outer part of the diagram on the right commutes, so that we have the dotted $\psi_0\circ K_b\to K_0$.
 Now, on the one hand $\phi_0\circ \psi_0=\id{K_b}$ because if we compute we get
  \[(\pi^1_b, \pi^2_b)\circ \phi_0\circ \psi_0=(p_0, q_0)\circ \psi_0=(\pi^1_b, \pi^2_b)\]}\hfill \parbox{4cm}{\vspace{-.5cm}\xymatrix{K_b \ar[dr]_{(\pi^1_b, \pi^2_b)}  \ar@{>->}@/^.4cm/[rr]^{k_{n'}} \ar@{.>}[r]_{\psi_0}& K_0 \ar@{>->}[r]_{j_0}  \ar[d]^{(p_0, q_0)}& K_d   \ar[d]^{(\pi^1_d, \pi^2_d)} \\
 	& B'\times B' \ar@{>->}[r]_-{n'\times n'} & D'\times D' }
 	}

 On the other hand, notice that 
 \begin{align*}
 \pi^1_{d}\circ k_{n'}\circ \phi_0&=n'\circ \pi^1_b\circ \phi_0=n'\circ p_0=\pi^1_d\circ j_0 \\
 \pi^2_{d}\circ k_{n'}\circ \phi_0&=n'\circ \pi^2_b\circ \phi_0=n'\circ p_0=\pi^2_d\circ j_0
 \end{align*}

  Thus $k_{n'}\circ \phi_0=j_0$ and we have
 \[(p_0, q_0)\circ \psi_0\circ \phi_0=(\pi^1_b, \pi^2_b)\circ \phi_0=(p_0, q_0) \qquad  j_0\circ \psi_0\circ \phi_0 = k_{n'}\circ \phi_0=j_0\]
We can conclude that $\phi_0$ is an isomorphism, $(K_b, \{k_{n'}, (\pi^1_b, \pi^2_b)\})$ is a pullback for $(\pi^1_d, \pi^2_d)$ along $n'\times n'$ and  that$(K_d, \{k_{n'}, j_3\})$ is a coproduct.}
 
By definition of complement, we also know that $(C', \{\iota', m'\})$ is a coproduct. We can then apply again \Cref{lem:nodim}, decomposing $K_c$ in four parts, given by the pullbacks below
 \[\xymatrix@C=45pt@R=15pt{H_0 \ar@{>->}[r]^{x_0}  \ar[d]_{(y_0, z_0)}& K_c   \ar[d]^{(\pi^1_c, \pi^2_c)}& H_1 \ar@{>->}[r]^{x_1}  \ar[d]_{(y_1, z_1)} & K_c \ar[d]^{(\pi^1_c, \pi^2_c)} \\
 	A'\times A' \ar@{>->}[r]_-{m'\times m'} & C'\times C' & A'\times E' \ar@{>->}[r]_-{m'\times  \iota'} & C'\times C'}\]
 \[ \xymatrix@C=45pt@R=15pt{H_2 \ar@{>->}[r]^{x_2}  \ar[d]_{(y_2, z_2)}& K_c   \ar[d]^{(\pi^1_c, \pi^2_c)}& H_3 \ar@{>->}[r]^{x_3}  \ar[d]_{(y_3, z_3)} & K_c \ar[d]^{(\pi^1_c, \pi^2_c)} \\
 	E'\times A' \ar@{>->}[r]_-{\iota' \times m'} & C'\times C' & E'\times E' \ar@{>->}[r]_-{ \iota'\times  \iota'} & C'\times C'}\]
 
 \noindent
\parbox{7cm}{ Let us examine $H_0$. We start noticing that 
 \begin{align*}
 	(\pi^1_d, \pi^2_d)\circ k_{g'}\circ x_0=&(g'\times g')\circ(\pi^1_c, \pi^2_c)\circ x_0\\=&(g'\times g')\circ (m'\times m')\circ (y_0, z_0)\\=&((g'\circ m')\times (g'\circ m'))\circ (y_0, z_0)\\=&((n'\circ f')\times (n'\circ f'))\circ (y_0, z_0)\\=&(n'\times n')\circ (f'\times f')\circ (y_0, z_0) \end{align*}}\hfill \parbox{5cm}{\xymatrix@C=30pt{H_0 \ar[rr]^{(y_0, z_0)} \ar[d]_{x_0} \ar@{.>}[dr]^{h_0}&& A'\times A' \ar[d]_{f'\times f'}\\K_c \ar[dr]_{k_{g'}}& K_b \ar[r]_-{(\pi^1_b, \pi^2_b)} \ar[d]_{k_{n'}} & B'\times B' \ar[d]_{n'\times n'}\\ &K_d \ar[r]^-{(\pi^1_d, \pi^2_d)} & D'\times D'}}
 
 Thus we get the dotted arrow $h_0\colon H_0\to K_b$ in the diagram above. Moreover, we have
 \[m\circ a\circ y_0=c\circ m'\circ y_0=c\circ \pi^1_c\circ x_0=c\circ \pi^2_c\circ x_0=c\circ m'\circ z_0=m\circ a \circ z_0\]
 Since $m$ is a mono $a\circ y_0=a\circ z_0$ and there exists an arrow $h'_0\colon H_0\to K_a$ such that $\pi^1_a\circ h'_0=y_0$ and $\pi^2_a\circ h'_0=z_0$. 
 We now can conclude that $k_{f'}\circ h'_0=h_0$, noticing that
 \[(\pi^1_b, \pi^2_b)\circ h_0=(f'\times f')\circ (y_0, z_0)=(f'\times f')\circ (\pi^1_a, \pi^2_a)\circ h'_0=(\pi^1_b, \pi^2_b)\circ k_{f'}\circ h'_0\]

We can prove another property of $h'_0$. By computing we have
\[(\pi^1_c, \pi^2_c)\circ k_{m'}\circ h'_0=(m'\times m')\circ (\pi^1_a, \pi^2_a)\circ h'_0=(m'\times m')\circ (y_0, z_0)=(\pi^1_c, \pi^2_c)\circ x_0\]
so that $k_{m'}\circ h'_0$ must coincide with $x_0$.

As a next step, let us focus on $H_1$ and $H_2$. Two computations yield \small
 \[\begin{split}
(\pi^1_d, \pi^2_d)\circ k_{g'}\circ x_1&=(g'\times g')\circ (\pi^1_c, \pi^2_c)\circ x_1\\=&(g'\times g')\circ (m'\times \iota')\circ (y_1, z_1)\\=&((g'\circ m') \times (g'\circ \iota'))\circ (y_1, z_1)\\=&((n'\circ f') \times (g'\circ \iota'))\circ (y_1, z_1)\\=&(n'\times  (g'\circ \iota'))\circ (f'\circ y_1, z_1)
 \end{split} \hspace{1pt} \begin{split}
 (\pi^1_d, \pi^2_d)\circ k_{g'}\circ x_2&=(g'\times g')\circ (\pi^1_c, \pi^2_c)\circ x_2\\=&(g'\times g')\circ ( \iota' \times m')\circ (y_2, z_2)\\=&((g'\circ \iota') \times (g'\circ m'))\circ (y_2, z_2)\\=&( (g'\circ \iota') \times (n'\circ f'))\circ (y_2, z_2)\\=&((g'\circ \iota') \times n')\circ (y_2, f'\circ z_2)
 \end{split} \]
\normalsize
 Hence, we have arrows $h_1\colon H_1\to K_1$ and $h_2\colon H_2\to K_2$, showing that $H_1$ and $H_2$ are empty. For $H_3$, let us consider the two diagrams below
 \[ \xymatrix@C=45pt@R=15pt{ H_3 \ar@{.>}[d]^{\alpha_3} \ar[r]^{x_3} \ar@/_.4cm/[dd]_{(y_3, z_3)} & K_c \ar[d]_{k_{g'}}& K_3 \ar@{.>}[d]^{\beta_3} \ar@/^.3cm$ $/[dr]^{\phi_3} \ar@/_.4cm/[dd]_{(p_3, q_3)}\\K_3 \ar@{>->}[r]^{j_3}  \ar[d]^{(p_3, q_3)} & K_d \ar[d]^{(\pi^1_d, \pi^2_d)} &  H_3 \ar@{>->}[r]^{x_3}  \ar[d]^{(y_3, z_3)} & K_c \ar[d]^{(\pi^1_c, \pi^2_c)}  \\
   E'\times E' \ar@{>->}[r]_-{(g'\circ \iota')\times (g'\circ \iota')} & D'\times D' & E'\times E' \ar@{>->}[r]_-{ \iota'\times  \iota'} & C'\times C'}\]
 Their solid part commute. Indeed, we have
 \[\begin{split}
 	(\pi^1_d, \pi^2_d)\circ k_{g'}\circ x_3&=(g'\times g')\circ (\pi^1_c, \pi^2_c) \circ x_3\\=&(g'\times g')\circ (\iota'\times \iota')\circ (y_3\times z_3)\\=&(g'\circ \iota') \times (g'\circ \iota')\circ (y_3, z_3)
 \end{split}\qquad \begin{split}
 (\pi^1_c, \pi^2_c)\circ \phi_3=&(\iota'\circ p_3, \iota' \circ q_3)\\=&
(\iota'\times \iota')\circ (p_3, q_3)\\& \end{split}\] 
Thus we get the dotted arrows $\alpha_3$ and $\beta_3$ which are one the inverse of the other. Indeed, on the one hand we have
\begin{align*}
j_3\circ \alpha_3\circ \beta_3=k_{g'}\circ x_3\circ \beta_3=k_{g'} \circ \phi_3=j_3  \qquad 
(p_3, q_3)\circ \alpha_3\circ \beta_3=(y_3, z_3)\circ \beta_3=(p_3, q_3) 
\end{align*}
 On the other hand, notice that
 \[(\pi^1_c, \pi^2_c)\circ \phi_3\circ \alpha_3 =(\iota'\times \iota')\circ (p_3, q_3)\circ \alpha_3= (\iota'\times \iota')\circ (y_3, z_3)=(\pi^1_c, \pi^2_c)\circ x_3\]
 Therefore $\phi_3\circ \alpha_3=x_3$ and thus $x_3\circ \beta_3\circ \alpha_3=x_3$. Moreover
 \[(y_3, z_3)\circ \beta_3\circ \alpha_3=(p_3, q_3)\circ \alpha_3=(y_3, z_3)\]
 Summing up, we have just proved that $(K_c, \{x_0, \phi_3\})$ is a coproduct.
 
 Let now $\gamma\colon K_b\to Z$ and $\delta\colon K_c\to Z$ be two arrows such that $\gamma\circ k_{f'}=\delta\circ k_{m'}$. We want to construct an arrow $\theta\colon K_d\to Z$ such that $\theta\circ k_{g'}=\delta$ and $\theta\circ k_{n'}=\gamma$.	
 
 We have already proved that $(K_d, \{k_n, j_3\})$ is a coproduct, thus there is a unique arrow $\theta\colon K_d\to Z$ such that  $\theta \circ k_{n'}=\gamma$ and $\theta \circ j_3=\delta \circ \phi_3 $.
 Now, on the one hand we have
 \begin{align*}
 	\theta \circ k_{g'}\circ \phi_3=\theta \circ  j_3=\delta\circ \phi_3 
 \end{align*}
 On the other hand, we can conclude that $\theta \circ k_{g'}=\delta$ as wanted, since
 \[\theta \circ k_{g'}\circ x_0=\theta \circ k_{n'}\circ h_0=\gamma \circ h_0=\gamma \circ k_{f'}\circ h'_0=\delta \circ k_{m'}\circ h'_0=\delta \circ x_0 \]
 We are left with uniqueness. If $\theta'$ is another arrow $K_d\to Z$ such that $\theta'\circ k_{n'}=\gamma$ and $\theta'\circ k_{g'}= \delta$. If we compute we get
 \[\theta'\circ j_3=\theta'\circ k_{g'}\circ \phi_3=\delta \circ \phi_3=\theta \circ j_3 \]
 We already know that $\theta'\circ k_{n'}=\theta \circ k_{n'}$, so the previous identity entails that $\theta =\theta'$.
			 \qedhere 
	\end{enumerate}
\end{proof}

\newpage
\subsection{Proofs for \Cref{hypereq}}\label{app:due}
\fhyp*

\begin{proof}\label{proof:forghyp}
	\begin{enumerate}
		\item  This follows at once from \Cref{rem:eqhyp_morphs}.
		\item  	Let $\mathcal{H}$ be a hypergraph, and define $L(\mathcal{H}) := (E_\mathcal{H}, V_{\mathcal{H}}, V_{\mathcal{H}}, s_\mathcal{H}, t_\mathcal{H}, \id{V_\mathcal{H}})$.  By construction we have $T(L(\mathcal{H}))=\mathcal{H}$, thus we can define
		$\eta_\mathcal{H}\colon \mathcal{H}\to T(L(\mathcal{H}))$ as the identity $\id{\mathcal{H}}$. 
		
		To see that in this way we get a unit,  take an arrow $(h, k)\colon \mathcal{H}\to T(\mathcal{G})$ for some $\mathcal{G}$ in $\EqHyp$. Then $(h,k, q_{\mathcal{G}}\circ k)$ is an arrow $L(\mathcal{H})\to \mathcal{G}$ and the unique one such that $T(h,k, q_{\mathcal{G}}\circ k)\circ \eta_{\mathcal{H}}=(h,k)$
		
		\item For every hypergraph $\mathcal{H}$ define $R(\mathcal{H})$ as $(E_\mathcal{H}, V_{\mathcal{H}}, 1, s_\mathcal{H}, t_\mathcal{H}, !_{V_\mathcal{H}})$, so that $T(R(\mathcal{H}))$ is again $\mathcal{H}$.  Now, let $\epsilon_{\mathcal{H}}\colon T(R(\mathcal{H}))\to \mathcal{H}$ be the identity and take an arrow $(h,k)\colon T(\mathcal{G})\to \mathcal{H}$ for some $\mathcal{G}$ in $\EqHyp$.  Notice that $!_{Q_\mathcal{G}}\circ q_{\mathcal{G}}=!_{V_{\mathcal{H}}}\circ k$
		so that $(h,k, !_{Q_\mathcal{G}})$ is an arrow $\mathcal{G}\to R(\mathcal{H})$ of $\EqHyp$ such that $\epsilon_{\mathcal{H}}\circ T(h,k, !_{Q_\mathcal{G}})=(h,k)$.
		
		Uniqueness of such an arrow follows at once from the first point and the fact that $1$ is terminal. \qedhere 
	\end{enumerate}
\end{proof}

\begin{remark}\label{rem:fact}
	Before proceeding further, let us recall the following result about the classes of regular epis (i.e.~surjections) and of monos in $\Set$. In particular, we need the fact that they form a a \emph{factorization system} \cite{kelly2006note} on $\Set$.  This amounts to ask that
	
	\noindent 
	\parbox{11.5cm}{
	\begin{enumerate} 
		\item every arrow $f\colon X\to Y$ can be factored as $m\circ e$, where $e\colon X\eto Q$ is a regular epi and $m\colon Q\mto Y$ is a mono;
			\item for every commuting square as the one on the right, with $e\colon X\eto E$ a surjection and $m\colon N\mto Y$ a mono, there exists a unique $k\colon E\to N$ making it commutative.
	\end{enumerate}}\hfill \parbox{2cm}{\xymatrix{X \ar@{>>}[d]_{e}\ar[r]^{f}& N \ar@{>->}[d]^{m} \\ E \ar@{.>}[ur]^{k}\ar[r]_{g} & Y}}
	
	\hspace{1pt}\newline 
	From these properties, one can deduce that if $f=m\circ e$ and $f=m'\circ e'$ are two factorizations of $f$ then there is a bijection $\phi$ such that $m=m'\circ \phi$ and $\phi \circ e=e'$.
\end{remark}

We will also need the following well-known fact about regular epis.

\begin{lemma}\label{lemma:nat_trans_reg_epi_canonical_arrow_reg_epi}
	Let $F,G\colon \D\rightrightarrows \X$ be two diagrams and suppose that $\X$ has all colimits of shape $\D$. Let $(X, \{x_d\}_{d \in \D})$ and $(Y, \{y_d\}_{d\in D})$ be the colimits of $F$ and $G$, respectively.  If $\phi\colon  F \to G$ is a natural transformation whose components are regular epis, then the arrow induced by $\phi$ from $X$ to $Y$ is a regular epi.
\end{lemma}

\noindent
\parbox{11.4cm}{
\comp*}\hfill 
\parbox{2cm}{\xymatrix{V \ar@{>>}[d]_{q} \ar[r]^{\pi^D_V}& V_D \ar@{>>}[dd]^{q_D}\\Q \ar@{>.>}[d]_{m}&\\ L \ar[r]_{l_d} & Q_D}}   

\begin{proof}\label{proof:comp}
	\begin{enumerate}
		\item We know by \Cref{prop:cocomp} that $\hyp$ is cocomplete. Thus, let $(E, V, s, t)$ together with $\{(\kappa_E^D, \kappa_V^D)\}_{D \in \D}$ be a colimit for $T \circ F$. By \Cref{colim} we also know that 
		 
		 \noindent 
		 \parbox{10.5cm}{$(V, \{\kappa^D_V\}_{D\in \D})$ is a colimit for $U_\eq\circ F$. Moreover, since $\Set$ is cocomplete too we can also take a colimit $(C, \{c_d\}_{d\in \D})$ for $K\circ F$. Now,  let $\alpha$ be an arrow $D\to D'$ in $\D$, and suppose that $F(\alpha)$ is $(h_1, h_2, h_3)$. By definition of morphisms in $\EqHyp$, the square on the right commutes.}
		\hfill 
		\parbox{2cm}{\xymatrix{V_d \ar[r]_{h_2} \ar@{>>}[d]_{q_D} & V_{D'} \ar@{>>}[d]^{q_{D'}}\\ Q_{D} \ar[r]^{h_3} & Q_{D'}}}
		\vspace{.05cm}
		
		Thus the family $\{q_D\}_{D\in \D}$ form a natural transformation $U_\eq\circ F\to K\circ F$.  By \Cref{lemma:nat_trans_reg_epi_canonical_arrow_reg_epi}, the induced arrow $q: V \to C$ between the colimits is a surjection.  We can then consider the object $(E, V, C, s, t, q)$ of $\EqHyp$, together with the family  $\{(\kappa_E^D, \kappa_V^D, c_D)\}_{D\in \D}$, which by construction is a cocone on $F$. Let $((E_\mathcal{G},V_\mathcal{G}, Q_{\mathcal{G}}, s_{\mathcal{G}}, t_{\mathcal{G}}, q_{\mathcal{G}}), \{(h^D_E, h^D_V, h^D_Q)\}_{D\in \D})$ be a cocone, then $((E_\mathcal{G},V_\mathcal{G}, s_{\mathcal{G}}, t_{\mathcal{G}}), \{(h^D_E, h^D_V)\}_{D\in \D})$ and $(Q_{\mathcal{G}}, \{h^D_Q\}_{D\in \D})$ are cocone on, respectively, $T\circ F$ and $K\circ D$, giving arrows $(h_E, h_V)$ in $\hyp$ and $h_Q$ in $\Set$ such that
		\[(h_E, h_V)\circ (\kappa_E^D, \kappa_V^D)=(h^D_E, h^D_V) \qquad h_Q\circ c_D=h^D_Q\]
		Now, to show that $(h_E, h_V, h_Q)$ is an arrow of $\EqHyp$ we can compute
		\[h_Q\circ q\circ \kappa^D_V=h_Q\circ c_D\circ q_D=h^D_Q\circ q_D=q_{\mathcal{G}} \circ h^D_V=q_{\mathcal{G}}\circ h_V\circ \kappa^D_V \]
		
		Uniqueness of such arrow follows form the fact that $T$ is faithful and \Cref{rem:eqhyp_morphs}.
		
		\item   $\hyp$ is complete, again by \Cref{prop:cocomp}, we can then consider a limiting cone $((E, V, s, t), \{(\pi^D_E, \pi^D_V)\}_{D\in \D})$ over of $T \circ F$. Now, $(V, \{q_D\circ \pi_V^D\}_{D\in \D})$, is a cone for $K \circ F$: indeed, if $\alpha\colon D \to D'$ is an arrow of $\D$, and $F(\alpha)=(h_1, h_2, h_3)$, then we have
	\[h_3\circ q_D\circ \pi^D_V =q_{D'} \circ h_2\circ \pi^D_{V}=q_{D'}\circ \pi^{D'}_V\]

\noindent 
\parbox{9.5cm}{Thus there is an arrow $l\colon V\to L$ such that $l_D\circ l=q_D\circ \pi^D_V$. By \Cref{rem:fact}, we know that there exist $m\colon Q\mto L$ and $q\colon X\eto Q $ such that $m\circ q=l$. Since the identity is mono, \Cref{rem:fact} yield a unique arrow $\pi^D_Q$ fitting in the rectangle aside.} 	\hfill \parbox{3cm}{\xymatrix{
	V\ar[r]^{\pi_V^D}\ar@{>>}[d]_{q}&V_D\ar@{>>}[r]^{q_D}&Q_D\ar[d]^{\id{Q_D}}\\Q\ar@{>->}[r]_{m}\ar@{.>}[urr]_{\pi_Q^D}&L\ar[r]_{l_D}&{Q_D}
	}}
	
	Let $\alpha$ be an arrow $D\to D'$ in $\D$. Then we have
		\[T(F\alpha\circ (\pi^D_{E}, \pi^D_V, \pi^D_Q))=T(F(\alpha))\circ (\pi^D_E, \pi^D_V)=(\pi^{D'}_E, \pi^{D'}_V)=T(\pi^{D'}_{E}, \pi^{D'}_V, \pi^{D'}_Q)\]

		Thus, by faitfhulness of $T$, $((E, V, C, s, t, q), \{(\pi_E^D, \pi_V^D, \pi_Q^D)\}_{D\in \D})$ is a cone over $F$. To see that it is terminal, let $((E_\mathcal{G},V_\mathcal{G}, Q_{\mathcal{G}}, s_{\mathcal{G}}, t_{\mathcal{G}}, q_{\mathcal{G}}), \{(h^D_E, h^D_V, h^D_Q)\}_{D\in \D})$ be another cone. In particular there is an arrow $(h_E, h_V)\colon (E_\mathcal{G},V_\mathcal{G},  s_{\mathcal{G}}, t_{\mathcal{G}})\to (E, V, s, t)$ in $\hyp$ such that
		\[\pi^D_E\circ h_E=h^D_E \qquad \pi^D_V\circ h_V=h^D_V\]

		Moreover, applying $K$ we get that $(Q_{\mathcal{G}}, \{h^D_{Q}\}_{D\in \D})$ is a cone over $K\circ F$, so that there is an arrow $h\colon Q_{\mathcal{G}} \to L$ such that $l_D\circ h=h^D_{Q}$.  Thus
		\[l_D\circ h\circ q_\mathcal{G}=h^D_Q\circ q_{\mathcal{G}}=q_{D}\circ h^D_V=q_D\circ \pi^D\circ h_V=l_D\circ l\circ h_V\]

\noindent 
\parbox{2cm}{\xymatrix@R=15pt{V_\mathcal{G} \ar@{>>}[d]_{q_{\mathcal{G}}} \ar[r]^{h_V}& V \ar@{>>}[r]^{q}  & Q \ar@{>->}[d]^{m}\\ Q_{\mathcal{G}}  \ar@{.>}[urr]_{k_Q}\ar[rr]_{h} && L }}\hfill \parbox{9.5cm}{Hence the  outer boundary of the square on the left commutes, yielding the dotted diagonal arrow $k_Q\colon Q_\mathcal{G}\to Q$.  We have therefore built an arrow $(h_E, h_V, k_Q)$ of $\EqHyp$ such that
\[(\pi^D_E, \pi^D_V, \pi^D_Q)\circ (h_E, h_V, k_Q)=(h^D_E, h^D_V, h^D_Q)\]}

As in the point above, uniqueness is guaranteed by faithfulness of $T$ and \Cref{rem:eqhyp_morphs}.	\qedhere 
	\end{enumerate}
\end{proof}

\mn*

\begin{proof}\label{proof:regmono}

$(\Rightarrow)$	If $(h_E, h_V, h_Q)$ is mono, from \Cref{cor:mono1} we have that $h_E$ and $h_V$ are monos.

\noindent
	\parbox{11cm}{
		\hspace{15pt}Let now $(f_E, f_V, f_Q), (g_E, g_V, g_Q): \mathcal{H\rightrightarrows K}$ be arrows equalized by $(h_E, h_V, h_Q)$. Let $e\colon E_Q\mto Q_{\mathcal{H}}$ be the equalizer in
	$\Set$ of $f_Q, g_Q\colon Q_{\mathcal{H}}\rightrightarrows Q_{\mathcal{K}}$. From the third point of \Cref{prop:fact} and \Cref{lim} we know that $h_V\colon V_{\mathcal{G}}\mto V_{\mathcal{H}} $ 
	is the equalizer of $f_V, g_V\colon V_{\mathcal{H}}\rightrightarrows V_{\mathcal{K}}$, thus by the second point of \Cref{prop:eqhyp_complete} we know that there is a mono 
	$m\colon Q_{\mathcal{G}}\mto E_Q$ fitting in the diagram aside. Now, notice that 
	$e\circ m\circ q_{\mathcal{G}}=q_\mathcal{H}\circ h_V=h_{Q}\circ q_{\mathcal{G}}$.
	Since $q_{\mathcal{G}}$ is epi we conclude that $e\circ m=h_Q$, from which the claim follows.}
	\hfill
	\parbox{2cm}{\xymatrix@R=18pt{V_\mathcal{G} \ar@{>>}[d]_{q_{\mathcal{G}}} \ar@{>->}[r]^{h_V}& V_\mathcal{H} \ar@{>>}[dd]^{q_{\mathcal{H}}}\\Q_\mathcal{G} \ar@{>.>}[d]_{m}&\\ E_Q \ar@{>->}[r]_{e} & Q_\mathcal{{H}}}}   
	
	$(\Leftarrow)$ Suppose that $h_E$, $h_V$, and $h_Q$  are monos. 
	We can build the cokernel pair of the three arrows, taking the pushout of $h_E$, $h_V$ and $h_Q$ along itself, obtaining the three diagrams below, which by \Cref{prop:regmono} are also pullbacks
	\[     \xymatrix@R=15pt{
	E_\mathcal{G} \ar@{>->}[d]_{h_E} \ar@{>->}[r]^{h_E}&E_\mathcal{H} \ar@{>->}[d]^{g_E} &	V_\mathcal{G} \ar@{>->}[d]_{h_V} \ar@{>->}[r]^{h_V}& V_\mathcal{H} \ar@{>->}[d]^{g_V} &Q_\mathcal{G} \ar@{>->}[d]_{h_Q} \ar@{>->}[r]^{h_Q}& Q_\mathcal{H} \ar@{>->}[d]^{g_Q} \\
		E_\mathcal{H} \ar@{>->}[r]_{f_E} & E_\mathcal{{K}} & V_H \ar@{>->}[r]_{f_V} & V_\mathcal{{K}} & Q_\mathcal{H} \ar@{>->}[r]_{f_Q} & Q_\mathcal{{K}}
	}
	\]   
	
	Now, we know by \Cref{prop:com} and \Cref{colim} that there exists arrows $s,t\colon E_{\mathcal{K}}\rightrightarrows V_{\mathcal{K}}$ such that the resulting hypergraph is the cokernel pair of $(h_E, h_V)$ in $\hyp$. Moreover, by \Cref{lemma:nat_trans_reg_epi_canonical_arrow_reg_epi} we know that the unique arrow $q\colon V_{\mathcal{K}}\to Q_{\mathcal{K}}$
	induced by $q_{Q}\circ q_{\mathcal{H}}$ and $f_{Q}\circ q_{\mathcal{H}}$ is a regular epi. Let thus $\mathcal{K}$ be the resulting hypergraph with equivalence.

	By construction we have two arrows 	$(f_E, f_V, f_Q), (g_E, g_V, g_Q): \mathcal{H\rightrightarrows K}$ such that 
	\[(f_E, f_V, f_Q)\circ (h_E, h_V, h_Q) =(g_E, g_V, g_Q)\circ (h_E, h_V, h_Q)\]
	
	Now, in $\Set$ every mono is regular \cite{maclane2012sheaves}, so by the dual of \Cref{prop:reg_epi_coeq} we deduce that $h_E$, $h_V$ and $h_Q$ are the equalizers of the arrows built above. By \Cref{prop:com,lim} we also know that $(h_E, h_V)$ is the equalizer in $\hyp$ of $(f_E, f_V)$ and $(g_E, g_V)$.  Hence, if $(z_E, z_V, z_Q)\colon \mathcal{Z}\to \mathcal{H}$ is an arrow such that 
			\[(f_E, f_V, f_Q)\circ (z_E, z_V, z_Q) =(g_E, g_V, g_Q)\circ (z_E, z_V, z_Q)\]
	then we get a unique morphism $(a_E, a_V)\colon (E_{\mathcal{Z}}, V_{\mathcal{Z}})\to (E_{\mathcal{G}}, V_{\mathcal{G}})$ such that $(h_E, h_V)\circ (a_E, a_V) =(z_E, z_V)$. Moreover, since $g_Q\circ z_Q=f_Q\circ z_Q$ there is a unique $a_Q\colon Q_{\mathcal{Z}}\to Q_\mathcal{G}$ such that $h_Q\circ a_Q=z_Q$. Thus
	\[h_Q\circ a_Q\circ q_{\mathcal{Z}}=z_Q\circ q_{\mathcal{Z}}=q_{\mathcal{H}}\circ z_{V}=q_{\mathcal{H}}\circ h_V\circ a_V=h_Q\circ q_{\mathcal{G}}\circ a_V\]
	Since $h_Q$ is mono we conclude that $a_Q\circ q_{\mathcal{Z}}=q_{\mathcal{G}}\circ a_V$ and we conclude.
\end{proof}

\commentato{
\quot*
\begin{proof}\label{proof:quot}
  Given a hypergraph $\mathcal{H}:=(E, V, s, t)$, let us define $R(E, V, s, t)$ as $(E, V, V, s, t, \id{B})$. By construction $\quo(R(\mathcal{{H}})) =\mathcal{H}$ so that we can take the identity as $\epsilon_{\mathcal{H}}\colon \quo(R(\mathcal{H}))\to \mathcal{H} $.
		
		Suppose that $\mathcal{G}$ is an object of $\EqHyp$ and take an arrow $(h_E, h_V): \quo(\mathcal{G}) \to \mathcal{H}$, then the three squares below commute, so that $(h_E,  h_V \circ q_{\mathcal{G}}, h_V)$ is a morphism of $\EqHyp$.
		\[\xymatrix@R=15pt@C=30pt{
			E_{\mathcal{G}} \ar[r]^{h_E}\ar[d]_{s_{\mathcal{G}}} & E \ar[d]^{s} & E_\mathcal{G} \ar[r]^{h_E}\ar[d]_{t_{\mathcal{G}}} & V \ar[d]^{t} & V_{\mathcal{G}} \ar[r]^{h_V \circ q_{\mathcal{G}}} \ar@{>>}[d]_{q_{\mathcal{G}}} & V \ar@{>>}[d]^{\id{V}} \\
			V_{\mathcal{G}}^{\star} \ar[r]_-{h_V^\star\circ q_{\mathcal{G}}^\star} & V^{\star} & V_{\mathcal{G}}^{\star} \ar[r]_{h_V^\star\circ q_\mathcal{G}^\star} & V^{\star} & C \ar[r]_{h_V} & V
		}\]
		
		By construction $\epsilon_{\mathcal{H}}\circ\quo(h_E,  h_V\circ q, h_V)=(h_E, h_V)$. To see that such an arrow is unique, let $(f_E, f_V, f_Q)$ be a morphism such that $\epsilon_{\mathcal{H}}\circ\quo(f_E, f_V, f_Q)=(h_E, h_V)$, then $f_E=h_E$ and $f_Q=h_V$, moreover
		\[f_V=\id{V}\circ f_V=f_Q\circ q_{\mathcal{G}}=h_{V}\circ q_\mathcal{G}  \]
	allowing us to conclude.		
\end{proof}
\commentato{		\item 	Since $\quo$ is a left adjoint, it preserves colimits.
		Let $F\colon \D \to \EqHyp$ be a diagram, to fix the notation, let  $\mathcal{C}$ be  $(E, V, s, t)$,  and $F(d)$ be $(E_d, V_d, Q_d, s_d, t_d, q_d)$. 
		
		By \Cref{prop:eqhyp_complete} we know that $F$ has a colimit, thus we only have to proof that $\quo$ reflects colimits.  Let $(\mathcal{G}, \{\kappa^d_E, \kappa^d_V, \kappa^d_Q\}_{d\in \D})$ be a cocone over $F$ such that $(\quo(\mathcal{G}). \{\kappa^d_E, \kappa^d_Q\}_{d\in \D})$ is colimiting for $\quo\circ F$.

		and let $(\mathcal{C}, \{(c^d_E, c^d_V)\}_{d\in \D})$, be a colimiting cocone for $\quo \circ F$.  
		
		Let $, )$ be a colimit of $F$, which exists by \Cref{prop:eqhyp_complete},  then by the same proposition we know that $(Q, \{\kappa^d_Q\}_{d\in \D})$ is colimiting for

		Now, let $(U, \{\kappa_d\}_{d\in \D})$  be a colimit for $U_\eq\circ F$. We have already noticed in the proof of the first point of \Cref{prop:eqhyp_complete} that  the family $\{q_d\}_{d\in \D}$ defines a natural transformation $U_{\eq}\circ F\to K\circ F$, moreover, by \Cref{prop:com,colim}, $(V, \{c^d_V\}_{d\in \D})$ is colimiting for  $K\circ F$, yielding an arrow $q: V \to C$, which is a surjection by \Cref{lemma:nat_trans_reg_epi_canonical_arrow_reg_epi}.

Define now $W: \EqHyp \to \Set$ as the functors $(X, Y, Z, x, y, z)$ to the set of hyperedges $X$ and each morphism to its first component. Again by \Cref{prop:com,colim} $(E, \{c^d_E\}_{d\in \D})$ is colimiting for $W \circ F$. Moreover, let $\alpha\colon d\to d'$  be an arrow in $\D$ and suppose that $F(\alpha)=(h_1, h_2, h_3)$, then we have
\begin{align*}
	\kappa^\star_{d'}\circ s_{d'} \circ W(F(\alpha))&= \kappa^\star_{d'}\circ s_{d'} \circ h_1= \kappa^\star_{d'}\circ h_2\circ s_{d}=\kappa^\star\_{d'}circ U(F(\alpha))\circ s^d =\kappa^\star_d\circ s_d
	\\\kappa^\star_{d'}\circ t_{d'} \circ W(F(\alpha))&= \kappa^\star_{d'}\circ t_{d'} \circ h_1= \kappa^\star_{d'}\circ h_2\circ t_{d}=\kappa^\star\_{d'}circ U(F(\alpha))\circ s^d =\kappa^\star_d\circ t_d
\end{align*}
So that $(U^\star, \{\kappa^\star_d\circ s_d\}_{d\in \D})$ and $(U^\star, \{\kappa^\star_d\circ t_d\}_{d\in \D})$ are cocones over $W\circ F$, yielding two arrows $s, t\colon E\rightrightarrows U^\star$.

By the proof of the first point of \Cref{prop:eqhyp_complete} the . \qedhere }
}

\pbm*

\begin{proof}\label{proof:pbmono}
	Closure under composition and decomposition follows from \Cref{lem:pb1}, while stability follows from the first point of \Cref{prop:kerset} because pushouts along injections are also pullbacks by \Cref{prop:regmono}. We are left with stability under pullbacks. 
	Let $h\colon \mathcal{G}_1\to \mathcal{G}_3$ be an arrow in $\pbc$ and $k\colon \mathcal{G}_2\to \mathcal{G}_3$ be another morphism of $\EqHyp$, with $\mathcal{G}_{i}=(E_i, V_i, Q_i, s_i, t_i, q_i)$.
	
	Consider the diagrams below, in which the square on the left is a pullback in $\EqHyp$, with $\mathcal{P}=(E_{\mathcal{P}, V_{\mathcal{P}}}, Q_{\mathcal{P}}, s_{\mathcal{P}}, t_{\mathcal{P}}, q_{\mathcal{P}})$, and in which the right half of the right rectangle is a pullback too, so that the existence of $m_\mathcal{P}$ is guaranteed by \Cref{prop:eqhyp_complete}. 
	
	Again by \Cref{prop:eqhyp_complete} the two halves of the central rectangle are pullbacks and so by \Cref{lem:pb1} the whole rectangle is a pullback too. This, in turn entails that the whole right rectangle is a pullback and so, also its left half is a pullback (again by \Cref{lem:pb1})
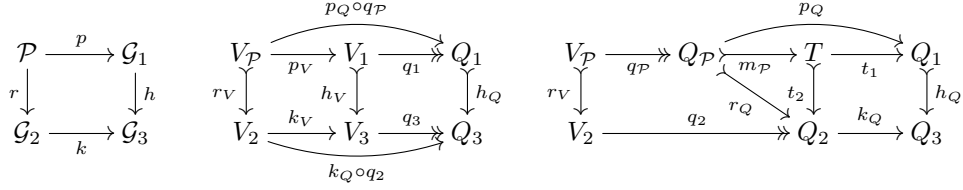
	\[
	\xymatrix@R=15pt{\mathcal{P} \ar[r]^{p} \ar[d]_{r}&  \mathcal{G}_1 \ar[d]^{h} &
		V_{\mathcal{P}} \ar@/^.4cm/[rr]^{p_Q\circ q_{\mathcal{P}}} \ar[r]_{p_V} \ar@{>->}[d]_{r_V}& V_1  \ar@{>>}[r]_{q_1} \ar@{>->}[d]_{h_V}& Q_1 \ar@{>->}[d]^{h_Q}& V_{\mathcal{P}} \ar@{>>}[r]_{q_{\mathcal{P}}} \ar@{>->}[d]_{r_V} & Q_{\mathcal{P}} \ar@{>->}[dr]_{r_Q} \ar@/^.4cm/[rr]^{p_Q} \ar@{>->}[r]_{m_{\mathcal{P}}}& T \ar[r]_{t_1} \ar@{>->}[d]_{t_2} & Q_1 \ar@{>->}[d]^{h_Q}\\
		\mathcal{G}_2 \ar[r]_{k} & \mathcal{G}_3 &   V_2  \ar@/_.3cm/[rr]_{k_Q\circ q_2}\ar[r]^{k_V}& V_3 \ar@{>>}[r]^{q_3} & Q_3 &V_2 \ar@{>>}[rr]^{q_2}&& Q_2 \ar[r]^{k_Q} & Q_3 }
	\]

	Let now $z_1\colon Z\to Q_{\mathcal{P}}$ and $z_2\colon Z\to V_2$ be two arrows such that $q_2\circ z_2=r_Q\circ z_1$. Then there exists a unique arrow $z\colon Z\to V_{\mathcal{P}}$ such that
	\[m_{\mathcal{P}}\circ q_{\mathcal{P}}\circ z= m_{\mathcal{P}}\circ z_1 \qquad r_V\circ z=z_2\]
	Since $m_\mathcal{P}$ is mono we conclude that $q_{\mathcal{P}}\circ z=z_1$ as wanted.
	
	Uniqueness follows at once since $r_V$, being the pullback of $h_V$ is a mono.
\end{proof}

\pbs*

\begin{proof}\label{proof:pbstable}
	Let $\mathcal{G}_i = (A_i, B_i, Q_i, s_i, t_i, q_i)$, $\mathcal{G}'_i=(A'_i, B'_i, Q'_i, s'_i, t'_i, q'_i)$, for $i \in \{1, 2, 3, 4\}$, be hypergraphs with equivalence, 
	and  suppose that in the first diagram below,  all the vertical faces are pullbacks, the bottom face is a pushout $h$ is a regular mono and $k_Q\colon Q_1\mto Q_3$ is mono. By \Cref{prop:eqhyp_complete} the same is true for the other two cubes, by \Cref{cor:mono2}, $h_E$ and $h_V$ are monos and so the top faces of these cubes are pushouts
	
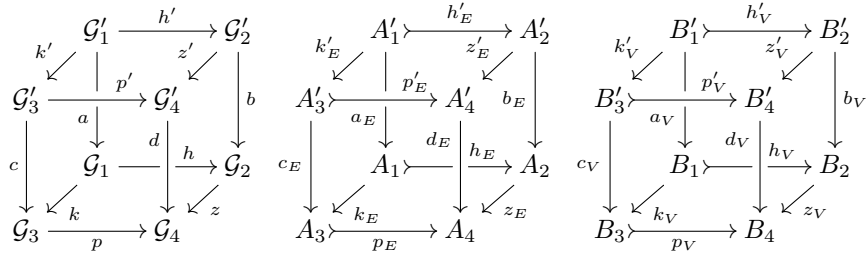
\[
	\xymatrix@C=10pt@R=10pt{&\mathcal{G}_1'\ar[dd]|\hole_(.65){a}\ar[rr]^{h'} \ar[dl]_{k'} && \mathcal{G}_2' \ar[dd]^{b} \ar[dl]_{z'} & &A_1'\ar[dd]|\hole_(.65){a_E}\ar@{>->}[rr]^{h_E'} \ar[dl]_{k_E'} && A_2' \ar[dd]_{b_E} \ar[dl]_{z_E'} &&B_1'\ar[dd]|\hole_(.65){a_V}\ar@{>->}[rr]^{h_V'} \ar[dl]_{k_V'} && B_2' \ar[dd]^{b_V} \ar[dl]_{z_V'} \\ 
		\mathcal{G}_3'  \ar[dd]_{c}\ar[rr]^(.7){p'} & & \mathcal{G}_4' \ar[dd]_(.3){d} &&A_3'  \ar[dd]_{c_E}\ar@{>->}[rr]^(.7){p_E'} & & A_4' \ar[dd]_(.3){d_E}
		&&B_3'  \ar[dd]_{c_V}\ar@{>->}[rr]^(.7){p_V'} & & B_4' \ar[dd]_(.3){d_V}\\
		&\mathcal{G}_1\ar[rr]|\hole^(.65){h} \ar[dl]^{k} && \mathcal{G}_2 \ar[dl]^{z} & &A_1\ar@{>->}[rr]|\hole^(.65){h_E} \ar[dl]^{k_E} && A_2 \ar[dl]^{z_E} &&B_1\ar@{>->}[rr]|\hole^(.65){h_V} \ar[dl]^{k_V} && B_2 \ar[dl]^{z_V} \\
		\mathcal{G}_3 \ar[rr]_{p} & & \mathcal{G}_4 && A_3 \ar@{>->}[rr]_{p_E} & & A_4&&B_3 \ar@{>->}[rr]_{p_V} & & B_4 }
	\]
	
	\noindent 
\parbox{10cm}{Now, if $d=(d_E, d_V, d_Q)$, we can pull back the third component to get the solid part of the cube aside. Notice moreover that the commutativity of the solid diagram yields the existence of the dotted $w\colon T\to Y$. By the usual composition and decomposition properties of pullbacks (cfr.~\Cref{lem:pb1}) the left face is a pullback. By the first point of \Cref{prop:eqhyp_complete} the bottom face is a pushout and $h_Q$ is a mono by \Cref{cor:mono2}, so the top face is a pushout too.}\hfill\parbox{3cm}
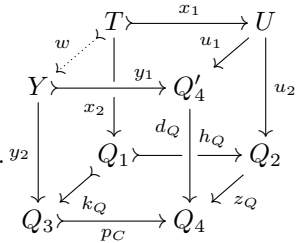
{\xymatrix@C=10pt@R=10pt{&T\ar[dd]|\hole_(.65){x_2}\ar@{>->}[rr]^{x_1} \ar@{>.>}[dl]_{w} && U \ar[dd]^{u_2} \ar[dl]_{u_1} \\ Y  \ar[dd]_{y_2}\ar@{>->}[rr]^(.7){y_1} & & Q_4' \ar[dd]_(.3){d_Q}\\&Q_1\ar@{>->}[rr]|\hole^(.65){h_Q} \ar@{>->}[dl]^{k_Q} && Q_2 \ar[dl]^{z_Q} \\Q_3 \ar@{>->}[rr]_{p_C} & & Q_4 }}

	By the second point of \Cref{prop:eqhyp_complete} we know that there are monos $m_2\colon Q'_2\mto U$ and $m_3\colon Q'_3\mto Y$  fitting in the diagrams below
	
	
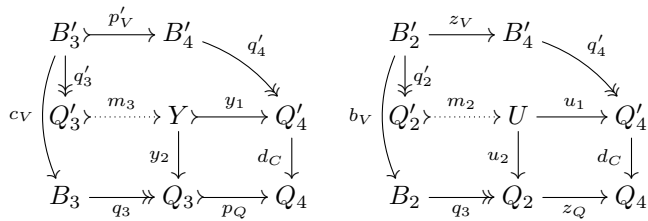
\[\xymatrix@R=15pt{B'_3 \ar@{>>}[d]^{q'_3} \ar@/_.3cm/[dd]_{c_V} \ar@{>->}[r]^{p'_V}& B'_4 \ar@/^.2cm/@{>>}[dr]^{q'_4} & & B'_2 \ar[r]^{z_V} \ar@{>>}[d]^{q'_2} \ar@/_.3cm/[dd]_{b_V}& B'_4 \ar@/^.2cm/@{>>}[dr]^{q'_4}\\ Q'_3 \ar@{>.>}[r]^{m_3}& Y  \ar@{>->}[r]^{y_1} \ar[d]_{y_2}& Q'_4 \ar[d]_{d_C}& Q'_2 \ar@{>.>}[r]^{m_2} & U \ar[d]_{u_2} \ar[r]^{u_1} & Q'_4 \ar[d]_{d_C} \\ B_3 \ar@{>>}[r]_{q_3} & Q_3 \ar@{>->}[r]_{p_Q} & Q_4 &B_2 \ar@{>>}[r]_{q_3} & Q_2 \ar[r]_{z_Q} & Q_4}\]

\newpage
	\noindent 
	\parbox{8.5cm}{\hspace{15pt}
	For $Q'_1$, we can make a similar argument. Taking $(S, s_1, s_2)$ as the pullback of $m_2$ along $x_1$, we can use again the composition properties of pullbacks and the second point of \Cref{prop:eqhyp_complete} to guarantee the existence of a monomorphism $m_1\colon Q'_1\to S$ that makes the diagram aside commutative.}\hfill
	\parbox{4cm}
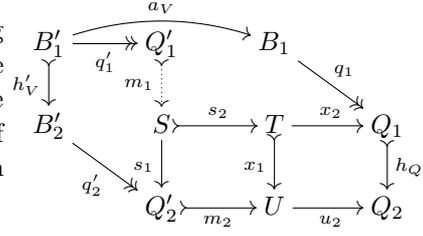
{\xymatrix@R=15pt{B'_1\ar@{>>}[r]_{q'_1} \ar@{>->}[d]_{h'_V} \ar@/^.3cm/[rr]^{a_V} & Q'_1 \ar@{>.>}[d]_{m_1} & B_1 \ar@{>>}[dr]^{q_1}\\B'_2 \ar@{>>}[dr]_{q'_2} & S \ar@{>->}[r]^{s_2} \ar[d]_{s_1} & T \ar[r]^{x_2} \ar@{>->}[d]_{x_1}  & Q_1 \ar@{>->}[d]^{h_Q}\\& Q'_2 \ar@{>->}[r]_{m_2}& U \ar[r]_{u_2}& Q_2 }}

We have to show that the top face of the cube with which we have begun is a pushout. Let $\mathcal{H}$ be $(E, V, Q, s, t, q)$ and suppose that there exist  $o\colon  \mathcal{G}_2' \to \mathcal{H}$ and $w\colon  \mathcal{G}_3' \to \mathcal{H}$ such that $o \circ h' = w \circ k'$. Since $T$ preserves colimits by \Cref{prop:forghyp}, we know that there exists a morphism $(v_E, v_V)\colon T(\mathcal{G'}_4)\to T(\mathcal{H})$ such that
\[v_E\circ p'_E=w_E \quad v_V\circ p'_V=w_V \quad v_E\circ z'_E=o_e \quad v_V\circ z'_V=o_V\]

\commentato{\xymatrix@C=10pt@R=6pt{
		&B_1'\ar@{>>}[dd]|\hole_(.65){q_1'}\ar@{>->}[rr]^{h_V'} \ar[dl]_(.6){k_V'} && B_2' \ar@{>>}[dd]|\hole_(.65){q_2'} \ar[dl]_(.6){z_V'}\ar[dr]^{o_V}\\
		B_3'  \ar@{>>}[dd]_{q_3'}\ar@{>->}[rr]^(.7){p_V'} & & B_4' \ar@{>>}[dd]_(.3){q_4'} \ar[rr]^(.7){v_V}&&  V\ar[dd]^{q}\\
		&Q_1'\ar@{>->}[rr]|\hole^(.65){h'_Q} \ar[dl]_{k'_Q} && Q_2' \ar[dl]^(.4){z_Q'}\ar[dr]^{o_Q}\\
		Q'_3 \ar@{>->}[rr]_{p_Q'} & & Q'_4\ar@{.>}[rr]_{v_C} && Q
}}

\smallskip \noindent 
\parbox{8cm}{ \hspace{15pt}We want to extend such a morphism to one of $\EqHyp$ between $\mathcal{G}'_4\to \mathcal{H}$. If we are able to do so we can conclude because uniqueness is guaranteed by \Cref{rem:eqhyp_morphs}. 
	
	\hspace{15pt}Now, consider the cube aside, $h$ is in $\pbc$ so by \Cref{lem:pb1} we know that its back face is a pullback. We can then apply the second point of \Cref{prop:kerset} to deduce that the square on the right is a pushout.
	
	By construction we have that 
	\begin{align*}
		m_3\circ q'_3 \circ \pi_{m_3 \circ q'_3}^1& = m_3\circ q'_3 \circ \pi_{m_3\circ q_3'}^2 \\
		m_2\circ q_2' \circ \pi_{m_2 \circ q_2'}^1 &= m_2\circ q'_2 \circ \pi_{m_2 \circ q_2'}^2
		\end{align*} 
	}\hfill\parbox{4cm}{\xymatrix@C=10pt@R=10pt{
	&B_1'\ar@{>>}[dd]|(.53)\hole_(.67){s_2\circ m_1\circ q_1'}\ar@{>->}[rr]^{h_V'} \ar[dl]_(.6){k_V'} && B_2' \ar@{>}[dd]^{m_2\circ q_2'} \ar[dl]_(.6){z_V'} & & \\
	B_3'  \ar[dd]_{m_3\circ q_3'}\ar@{>->}[rr]^(.7){p_V'}  &&B_4' \ar@{>>}[dd]_(.3){q_4'} \\
	&T\ar@{>->}[rr]_(.3){x_1}|(.5)\hole \ar[dl]_{w} && U\ar[dl]^(.4){u_1} \\
	Y \ar@{>->}[rr]_{y_1} & &  Q'_4\\K_{s_2\circ m_1\circ q'_1}  \ar@{>->}[rr]_-{k_{h'_V}} \ar[dd]_{k_{k'_V}}&& K_{m_2\circ q'_2} \ar[dd]^{k_{z'_V}}\\ \\ K_{m_3\circ q'_3} \ar@{>->}[rr]^-{k_{p'_V}} && K_{q'_4}}
}
	but since $m_3$ and $m_2$ are monos this implies
\[q'_3 \circ \pi_{m_3 \circ q'_3}^1 = q'_3 \circ \pi_{m_3\circ q_3'}^2 \qquad q_2' \circ \pi_{m_2 \circ q_2'}^1 = q'_2 \circ \pi_{m_2 \circ q_2'}^2\]
	
By computing, we obtain
\begin{align*}
	&q \circ v_V \circ \pi_{q_4'}^1 \circ k_{p_V'} = q \circ v_V \circ p_V' \circ \pi_{m_3 \circ q_3'}^1= q \circ w_V \circ \pi_{m_3 \circ q_3'}^1 = w_Q \circ q_3' \circ \pi_{m_3 \circ q_3'}^1 \\
	=& w_Q \circ q_3' \circ \pi_{m_3 \circ q_3}^2 = q \circ w_V \circ \pi_{m_3 \circ q_3'}^2 = q \circ v_V \circ p_V' \circ \pi_{m_3 \circ q_3'}^2 = q \circ v_V \circ \pi_{q'_4}^2 \circ k_{p'_V}\\
&q \circ v_V \circ \pi_{q'_4}^1 \circ k_{z_V'} = q \circ v_V \circ z_V' \circ \pi_{m_2 \circ q_2'}^1 = q \circ o_V \circ \pi_{m_2 \circ q_2'}^1 = o_Q\circ q_2' \circ \pi_{m_2 \circ q_2'}^1 \\
	= &o_Q \circ q_2' \circ \pi_{m_2 \circ q_2}^2 = q \circ o_V \circ \pi_{m_2 \circ q_2'}^2 = q \circ v_V \circ t_V' \circ \pi_{m_2 \circ q_2'}^2 = q \circ v_V \circ \pi_{q'_4}^2 \circ k_{z'_V}
\end{align*}

As already noticed the square above is a pushout, hence we can conclude that
\[
q \circ v_V \circ \pi_{q'_4}^1 = q \circ v_V \circ \pi_{q'_4}^2
\]

Now, $q'_4$ is a regular epi and so it is the coequalizer of its kernel pair, hence there exists $v_Q\colon Q'_4\to Q$ such that $v_Q\circ q'_4=q\circ v_V$ and we can conclude.
\end{proof}

\pbvk*

\begin{proof}\label{proof:pbvk}
	
Consider a cube as the one below on the left, with regular monos as vertical arrows, pullbacks as back faces, pushouts as bottom and top faces and such that $h$ is in $\pbc$.
Given \Cref{lemma:stab}, if we show that the front faces are pullbacks too we can conclude.

To fix notation, let $\mathcal{G}_i = (A_i, B_i, Q_i, s_i, t_i, q_i)$, $\mathcal{G}'=(A_i', B_i', Q_i', s_i', t_i', q_i')$, for $i = 1, 2, 3, 4$. By \Cref{prop:eqhyp_complete} and \Cref{cor:mono2} we know that the central and right cube below have pushouts as bottom faces and pullbacks as back faces,
	thus their front faces are pullbacks
	
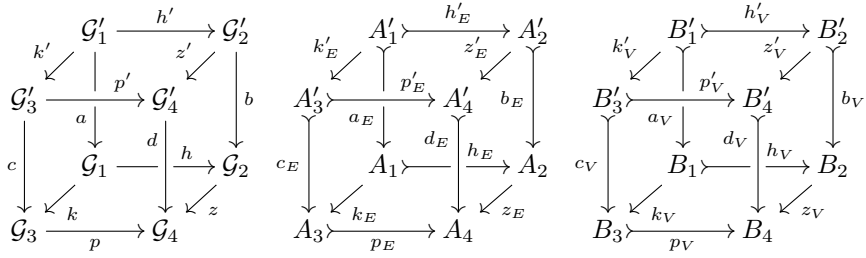
\[
	\xymatrix@C=10pt@R=10pt{&\mathcal{G}_1'\ar[dd]|\hole_(.65){a}\ar[rr]^{h'} \ar[dl]_{k'} && \mathcal{G}_2' \ar[dd]^{b} \ar[dl]_{z'} & &A_1'\ar@{>->}[dd]|\hole_(.65){a_E}\ar@{>->}[rr]^{h_E'} \ar[dl]_{k_E'} && A_2' \ar@{>->}[dd]_{b_E} \ar[dl]_{z_E'} &&B_1'\ar@{>->}[dd]|\hole_(.65){a_V}\ar@{>->}[rr]^{h_V'} \ar[dl]_{k_V'} && B_2' \ar@{>->}[dd]^{b_V} \ar[dl]_{z_V'} \\ 
		\mathcal{G}_3'  \ar[dd]_{c}\ar[rr]^(.7){p'} & & \mathcal{G}_4' \ar[dd]_(.3){d} &&A_3'  \ar@{>->}[dd]_{c_E}\ar@{>->}[rr]^(.7){p_E'} & & A_4' \ar@{>->}[dd]_(.3){d_E}
		&&B_3'  \ar@{>->}[dd]_{c_V}\ar@{>->}[rr]^(.7){p_V'} & & B_4' \ar@{>->}[dd]_(.3){d_V}\\
		&\mathcal{G}_1\ar[rr]|\hole^(.65){h} \ar[dl]^{k} && \mathcal{G}_2 \ar[dl]^{z} & &A_1\ar@{>->}[rr]|\hole^(.65){h_E} \ar[dl]^{k_E} && A_2 \ar[dl]^{z_E} &&B_1\ar@{>->}[rr]|\hole^(.65){h_V} \ar[dl]^{k_V} && B_2 \ar[dl]^{z_V} \\
		\mathcal{G}_3 \ar[rr]_{p} & & \mathcal{G}_4 && A_3 \ar@{>->}[rr]_{p_E} & & A_4&&B_3 \ar@{>->}[rr]_{p_V} & & B_4 }
	\]

Let us now consider the diagrams below, in which the inner squares are pullbacks.	Since the outer diagrams commute, by definition of morphism of $\EqHyp$, then we have the existence of $m_2\colon Q'_2\to U$, $m_3\colon Q'_3\to Y $, $a_3\colon B'_3\to Y$ and $a_2\colon B'_2\to Y$.

\smallskip 
\noindent 
\parbox{5.2cm}{\hspace{15pt}Now, notice that  $m_3$ and $m_2$ are monos because $c_Q$ and $b_2$ are injections. By the proof of \Cref{prop:eqhyp_complete}, to conclude it is enough to show that
	\[m_3\circ q'_3 = a_3 \qquad m_2\circ q'_2=a_2\]
	
		\hspace{15pt}Indeed, if the previous equations hold, then $Q'_3$ and   $Q'_2$ are images for $a_3$ and $a_2$ and the claim follows from \Cref{prop:eqhyp_complete}.
	}\hfill \parbox{3.5cm}{\xymatrix@R=16pt{Q'_3 \ar@{.>}[r]_{m_3} \ar@{>->}@/^.3cm/[rr]^{p'_Q} \ar@{>->}@/_.3cm/[dr]_{c_Q} &Y \ar@{>->}[d]_{y_2}\ar@{>->}[r]_{y_1}& Q'_4 \ar@{>->}[d]^{d_Q}& Q'_2 \ar@{.>}[r]_{m_2} \ar@/^.3cm/[rr]^{z'_Q} \ar@{>->}@/_.3cm/[dr]_{b_Q}&U \ar[r]_{u_1} \ar@{>->}[d]_{u_2}& Q'_4\ar[d]^{d_Q}\\&Q_3 \ar@{>->}[r]_{p_3}& Q_4 &&Q_2 \ar[r]_{z_Q}& Q_4 }
\xymatrix@R=16pt{B'_3 \ar@{.>}[dr]^{a_3} \ar@{>->}[r]^{p'_V} \ar@{>>}[d]_{q'_3} &B'_4\ar@{>>}[dr]^{q'_4}&& B'_2\ar@{.>}[dr]^{a_2} \ar[r]^{z'_2} \ar@{>>}[d]_{q'_2} &B'_4\ar@{>>}[dr]^{q'_4}\\Q'_3 \ar@{>->}[dr]_{c_Q} &Y \ar@{>->}[r]_{y_1} \ar@{>->}[d]_{y_2}& Q'_4\ar[d]^{d_Q}&Q_2' \ar@{>->}[dr]_{c_Q}& U \ar@{>->}[d]_{u_2}\ar[r]_{u_1}& Q'_4 \ar@{>->}[d]^{d_Q}\\&Q_3 \ar@{>->}[r]_{p_Q}& Q_4 &&Q_2 \ar[r]_{z_Q}& Q_4 }}

By computing we have
\[\begin{split}
	y_1\circ a_3&= q'_4\circ p'_2=p'_3 \circ q'_3=y_1\circ m_3\circ q'_3\\
u_1\circ a_2&= q'_4\circ t'_2=t'_3 \circ q'_3=u_1\circ m_2\circ q'_2 
\end{split} \qquad \begin{split}
 	y_2\circ a_3&= d_3\circ q_3'=y_2\circ m_3\circ q'_3\\
u_2\circ a_2&= d_2\circ q_2'=u_2\circ m_2\circ q'_2
\end{split}\]	
And we have done.
\end{proof} 

\trm*
\begin{proof}\label{proof:term}
	Let $F\colon \D\to \EqTG_{\Sigma}$ be the diagram of an equalizer, a binary product or of a pullback. By \Cref{prop:lim} we can consider a limiting cone $((\mathcal{L}, l), \{\pi_d\}_{d\in \D})$. By \Cref{rem:t} we know that $T_\Sigma$ preserves limits, thus by \Cref{rem:obv}  $(l, \{T_\Sigma(\pi_d)\}_{d\in \D})$ is limiting for $I_\Sigma\circ S_\Sigma$. Then by \Cref{prop:tlim} $l$ is a term graph. We conclude, again by \Cref{rem:obv} that $(\mathcal{L}, l)$ is in $\EqTG_{\Sigma}$ and the claim follows. 
\end{proof}

\po*
\begin{proof}\label{proof:tade}
	Suppose that the square on the left below is a pushout in $\EqHyp_\Sigma$, with $h$ in $\mathcal{T}$. Then, by \Cref{rem:t} the square on the right is a pushout and by the definition of $\mathcal{T}$ and \Cref{rem:obv}, $T_\Sigma(h)$ is a regular mono in $\tg$
	\[\xymatrix{J_\Sigma(\mathcal{G}_0, l_0) \ar[r]^{k} \ar[d]_{h}& J_\Sigma(\mathcal{G}_1, l_1) \ar[d]^{p} &T_\Sigma (J_\Sigma(\mathcal{G}_0, l_0)) \ar[r]^{T_{\Sigma}(k)} \ar[d]_{T_{\Sigma}(h)}& T_\Sigma (J_\Sigma(\mathcal{G}_1, l_1)) \ar[d]^{T_\Sigma(p)}\\
J_\Sigma(\mathcal{G}_1, l_1) \ar[r]_{q}& J_\Sigma(\mathcal{H}, l) & T_\Sigma(J_\Sigma(\mathcal{G}_1, l_1)) \ar[r]_{T_\Sigma(q)}& T_\Sigma(J_\Sigma(\mathcal{H}, l))}\]
	
	We conclude using \Cref{prop:push,rem:obv}.
\end{proof}

\subsection{Proofs for \Cref{eggs}}
\elim*
\begin{proof}\label{proof:elim}
	Let $F: \D \to \egg$ be a diagram, with $F(d) = (A_d, B_d, Q_d, s_d, t_d, q_d)$.  Let $(U_d, u_1^d, u_2^d)$ be a kernel pair for $q_d\circ s_d$.
	Now let $(A, B, Q, s, t, q)$, together with projections $(\pi_E^d, \pi_V^d, \pi_Q^d)_{d\in \D}$, be the limit of $I \circ F$. Suppose that $(U, u_1, u_2)$ is the kernel pair of $q^\star\circ s$ and let $(L, (l_i)_{i \in \cat I})$ be the limit of $K \circ I \circ F$. If we show that it lies in $\egg$ we are done.
	
	By \Cref{prop:eqhyp_complete} there exists a mono $m\colon Q \mto L$ such that $\pi_Q^d = l_d \circ m$. Notice that
	\begin{align*}
		&q_d^\star \circ s_d\circ \pi^d_E\circ u_1  = q_d^\star\circ (\pi^d_V)^\star \circ s\circ u_1= (\pi_Q^d)^\star\circ q^\star\circ s\circ u_1\\=&(\pi_Q^d)^\star\circ q^\star\circ s\circ u_2= q_d^\star\circ (\pi^d_V)^\star\circ s\circ u_2=q_d^\star \circ s_d \circ \pi_E^d \circ u_2\\
	\end{align*}
	
	\vspace{.1cm}
	\noindent 
	\parbox{9.5cm}{ \vspace{-.8cm}\hspace{15pt}
	Thus for each $d$ in $\D$, there exists an arrow $a_d:U\to U_i$ making the diagram on the right commutative.
Then we have
\begin{align*}
	&l_d^\star\circ m^\star \circ q^\star \circ t \circ u_1    = q_d^\star\circ (\pi_V^d)^\star \circ t \circ u_1 \\
	= &q_d^\star \circ t_d \circ \pi_E^d \circ u_1
	= q_d^\star \circ t_d \circ u_1^d \circ a_d 
	= q_d^\star \circ t_d \circ u_2^d \circ a_d\\ 
	=& q_d^\star \circ t_d \circ \pi_E^d \circ u_2 = q_d^\star \circ (\pi_V^d)^\star \circ t \circ u_2 = l_d^\star \circ m^\star \circ q^\star \circ t \circ u_2
\end{align*}}\hfill \parbox{4cm}{\vspace{-1cm}	\xymatrix@R=20pt{
		U \ar[r]^{u_1} \ar[d]_{u_2} \ar@{.>}[dr]_{a_d} & A \ar[dr]^{\pi_E^d} & \\
		A \ar[dr]_{\pi_E^d} & U_d \ar[r]^{u_1^d} \ar[d]_{u_2^d}& A_d \ar[d]^{q_d^\star \circ s_d} \\
		& A_i  \ar[r]_{q_d^\star \circ s_d}& Q_d^\star
	}}

	By universal property of limits, we have that \- $m^\star\circ q^\star \circ t \circ u_1 = m^\star \circ q^\star \circ t \circ u_2$. Since $m$ is mono we deduce that $q^\star \circ t \circ u_1 = q^\star \circ t \circ u_2$, hence the claim.
\end{proof}

\begin{remark}\label{rem:kleene}
	Let $\delta_1\colon 1\to \mathbb{N}$ be the arrow which picks $1\in \mathbb{N}$. Consider the
	
	\noindent
	\parbox{11.5cm}{ inclusion $v_1\colon X\to X^\star $  defined by the first point of \Cref{prop:fact}. By extensivity we know that the square on the right is a pullback.}\hfill \parbox{4cm}{\vspace{-.5cm}\xymatrix@R=15pt{X \ar[r]_{!_{X}} \ar[d]_{!_{X}}& 1\ar[d]^{\delta_1} \\1 \ar[r]^{\lgh_{X}} & \mathbb{N}} }
	
	\noindent 
	\parbox{4cm}{\vspace{-.2cm} \xymatrix@C=15pt@R=1pt{&X\ar[dd]|\hole_(.65){!_{X}}\ar@{->}[rr]^{!_X} \ar@{>}[dl]_{f} && 1 \ar[dd]^{\delta_1} \ar[dl]^{\id{1}} \\ Y  \ar[dd]_{!_Y}\ar@{>}[rr]^(.7){!_Y} & & 1 \ar[dd]_(.3){\delta_1}\\&X^\star\ar@{>->}[rr]|(.53)\hole^(.70){\lgh_{X}} \ar@{>->}[dl]_{f^\star} && \mathbb{N} \ar[dl]^{\id{N}} \\Y^\star \ar@{>->}[rr]_{\lgh_{Y}} & & \mathbb{N}}}\hfill \parbox{9.5cm}{\vspace{-.3cm}\hspace{15pt}Let now $f\colon X\to Y$ be an arrow. Then we can build the cube aside which, by what we have just observed above, has pullbacks as back, front and right faces. Thus, by \Cref{lem:pb1} its left face is a pullback too.}
	
\end{remark}

\noindent
\parbox{11.5cm}{
\epo*}\hfill 
\parbox{4cm}{\xymatrix@R=15pt{\mathcal{G}_1 \ar[r]^{h}\ar[d]_{m}&\mathcal{G}_2\ar[d]^{n}\\\mathcal{G}_3\ar[r]_{z}&\mathcal{P}}}
\begin{proof}\label{proof:epo}
	Let $\mathcal{P} $ be $(A, B, C, s, t, q)$ and consider the kernel pair $(K_i, \pi_i^1, \pi_i^2)$ the kernel pair of $q_i^\star \circ s_i$, for $i\in \{1, 2, 3\}$. Let also $(U, u_1, u_2)$ be the kernel pair of $q^\star \circ s$. 
	
	\noindent
	\parbox{10cm}{\hspace{15pt}Consider now the cube on the right. By hypothesis and \Cref{rem:kleene} its left face is a pullback, moreover its bottom face, while not a pushout it is still a pullback by the first point of \Cref{prop:regmono} and the third one of \Cref{prop:fact}.}\hfill \parbox{3cm}{\xymatrix@C=10pt@R=2pt{&A_1\ar[dd]|\hole_(.65){q_1^\star \circ s_1}\ar[rr]^{h_E} \ar@{>->}[dl]_{m_E} && A_2 \ar[dd]^{q_2^\star \circ s_2} \ar@{>.>}[dl]_(.6){n_E} \\ A_3 \ar[dd]_{q_3^\star \circ s_3}\ar[rr]^(.7){z_E} & & A \ar[dd]_(.3){q^\star \circ s}\\&{Q_1^\star}\ar[rr]|\hole^(.65){h_Q^\star}\ar@{>->}[dl]_{m_Q^\star} && {Q_2^\star} \ar@{>->}[dl]^{n_Q^\star} \\{Q_3^\star} \ar[rr]_{z_Q^\star} & & Q^\star}}
	
	\noindent
	\parbox{2cm}{\vspace{-.25cm}\xymatrix@R=15pt{K_1\ar[r]_{k_{h_E}}\ar@{>->}[d]_{k_{m_E}}&K_2\ar@{>->}[d]^{k_{n_E}}\\K_3\ar[r]^{k_{z_E}}&U}}
\hfill \parbox{11.5cm}{ \vspace{-.25cm}\hspace{15pt}By hypothesis and the first point of \Cref{prop:eqhyp_complete} the top face is a pushout. Thus we can apply \Cref{prop:kerset} to deduce that the square on the left is a pushout.}
		By computing we obtain
	\begin{align*}
		q^\star \circ t \circ u_1 \circ f_n &= q^\star \circ t \circ n_E \circ \pi_2^1 = n_C^\star \circ q_2^\star \circ s_2 \circ \pi_2^1 = n_C^\star \circ q_2^\star \circ s_2 \circ \pi_2^2 =q^\star \circ t \circ u_2 \circ f_n \\
		q^\star \circ t \circ u_1 \circ f_k &= q^\star \circ t \circ k_E \circ \pi_3^1 = k_C^\star \circ q_3^\star \circ s_3 \circ \pi_3^1 = k_C^\star \circ q_3^\star \circ s_3 \circ \pi_3^2 =q^\star \circ t \circ u_2 \circ f_k
	\end{align*}
	We can therefore deduce that $q^\star \circ t \circ u_1 = q^\star \circ t \circ u_2$, and the claim follows.
\end{proof}

\limt*
\begin{proof}\label{proof:tlim}
Let $F\colon \D\to \eg$ be a diagram of one of the shapes mentioned in the statement. Let $((\mathcal{L}, l), \{\pi_d\}_{d\in \D})$ be a limiting cone for $k_\Sigma \circ F$. By \Cref{lim}, \Cref{prop:term}, and \Cref{proof:elim} we know that $\mathcal{L}$ is an e-termgraph and we can conclude.
\end{proof}

\section{Some properties of comma categories}
In this section we briefly recall the definition of the comma category \cite{mac2013categories} associated to two functors and some of its properties.
\begin{definition}\index{category!comma -}
	Let $L\colon \A\to \X$ and  $R\colon \B\rightarrow \X$ be two functors with the same codomain, the \emph{comma category} $\comma{L}{R}$ is the category in which

\noindent
\parbox{10.5cm}{\vspace{-.5cm}\begin{itemize}
		\item objects are triples $(A, B, f)$ with $A\in \A$, $B\in \B$, and $f\colon L(A)\rightarrow R(B)$; 
		\item a morphism $(A, B, f)\rightarrow (A', B', g)$ is a pair $(h, k)$ with $h\colon A\rightarrow A'$ in $\A$ and $k\colon B\rightarrow B'$ in $\B$ such that the diagram aside commutes\end{itemize}}\hfill 
	\parbox{4cm}{\vspace{-.1cm}\xymatrix@C=25pt{L(A) \ar[r]^{L(h)} \ar[d]_{f} & L(A') \ar[d]^{g}\\ R(B) \ar[r]_{R(k)}& R(B')}}
\end{definition} 
We have two forgetful functors 	$U_L\colon \comma{L}{R}\to \A$ and $U_R\colon \comma{L}{R} \to \B$ given respectively by
\[
\begin{split}
	\functor[l]{(A,B, f)}{(h,k)}{(A', B', g)}
	& \functormapsto
	\rfunctor{A}{h}{A'}
\end{split}\quad 
\begin{split}
	\functor[l]{(A,B, f)}{(h,k)}{(A', B', g)}
	& \functormapsto
	\rfunctor{B}{k}{B'}
\end{split}
\]

Given $L\colon \A\to \X$ and $R\colon \B\to \X$, we can also consider their duals $L^{op}\colon \A^{op}\to \X^{op}$ and $R^{op}\colon \B^{op}\to \X^{op}$.  An arrow $f\colon L(A)\to R(B)$ in $\X$ is the same thing as an arrow $f\colon R^{op}(B)\to L^{op}(A)$ in $\X^{op}$. Thus $\comma{L}{R}$ and $\comma{R^{op}}{L^{op}}$ have the same objects. Moreover, the left square below commutes in $\X$ if and only if the right one commutes in $\X^{op}$.
\[\xymatrix@C=25pt{L(A) \ar[r]^{L(h)} \ar[d]_{f} & L(A') \ar[d]^{g} & R(B') \ar[d]_{g} \ar[r]^{R(k)}& R(B) \ar[d]^{f}\\ R(B) \ar[r]_{R(k)}& R(B') & L(A') \ar[r]_{L(h)}  & L(A)}\]

Summing up we have just proved the following fact.
\begin{proposition}\label[proposition]{prop:dual}
	$(\comma{L}{R})^{op}$ is equal to	$\comma{R^{op}}{L^{op}}$, and $U^{op}_L=U_{L^{op}}$ and $U^{op}_R=U_{R^{op}}$.
\end{proposition}

\begin{lemma}\label[lemma]{colim}
	Let $L\colon \A\to \X$ and $R\colon \B\to \X$ be functors and $F\colon \D\to\comma{L}{R}$ be a diagram such that $L$ preserves colimits along $U_L\circ F$. Then the family $\{U_L, U_R\}$ jointly creates colimits of $F$ (see \cite[Sec.~5.1.3]{castelnovo2023thesis} or \cite[Sec.~2.3]{CastelnovoGM24}).
\end{lemma}
\begin{proof}
	Suppose that $U_L\circ F$ and $U_R\circ F$ have respectively colimiting cocones  $\left(A, \{a_D\}_{D\in \D}\right)$ and $\left(B, \{b_D\}_{D\in \D}\right)$. By hypothesis $\left(L(A), \{L\left(a_D\right)\}_{D\in \D}\right)$ is colimiting for $L\circ U_L\circ F$. Now, let $F(D)$ be $(A_D, B_D, f_D)$,
	then we have arrows $R(b_D)\circ f_D\colon L(A_D)\to R(B)$ that forms a cocone on $L\circ U_L\circ F$: if $d\colon D\to D'$ is an arrow in $\D$ then $F(d)$ is an arrow in $\comma{L}{R}$ and so
	\begin{align*}
		&R\left(b_{D'}\right)\circ f_{D'}\circ L(U_L(F(d)))=R\left(b_{D'}\right)\circ R\left(U_R\left(F(d)\right)\right)\circ f_D\\=&R\left(b_{D'}\circ U_R\left(F(d)\right)\right)\circ f_D=R\left(b_D\right)\circ f_D
	\end{align*}
	
	\noindent \parbox{10cm}{
	\hspace{15pt}Thus there exists $f\colon L(A)\rightarrow R(B)$ fitting in the diagram on the right. Notice that $f$ is the unique arrow in $\X$ wich makes $\left(a_D, b_D\right)$ an arrow $\left(A_D, B_D, f_D\right)\to \left(A, B, f\right)$ of $\comma{L}{R}$. If we show that $\left((A, B, f), \left\{(a_D, b_D)\right\}_{D\in \D}\right)$ is colimiting for $F$ we are done.}\hfill 
	\parbox{4cm}{\xymatrix@C=35pt@R=15pt{ L(A_D) \ar[r]^{L(a_D)} \ar[d]_{f_D}& L(A) \ar[d]^{f}\\ R(B_D) \ar[r]_{R(b_D)} &R(B)}}
	
	First of all, let us show that it is a cocone. Given $d\colon D\to D'$ in $\D$ we have
	\begin{align*}
		&\left(a_{D'}, b_{D'}\right)\circ F(d)=	\left(a_{D'}, b_{D'}\right)\circ \left(U_L(F(d)), U_R(F(d))\right)\\=&\left( a_{D'}\circ U_L(F(d)),  b_{D'}\circ U_R(F(d)) \right)=\left(a_D, b_D\right)
	\end{align*}
	
	For the colimiting property, let $\left((X, Y, g), \left\{\left(x_D, y_D\right)\right\}_{D\in \D}\right)$ be another cocone on $F$. In particular, $\left(X, \left\{x_D\right\}_{D\in \D}\right)$ and $\left(Y, \left\{y_D\right\}_{D\in \D}\right)$ are cocones on $U_L\circ F$ and $U_R\circ F$ respectively, so we have uniquely determined arrows $x\colon A\rightarrow X$ and $y\colon B\rightarrow Y$ such that 
	\[x\circ a_D= x_D \qquad y\circ b_D= y_D\]
	Let us show that $(x,y)$ is an arrow of $\comma{L}{R}$. Given $D\in \D$ we have
	\begin{align*}
		&R(y)\circ f\circ L(a_D)=R(y)\circ R\left(b_D\right)\circ f_D=R\left(y\circ b_D\right)\circ f_D\\=&R\left(y_D\right)\circ f_D= g\circ L\left(x_D\right)=g\circ L\left(x\circ a_D\right)=g\circ L(x)\circ L\left(a_D\right)
	\end{align*}
	from which it follows  that $g\circ L(x)=R(y)\circ f$ as wanted.
\end{proof}

\Cref{prop:dual} and \Cref{colim} now yields the following.
\begin{corollary}\label[corollary]{lim} The family $\{U_L, U_R\}$ jointly creates limits along every diagram $F\colon \D\to \comma{L}{R}$ such that $R$ preserves the limit of $U_R\circ F$.
\end{corollary}

We can use \Cref{lim} to characterize monos in comma categories. 
\begin{corollary}\label[corollary]{cor:mono}
	If $R$ preserves pullbacks then an arrow $(h,k)$ in $\comma{L}{R}$ is mono if and only if both $h$ and $k$ are monos.
\end{corollary}
\begin{proof}
	$(\Rightarrow)$  If $(h,k)\colon (A,B,f)\to (A', B', g)$ is a mono then the first rectangle below is a pullback in $\comma{L}{R}$. By \Cref{lim} then also the other two squares are pullbacks, respectively in $\A$ and $\B$, proving that both $h$ and $k$ are monos
	\[\xymatrix{(A, B, f)  \ar[rr]^{\id{(A,B,f)}} \ar[d]_{\id{(A,B,f)}}&& (A, B, f)  \ar[d]^{(h,k)}& A \ar[r]^{\id{A}} \ar[d]_{\id{A}} & A \ar[d]^{h} & B \ar[r]^{\id{B}} \ar[d]_{\id{B}} & B  \ar[d]^{k}\\ (A, B, f)  \ar[rr]_{(h,k)}&& (A', B', g) &A \ar[r]_{h} & A' & B \ar[r]_{k} & B'}\]
	
\newpage
	\smallskip\noindent 
	\parbox{4cm}{ \xymatrix{A \ar[r]^{\id{A}} \ar[d]_{\id{A}} & A \ar[d]^{h} & B \ar[r]^{\id{B}} \ar[d]_{\id{B}} & B  \ar[d]^{k}\\ A \ar[r]_{h} & A' & B \ar[r]_{k} & B'}}\hfill \parbox{8.5cm}{$(\Leftarrow)$ Since $h$ and $k$ are monos then we have the two pullback squares on the left. Thus by \Cref{lim} the pullback of $(h,k)$ along itself has isomorphisms as projections and so $(h,k)$ is mono. \qedhere}
	
\end{proof}

We end this section pointing out another useful fact,  showing that in some cases we can guarantee  the existence of a left adjoint to $U_R$. 

\begin{proposition}\label[proposition]{prop:left}
	If $\A$ has initial objects and $L$ preserves them then the forgetful functor $U_R\colon \comma{L}{R}\to \B$ has a left adjoint $\Delta$.
\end{proposition}
\begin{proof} For an object $B\in \B$ we can define $\Delta(B)$ as $(0, B, ?_{R(B)})$, for $0$ is an initial object in $\A$ and $?_{R(B)}$ is the unique arrow $L(0)\to R(B)$. Let $\id{B}\colon B\to U_R(\Delta(B))$ be the identity,
	
	\noindent 
	\parbox{10cm}{and suppose that a $k\colon B\to U_R(A, B', f)$ in $\B$ is given. By initiality of $0$, there is only one arrow $?_A\colon 0\to A $ in $\A$ and, 
	since $L$ preserves initial objects, the square aside commutes. Thus $(?_A,k)$ is the unique morphism $\Delta(B)\to (A, B', f)$ such that $U_R(?_A,k)=k$. }\hfill 
	\parbox{2cm}
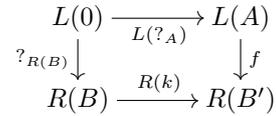
{\vspace{-.1cm}\xymatrix@C=30pt@R=15pt{L(0) \ar[r]_{L(?_A)} \ar[d]_{?_{R(B)}} & L(A) \ar[d]^{f}\\ R(B) \ar[r]^{R(k)}& R(B')}}  
\end{proof}
Dualizing we get immediately the following.
\begin{corollary}If $\B$ has terminal objects preserved by $R$ then $U_L\colon \comma{L}{R}\to \A$ has a right adjoint.
\end{corollary}

\subsection{Slice categories}

This section is devoted to recall some basic facts about the so called \emph{slice categories}.
\begin{definition}\index{category!slice -}
	Let $X$ be an object of a category $\X$. We define the following two categories.
	
	\noindent 
	\parbox{11cm}{\vspace{-.5cm}\begin{itemize}
		\item The \emph{ slice category over $X$} is the category $\X/X$ which has as objects arrows $f\colon Y\to X$ and  in which an arrow $h\colon f\to g$ is $h\colon Y\to Y'$ in $\X$ such that the triangle on the right commutes.
		\item  	 Dually, the \emph{ slice category under $X$} is the category $X/\X$ in which objects are arrows $f\colon X\to Y$  with domain $X$ and a morphism $h\colon  f\to g$ is an arrow of $\X$ fitting in a triangle as the one aside.
	\end{itemize}} \hfill \parbox{2cm}
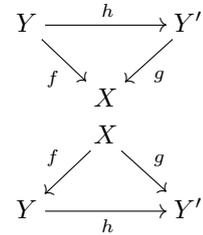
{\vspace{-0cm}\xymatrix@C=15pt@R=15pt{Y \ar[dr]_{f} \ar[rr]^{h}&& Y' \ar[dl]^{g}\\ & X} \xymatrix@C=15pt@R=15pt{&X \ar[dr]^{g}\ar[dl]_{f}\\ Y  \ar[rr]_{h}&& Y'}}
\end{definition}

\begin{remark} For every $X\in\ X$ we have forgetful functors
	\[\begin{split}
		\dom_X&\colon \X/X\to \X\\
		\functor[l]{f}{h}{g}
		&\functormapsto
		\rfunctor{\dom(f) }{h}{\dom(g)}
	\end{split}\qquad \begin{split}
		\cod_X&\colon X/\X\to \X\\
		\functor[l]{f}{h}{g}
		&\functormapsto
		\rfunctor{\cod(f) }{h}{\cod(g)}
	\end{split}\]
\end{remark}

We can realize the slice over and under an object $X\in \X$ as comma categories.

\begin{proposition}\label[proposition]{prop:commaapp}
	For every object $X$ in a category $\X$, if $\delta_X\colon \T\to \X$ is  the constant functor of value  $X$ from the category with only one object $*$, then $\X/X$ and $X/\X$ are isomorphic to, respectively,  $\comma{\id{X}}{\delta_X}$ and $\comma{\delta_X}{\id{X}}$ .
\end{proposition}

\newpage
\begin{proof} Define functors $F_1\colon \comma{\id{X}}{\delta_X}\to \X/X$ and $G_1\colon \X/X\to \comma{\id{X}}{\delta_X}$ as follows
	\[	\begin{split}
		\functor[l]{(Y, *,  f)}{(h, \id{*} )}{( Y', *,  g)}
		\functormapsto
		\rfunctor{f }{h}{g}
	\end{split} \qquad \begin{split}
		\functor[l]{f }{h}{g}
		\functormapsto
		\rfunctor{(\dom(f), *, f)}{(h, \id{*} )}{(\dom(g), *, g)}
	\end{split}\]
	Similarly, we have $F_2\colon \comma{\delta_X}{\id{X}}\to X/\X$ and $G_2\colon X/\X\to \comma{\delta_X}{\id{X}}$
	\[	\begin{split}
		\functor[l]{(*, Y, f)}{(\id{*},h )}{(*, Y', g)}
		\functormapsto
		\rfunctor{f }{h}{g}
	\end{split} \qquad \begin{split}
		\functor[l]{f }{h}{g}
		\functormapsto
		\rfunctor{(*, \cod(f), f)}{(\id{*}, h )}{(*, \cod(g), g)}
	\end{split}\]
	It is now obvious to see that $F_1,G_1$ and $F_2, G_2$ are pairs of inverses.   
\end{proof}

A straightforward application of \Cref{lim,colim} now yields the following.
\begin{corollary}\label[corollary]{cor:pbapp}
For every object $\X$, $\X/X$ has all colimits and connected limits that $\X$ has. Moreover such limits and colimits are created by $\dom_X$.

In particular, if $\X$ has pullbacks, equalizers or pushouts, then for every object $X$, the slice $\X/X$ has such limits and colimits.
\end{corollary}

\end{document}

%% file: hypergraph.tikzstyles
\tikzstyle{hedge}=[fill=white, draw=black, shape=rectangle, rounded corners=2mm, inner sep=0.2mm, outer sep=-2mm, scale=0.8, minimum height=8mm, minimum width=8mm, tikzit category=hypergraph]
\tikzstyle{hedge blue}=[hedge, fill={rgb,255: red,102; green,204; blue,255}, draw=black, shape=rectangle, tikzit category=hypergraph]
\tikzstyle{node}=[fill=black, draw=black, shape=circle, minimum size=1.5mm, inner sep=0mm, tikzit category=hypergraph]
\tikzstyle{red node}=[fill=red, draw=black, shape=circle, minimum size=1.5mm, inner sep=0mm, tikzit category=hypergraph]
\tikzstyle{node highlight}=[fill=black, draw=blue, thick, shape=circle, minimum size=1.5mm, inner sep=0mm, tikzit category=hypergraph]
\tikzstyle{red node highlight}=[fill=red, draw=blue, thick, shape=circle, minimum size=1.5mm, inner sep=0mm, tikzit category=hypergraph]
\tikzstyle{yellow hedge}=[hedge, fill=yellow, draw=black, shape=rectangle, tikzit category=hypergraph]
\tikzstyle{green hedge}=[hedge, fill=green, draw=black, shape=rectangle, tikzit category=hypergraph]
\tikzstyle{small box}=[fill=white, draw=black, shape=rectangle, minimum height=6mm, minimum width=6mm, tikzit category=string diagram]
\tikzstyle{vsmall box}=[fill=black, draw=black, shape=rectangle, minimum height=4mm, minimum width=1mm, tikzit category=string diagram, inner sep=0]
\tikzstyle{medium box}=[fill=white, draw=black, shape=rectangle, minimum height=11mm, minimum width=6mm, tikzit category=string diagram]
\tikzstyle{semilarge box}=[fill=white, draw=black, shape=rectangle, minimum height=16mm, minimum width=6mm, tikzit category=string diagram]
\tikzstyle{large box}=[fill=white, draw=black, shape=rectangle, minimum height=21mm, minimum width=6mm, tikzit category=string diagram]
\tikzstyle{black dot}=[fill=black, draw=black, shape=circle, minimum size=2mm, inner sep=0mm, tikzit category=string diagram]
\tikzstyle{white dot}=[fill=white, draw=black, shape=circle, minimum size=2mm, inner sep=0mm, tikzit category=string diagram]
\tikzstyle{red dot}=[fill=red, draw=black, shape=circle, minimum size=2mm, inner sep=0mm, tikzit category=string diagram]
\tikzstyle{wlabel}=[fill=none, draw=none, shape=rectangle, tikzit category=string diagram, font={\footnotesize}, inner sep=0pt, tikzit fill={rgb,255: red,102; green,204; blue,255}, tikzit draw={rgb,255: red,102; green,204; blue,255}, yshift=0.3mm]
\tikzstyle{BRchange}=[draw=black, shape=diamond, tikzit shape=circle, tikzit fill={rgb,255: red,96; green,0; blue,0}, diamond split part fill={black,red}, inner sep=-5mm, minimum width=2.7mm, minimum height=1.7mm]
\tikzstyle{RBchange}=[draw=black, shape=diamond, tikzit shape=circle, tikzit fill={rgb,255: red,165; green,0; blue,0}, diamond split part fill={red,black}, inner sep=0, minimum width=2.7mm, minimum height=1.7mm]
\tikzstyle{dummy}=[fill=none, draw=none, shape=circle, font={\small}, inner sep=1pt, tikzit draw=blue, tikzit fill=white]
\tikzstyle{node label}=[fill=none, draw=none, shape=rectangle, tikzit fill=cyan, tikzit draw=cyan, font={\scriptsize}, tikzit shape=circle, inner sep=0pt]
\tikzstyle{empty diag}=[fill=white, draw={rgb,255: red,165; green,165; blue,165}, shape=rectangle, minimum size=1.2 cm, dashed, thick]

\tikzstyle{dashed edge}=[-, dashed, very thick]
\tikzstyle{alt sort}=[-, dashed, dash pattern=on 2pt off 0.5pt, thick, draw=red]
\tikzstyle{diredge}=[->, >={Latex[length=1.5mm]}]
\tikzstyle{diredge highlight}=[->, >={Latex[length=1.5mm]}, draw=blue, thick]
\tikzstyle{boundary frame}=[-, draw={rgb,255: red,170; green,170; blue,255}, dashed, fill={rgb,255: red,238; green,238; blue,255}, thick, dash pattern=on 2pt off 0.5pt]
\tikzstyle{graph frame}=[-, draw={rgb,255: red,191; green,191; blue,191}, dashed, fill={rgb,255: red,238; green,238; blue,238}, thick, dash pattern=on 2pt off 0.5pt]
\tikzstyle{def sort}=[-]
\tikzstyle{component}=[-, draw=red, thick]
\tikzstyle{map edge}=[{|->}, >=latex, shorten <=0.5mm, shorten >=0.5mm]
\tikzstyle{hypergraph map edge}=[{|->}, draw=red, shorten <=1mm, shorten >=1mm]
\tikzstyle{cdedge}=[->]
\tikzstyle{big cdedge}=[->, very thick, >=latex]
\tikzstyle{pointer edge}=[->, draw=gray, thick]